\pdfoutput=1

\newcommand{\FORGET}[1]{}

\documentclass{article}
\usepackage{mya4}

\usepackage[ruled]{algorithm2e}

\usepackage{tabularx}
\usepackage{booktabs}
\usepackage{amsthm}
\usepackage{amsmath,amssymb}
\usepackage{graphicx}
\usepackage{epsfig}
\usepackage{verbatim}
\usepackage[usenames,dvipsnames,svgnames,table]{xcolor}
\usepackage{float}
\usepackage{mathptmx}
\usepackage{datetime}
\usepackage{listings}
\usepackage{subfig}
\usepackage{stmaryrd}
\usepackage{latexsym}
\usepackage{amsmath,amsfonts,amstext,amssymb,wasysym}
\usepackage{url}
\usepackage{mathpartir}
\usepackage{enumitem}
\usepackage[disable]{todonotes}
\usepackage{fancyvrb}

\usepackage{xifthen}
\usepackage{xargs}
\makeatletter
\newtheorem*{rep@theorem}{\rep@title}
\newcommand{\newreptheorem}[2]{%
	\newenvironment{rep#1}[2][]{
		\def\rep@title{Restatement of #2 \ref{##2}\ifthenelse{\isempty{##1}}{}{ (##1)}}
		\begin{rep@theorem}
	}{\end{rep@theorem}}
}
\let\newoldtheorem\newtheorem
\renewcommandx{\newtheorem}[3][2=NoValueAtAll]{
	\ifthenelse{\equal{#2}{NoValueAtAll}}{\newoldtheorem{#1}{#3}}{\newoldtheorem{#1}[#2]{#3}}
	\newreptheorem{#1}{#3}
}
\makeatother

\newtheorem{example}{Example}
\newtheorem{remark}{Remark}

\newtheorem{thm}{Theorem}
\newtheorem{lem}[thm]{Lemma}

\definecolor{dark-gray}{gray}{0}
\newcommand{\il}[1]{{\it \textcolor{gray}{;; #1}}} 
\newcommand{\km}[1]{{\bf \textcolor{red}{#1}}} 
\newcommand{\fn}[1]{\textcolor{blue}{#1}} 
\newcommand{\pr}[1]{\textcolor{blue}{#1}} 

\newfloat{listing}{tbp}{lol}
\floatname{listing}{Listing}
\floatstyle{boxed}
\restylefloat{listing}

\newif\ifpgf\pgftrue  

\lstdefinestyle{mystyle}{flexiblecolumns=true,showstringspaces=false,keepspaces=true,basewidth={0em,0em},,basicstyle=\sffamily,commentstyle=\itshape,stringstyle=\sffamily}

\definecolor{diffhilight}{rgb}{0.6, 0.81, 0.93}

\newcommand{\builtinop}[3]{\llparenthesis #1 \rrparenthesis_{#2}^{#3}}
\newcommand{\filter}{F}
\newcommand{\BNFcce}{{\bf ::=}}
\newcommand{\BNFmid}{\;\bigr\rvert\;}

\newcommand{\PROGRAM}{\mathtt{P}}
\newcommand{\FUNCTION}{\mathtt{F}}
\newcommand{\e}{\mathtt{e}}
\newcommand{\fname}{\mathtt{d}}
\newcommand{\xname}{\mathtt{x}}
\newcommand{\yname}{\mathtt{y}}

\newcommand{\oname}{\mathtt{b}}

\newcommand{\lvalue}{\ell}

\newcommand{\truevalue}{\mathtt{True}}
\newcommand{\falsevalue}{\mathtt{False}}
\newcommand{\zerovalue}{0}

\newcommand{\fvalue}{\phi}

\newcommand{\defK}{\mathtt{def}}

\newcommand{\nbrK}{\mathtt{nbr}}
\newcommand{\repK}{\mathtt{rep}}

\newcommand{\ifK}{\mathtt{if}}
\newcommand{\elseK}{\mathtt{else}}

\newcommand{\Trees}{\Theta}

\newcommand{\skiptransition}{\\[10pt]}

\newcommand{\proj}[2]{{#1}|_{#2}}

\newcommand{\envmap}[2]{#1\mapsto #2}

\newcommand{\ruleNameSize}[1]{{\scriptsize #1}}

\newcommand{\Topo}{\tau}
\newcommand{\Sens}{\Sigma}
\newcommand{\Envi}{\textit{Env}}
\newcommand{\EnviS}[2]{#1,#2}
\newcommand{\Cfg}{N}
\newcommand{\Field}{\Psi}
\newcommand{\SystS}[2]{\langle #1;#2\rangle}
\newcommand{\devset}{I}

\newcommand{\nettran}[3]{#1\xrightarrow{#2} #3}
\newcommand{\mapupdate}[2]{#1[#2]}

\newcommand{\wfn}[1]{\textit{WFE}(#1)}

\newcommand{\type}{\textit{T}}
\newcommand{\ltype}{\textit{L}}
\newcommand{\ftype}{\textit{F}}
\newcommand{\rtype}{\textit{R}}
\newcommand{\stype}{\textit{S}}
\newcommand{\builtintype}{\textit{B}}

\newcommand{\btype}{\mathtt{bool}}
\newcommand{\ntype}{\mathtt{num}}

\newcommand{\surfaceTyping}[3]{
  \begin{array}{l@{\;}c}
    \stackrel{~}{{\tiny \textrm{[#1]}}} & #2 \\ \hline 
    \multicolumn{2}{c}{#3}
  \end{array}
}
\newcommand{\nullsurfaceTyping}[2]{
  \surfaceTyping{#1}{}{#2}
}

\newcommand{\typeofNAME}{\OStypEnv}
\newcommand{\typeof}[1]{\typeofNAME(#1)}

\newcommand{\domofNAME}{\textbf{dom}}
\newcommand{\domof}[1]{\domofNAME(#1)}

\newcommand{\anyvalue}{\mathtt{v}}
\newcommand{\anyvaluealt}{\mathtt{u}}

\newcommand{\main}{\mathtt{main}}

\newcommand{\lengthOf}[1]{\textit{length}(#1)}

\newcommand{\noetherianOf}[1]{\noetherianOf(#1)}

\newcommand{\deviceIdSet}{\textbf{D}}

\newcommand{\netopsemRule}[3]{\surfaceTyping{#1}{#2}{#3}}

\newcommand{\vtree}{\theta}

\newcommand{\senstate}{\sigma}

\newcommand{\act}{\textit{act}}
\newcommand{\envact}{\textit{env}}

\newcommand{\actDevInOf}[1]{\envact}

\newcommand{\actDevOutOf}[1]{\envact}

\newcommand{\actEdgeInOf}[1]{\envact}

\newcommand{\actEdgeOutOf}[1]{\envact}

\newcommand{\actSnsUpdOf}[1]{\envact}

\newcommand{\emain}{\e_{\main}}

\newcommand{\bsopsem}[4]{#1;#2\vdash #3\Downarrow #4}
\newcommand{\deviceId}{\delta}

\newcommand{\vroot}{\mathbf{\rho}}
\newcommand{\vrootOf}[1]{\vroot(#1)}
\newcommand{\substitution}[2]{#1:=#2}
\newcommand{\applySubstitution}[2]{#1[#2]}

\newcommand{\HFC}{HFC}

\newcommand{\ofname}{\mathtt{g}}

\newcommand{\stvar}{s}
\newcommand{\rtvar}{r}
\newcommand{\ltvar}{l}
\newcommand{\tvar}{t}
\newcommand{\ftypeOf}[1]{\mathtt{field}(#1)}

\newcommand{\ltypescheme}{\textit{LS}}

\newcommand{\emptyseq}{\bullet}

\newcommand{\FVname}{\textbf{FV}}
\newcommand{\FV}[1]{\FVname(#1)}

\newcommand{\FTVname}{\textbf{FTV}}
\newcommand{\FTV}[1]{\FTVname(#1)}

\newcommand{\auxNAME}{\textit{aux}}
\newcommand{\aux}[1]{\auxNAME(#1)}

\newcommand{\bodyNAME}{\textit{body}}
\newcommand{\body}[1]{\bodyNAME(#1)}

\newcommand{\argsNAME}{\textit{args}}
\newcommand{\args}[1]{\argsNAME(#1)}

\newcommand{\TtypEnv}{\mathcal{A}}
\newcommand{\OStypEnv}{\mathcal{B}}
\newcommand{\LTStypEnv}{\mathcal{D}}

\newcommand{\expTypJud}[4]{#1 ; #2 \vdash #3 : #4}
\newcommand{\funTypJud}[3]{#1 \vdash #2 : #3}
\newcommand{\proTypJud}[2]{\vdash #1 : #2}

\newcommand{\mkvtree}[3]{#2 \langle #3 \rangle}

\newcommand{\piB}[1]{\pi^{#1}}
\newcommand{\piBof}[2]{\piB{#1}(#2)}
\newcommand{\piI}[1]{\pi_{#1}}
\newcommand{\piIof}[2]{\piI{#1}(#2)}

\newcommand{\fstK}{\mathtt{fst}}
\newcommand{\sndK}{\mathtt{snd}}
\newcommand{\headK}{\mathtt{head}}
\newcommand{\tailK}{\mathtt{tail}}

\newcommand{\funvalue}{\mathtt{f}}

\newcommand{\muxK}{\mathtt{mux}}

\newcommand{\dc}{\mathtt{c}}
\newcommand{\dcOf}[2]{#1(#2)}

\newcommand{\pairK}{\mathtt{pair}}
\newcommand{\PairK}{\mathtt{Pair}}

\newcommand{\NullK}{\mathtt{Null}}
\newcommand{\ConsK}{\mathtt{Cons}}
\newcommand{\listK}{\mathtt{list}}

\newcommand{\pairltypeOf}[2]{\pairK(#1,#2)}
\newcommand{\listltypeOf}[1]{\listK(#1)}

\newcommand{\WFVT}[3]{\textit{WFVT}(#1,#2,#3)}
\newcommand{\WFVTE}[3]{\textit{WFVTE}(#1,#2,#3)}

\definecolor{light-gray}{gray}{0.7}

\newcommand{\selfK}{\mathtt{uid}}
\newcommand{\minHoodK}{\texttt{min-hood}}
\newcommand{\pickHoodK}{\texttt{pick-hood}}
\newcommand{\snsNumK}{\texttt{sns-num}}
\newcommand{\snsFunK}{\texttt{sns-fun}}

\newcommand{\toSymK}{\mathrm{\texttt{=>}}}

\newcommand{\shift}[1]{\textbf{shift}(#1)}
\newcommand{\builtindenot}[2]{\mathcal{#1}\llbracket #2 \rrbracket}
\newcommand{\predevices}[1]{{#1}^{{}^-}\!\!\!}
\newcommand{\nbrdevice}[2]{{#1}^{#2}}
\newcommand{\repdevice}[1]{{#1}^-}
\newcommand{\decay}[0]{t_{\mathtt{d}}}

\newcommand{\PathS}[0]{\mathbf{P}}
\newcommand{\pathS}[0]{P}
\newcommand{\VarS}[0]{X}
\newcommand{\EventS}[0]{\mathbf{E}}
\newcommand{\eventS}[0]{E}
\newcommand{\eventId}[0]{\epsilon}
\newcommand{\timeS}[0]{t}
\newcommand{\posS}[0]{p}

\newcommand{\devF}[1]{\deviceId_{#1}}
\newcommand{\timeF}[1]{\timeS_{#1}}

\newcommand{\setFS}[0]{\textbf{F}}

\newcommand\pto{\mathrel{\ooalign{\hfil$\mapstochar$\hfil\cr$\to$\cr}}}

\newcommand{\neighbour}[2]{\textit{neigh}(#1,#2)}

\newcommand{\denottype}[1]{\mathcal{T}\llbracket{#1}\rrbracket}
\newcommand{\denotval}[1]{\mathcal{V}\llbracket {#1} \rrbracket}
\newcommand{\denotexp}[3]{\mathcal{E}\llbracket {#1} \rrbracket_{#2}^{#3}}
\newcommand{\denotf}[2]{\lambda #1.#2}

\newcommand{\dvalue}[0]{\Phi}

\begin{document}

\title{A Higher-order Calculus of Computational Fields\footnote{Author's addresses: M. Viroli  {and} D. Pianini, DISI, University of Bologna, Italy; G.Audrito {and} F. Damiani, Dipartimento di Informatica, University of Torino, Italy; J. Beal, Raytheon BBN Technologies, USA.}}
\author{Mirko Viroli, Giorgio Audrito, Ferruccio Damiani, Danilo Pianini, Jacob Beal}

\maketitle

\begin{abstract}
The complexity of large-scale distributed systems, particularly when deployed in physical space, calls for new mechanisms to address composability and reusability of collective adaptive behaviour.
Computational fields have been proposed as an effective abstraction to fill the gap between the macro-level of such systems (specifying a system's collective behaviour) and the micro-level (individual devices' actions of computation and interaction to implement that collective specification), thereby providing a basis to better facilitate the engineering of collective APIs and complex systems at higher levels of abstraction.
This paper proposes a full formal foundation for field computations, in terms of a core (higher-order) calculus of computational fields containing a few key syntactic constructs, and equipped with typing, denotational and operational semantics.
Critically, this allows formal establishment of a link between the micro- and macro-levels of collective adaptive systems, by a result of full abstraction and adequacy for the (aggregate) denotational semantics with respect to the (per-device) operational semantics.
\end{abstract}

\section{Introduction}\label{sec-intro}

The increasing availability of computational devices of every sort, spread throughout our living and working environments, is transforming the challenges in construction of complex software applications, particularly if we wish them to take full opportunity of this computational infrastructure.
Large scale, heterogeneity of communication infrastructure, need for resilience to unpredictable changes, openness to on-the-fly adoption of new code and behaviour, and pervasive collectiveness in sensing, planning and actuation: all these features will soon be the norm in a great variety of scenarios of pervasive computing, the Internet-of-Things, cyber-physical systems, etc.
Currently, however, it is extremely difficult to engineer collective applications of this kind, mainly due to the lack of computational frameworks well suited to deal with this level of complexity in application services.
Most specifically, there is need to provide mechanisms by which reusability and composability of components for collective adaptive behaviour becomes natural and implicit, such that they can support the construction of layered APIs with formal behaviour guarantees, sufficient to readily enable the creation of complex applications.

Aggregate computing \cite{BPV-COMPUTER2015} is a paradigm aiming to address this problem, by means of the notion of \emph{computational field} \cite{tota} (or simply \emph{field}): this is a global, distributed map from computational devices to computational objects (data values of any sort, including higher-order objects such as functions and processes). 
Computing with fields means computing such global structures, and defining a reusable block of behaviour means to define a reusable computation from fields to fields: this functional view holds at any level of abstraction, from low-level mechanisms up to whole applications, which ultimately work by getting input fields from sensors and process them to produce output fields to actuators. 

The \emph{field calculus} \cite{forte2015,DVB-SCP2016} is a tiny functional language providing basic constructs to work with fields, whose operational semantics can act as blueprint for actual implementations where myriad devices interact via proximity-based broadcasts.
Field calculus provides a unifying approach to understanding and analysing the wide range of approaches to distributed systems engineering that make use of computational fields~\cite{SpatialIGI2013}.
Recent works have also adopted this field calculus as a \emph{lingua franca} to investigate formal properties of resiliency to environment changes \cite{DiGamma,VBDP-SASO2015}, and to device distribution \cite{BVPD-SASO16}.

In this paper we propose a full foundation for aggregate computing and field computations.
We introduce syntax, typing, denotational semantics, properties, and operational semantics of a higher-order version of field calculus, where functions---and hence, computational behaviour---can be seen as objects amenable to manipulation just like any other data structure, and can hence be injected at run-time during system operation through sensors, spread around, and be executed by all (or some) devices, which then coordinate on the collective computation of a new service.

A key insight and technical result of this paper is that the notoriously difficult problem of reconciling local and global behaviour in a complex adaptive system \cite{SpatialIGI2013} can be connected to a well-known problem in programming languages: correspondence between denotational and operational semantics.
On the one hand, in field calculus, denotational semantics characterises computations in terms of their global effect across space (available devices) and time (device computation events)---i.e., the macro level.
On the other hand, operational semantics gives a transition system dictating each device's individual and local computing/interactive behaviour---i.e., the micro level.
Correspondence between the two, formally proved in this paper via \emph{adequacy} and  \emph{full abstraction} (c.f., \cite{Curien:FullAbstraction,Stoughton:FullyAbstract}), thus provides a formal micro-macro connection: one designs a system considering the denotational semantics of programming constructs, and an underlying platform running the distributed interpreter defined by the operational semantics guarantees a consistent execution.
This is a significant step towards effective methods for the engineering of self-adaptive systems, achieved thanks to the standard theory and framework of programming languages.

The remainder of this paper is organised as follows: 
Section~\ref{sec-concepts} reviews related works and the background for this work, describing the key elements of higher-order field computations.
Section~\ref{sec-calculus-syntax-and-typing} defines syntax and typing of the proposed calculus, which is then provided with two semantics:
Section~\ref{sec-denotational-semantics} defines denotational semantics, while
Section~\ref{sec-calculus-operational-semantics} defines operational semantics.
Section~\ref{sec-properties} then discusses and proves properties of these semantics, including adequacy and  full abstraction, and
Section~\ref{sec-examples} gives examples showing the expressive power of the proposed calculus for engineering collective adaptive systems.
Finally Section~\ref{sec-conclusion} concludes and discusses future directions.\footnote{This paper is an extended version of the work in \cite{forte2015}, adding: a reduced (yet more expressive) and reworked set of constructs, a type system, denotational semantics, and adequacy and full abstraction results.}

\section{Related Work and Background}\label{sec-concepts}

The work on field calculus presented in this paper builds on a sizable body of prior work.
In this section, we begin with a general review of work on the programming of aggregates.
Following this, we aim to provide the reader with examples and intuition that can aid in understanding the more formal presentation in subsequent sections,
presenting a conceptual introduction to programming with fields and the extension of these concepts to first-class functions over fields.

\subsection{Macro-programming and the aggregation problem}

One of the key challenges in software engineering for collective adaptive systems is that such systems frequently comprise a potentially high number of devices (or agents) that need to interact locally (e.g., interacting by proximity as in wireless sensor networks), either of necessity or for the sake of efficiency.
Such systems need to carry on their collective tasks cooperatively, and to leverage such cooperation in order to adapt to unexpected contingencies such as device failures, loss of messages, changes to inputs, modification of network topology, etc.
Engineering locally-communicating collective systems has long been a subject of interest in a wide variety of fields, from biology to robotics, from networking to high-performance computing, and many more.
%

Despite the diversity of fields involved, however, a uniting has been the search for appropriate mechanisms, models, languages and tools to organise cooperative computations as carried out by a potentially vast aggregation of devices spread over space.

A general survey of work in this area may be found in~\cite{SpatialIGI2013}, which we summarise and complement here.
Across the multitude of approaches that have been developed in the past, a number of common themes have emerged,
and prior approaches may generally be understood as falling into one of several clusters in alignment with these themes:
\begin{itemize}

\item \emph{Foundational approaches to group interaction:} These approaches present mathematically concise foundations for capturing the interaction of groups in complex environments, most often by extending the archetypal process algebra $\pi$-calculus, which originally models flat compositions of processes. 
Such approaches include various models of environment structure (from "ambients" to 3D abstractions) \cite{DBLP:conf/cie/CardelliG10,ambients,Milner200660}, shared-space abstractions by which multiple processes can interact in a decoupled way \cite{klaim,VCMZ-TAAS2011}, and attribute-based models declaratively specifying the target of communication so as to dynamically create ensembles \cite{SCEL}.

\item \emph{Device abstraction languages:} These approaches allow a programmer to focus on cooperation and adaptation by making the details of device interactions implicit. For instance, TOTA~\cite{tota} allows one to program tuples with reaction and diffusion rules, while in the SAPERE approach~\cite{VPMSZ-SCP2015} such rules are embedded in space and apply semantically, and the $\sigma\tau$-Linda model~\cite{spatialcoord-coord2012} manipulates tuples over space and time.
Other examples include MPI~\cite{MPI2}, which declaratively expresses topologies of processes in supercomputing applications, NetLogo~\cite{sklar2007netlogo}, which provides abstract means to interact with neighbours following the cellular automata style, and Hood~\cite{hood}, which implicitly shares values with neighbours; 

\item \emph{Pattern languages:} These approaches provide adaptive means for composing geometric and/or topological constructions, though with little focus on computational capability.  For example, the Origami Shape Language~\cite{nagpalphd} allows the programmer to imperatively specify geometric folds that are compiled into processes identifying regions of space, Growing Point Language~\cite{coorephd} provides means to describe topologies in terms of a ``botanical'' metaphor with growing points and tropisms, ASCAPE~\cite{inchiosa2002overcoming} supports agent communication by means of topological abstractions and a rule language, and the catalogue of self-organisation patterns in \cite{FDMVA-NACO2013} organises a variety of mechanisms from low-level primitives to complex self-organization patterns.

\item \emph{Information movement languages:} These are the complement of pattern languages, providing means for summarising information obtained from across space-time regions of the environment and streaming these summaries to other regions, but little control over the patterning of that computation.
Examples include TinyDB~\cite{tinydb} viewing a wireless sensor network as a database, Regiment~\cite{regiment} using a functional language to be compiled into protocols of device-to-device interaction, and the agent communication language KQML~\cite{Finin94kqml}.

\item \emph{Spatial computing languages:} These provide flexible mechanisms and abstractions to explicitly consider spatial aspects of computation, avoiding the limiting constraints of the other categories. For example, Proto~\cite{proto06a} is a Lisp-like functional language and simulator for programming wireless sensor networks with the notion of computational fields, and MGS~\cite{GiavittoMGS02} is a rule-based language for computation of and on top of topological complexes.
\end{itemize}

Overall, the successes and failures of these language suggest, as observed in~\cite{BPV-COMPUTER2015}, 
that adaptive mechanisms are best arranged to be implicit by default,
that composition of aggregate-level modules and subsystems must be simple, transparent, and result in highly predictable behaviours, 
and that large-scale collective adaptive systems typically require a mixture of coordination mechanisms to be deployed at different places, times, and scales.

\subsection{Computing with fields}
\label{s:conceptual:fields}

At the core of the approach we present in this paper is a shift from individual devices computing single values, to whole networks computing \emph{fields}, where a field is a collective structure that maps each device in some portion of the network to locally computed values over time.
Accordingly, instead of considering computation as a process of manipulating input events to produce output events, computing with fields means to take fields as inputs and produce fields as outputs.

This change of focus has a deep impact when it comes to the engineering of complex applications for large networks of devices, in which it is important that the identity and position of individual devices should not exert a significant influence on the operation of the system as a whole.
Applying the field approach to building such systems, one can create reusable distributed algorithms and define functions (from fields to fields) as building blocks, structure such building blocks into libraries of increasing complexity, and compose them to create whole application services~\cite{BPV-COMPUTER2015}.

For example, assume that one is able to define the following three functions:
\begin{itemize}
 \item \texttt{\fn{distance-to}(source)}: This function takes an indicator field \texttt{source} of Boolean values, holding true at a set of devices considered as \emph{sources}, and yields a field of real values, estimating shortest distance from each device to the closest source (if each device is assumed capable of locally estimating distance to close neighbours, long-range distance estimates can be computed transitively).
 
 \item \texttt{\fn{converge-sum}(potential,val)}: This function takes a field \texttt{potential} of real values, and a field \texttt{val} of numeric values, and it accumulates all values of \texttt{val} downward along the \texttt{potential}, summing them as they reach common devices. If the trajectory down \texttt{potential} always leads to a single global minimum, then the trajectories form a spanning tree with the minimum at its root, and in the resulting field the root holds the sum of all values of \texttt{val}.
 
 \item \texttt{\fn{low-pass}(alpha,val)}: This function takes a field \texttt{val} of real values, and at each device implements an exponential filter with blending constant \texttt{alpha}, thus acting as a low-pass filter smoothing rapid changes in the input \texttt{val} at each device.
 
\end{itemize}

Now consider an example of an application deployed into a museum, whose docents monitor their efficacy in part by tracking the number of patrons nearby while they are working.  
This application can be implemented by a simple function, taking as input Boolean fields indicating docents and patrons, and whose body is defined by purely-functional composition of the three blocks above, written e.g. in the following way:
\begin{Verbatim}[fontsize=\fontsize{8pt}{9pt}, frame=single, commandchars=\\\{\}, codes={\catcode`$=3\catcode`^=7\catcode`_=8}]
\km{def} \fn{track-count}(docent, patron) \{
   \fn{low-pass}( 0.5, \fn{converge-sum}( \fn{distance-to}(docent), \fn{mux}(patron,1,0)))
\}
\end{Verbatim}
in which the function \texttt{mux} acts as a simple multiplexer at each device, transforming true values to $1$ and false values to $0$.
This function creates a field of estimated distances out of each \texttt{docent}, and uses it as potential field for counting the number of \texttt{patron}s nearby, with the low-pass filter smoothing the result so as to deal with rapid fluctuations in the estimate that can be caused by device mobility.

As the aim of this paper is to clarify syntax and semantics of the field-based computational model, this example should already clarify the goal of compositionally stacking increasingly complex distributed algorithms, up to a point in which the focus on individual agent behaviour completely vanishes.
This can be taken even further by proving that the ``building block'' algorithms satisfy certain properties preserved by functional compositions, such as self-stabilisation \cite{VBDP-SASO2015} or consistency with a continuum model \cite{BVPD-SASO16}, thus implying the same properties hold for applications composed using those building blocks \cite{BPV-COMPUTER2015}.

\subsection{Higher-order fields and restriction}

The calculus that we present in this paper is a higher-order extension of the work in \cite{DVB-SCP2016} to include embedded first-class functions, 
with the primary goal of allowing field computations to handle functions just like any other value.
This extension hence provides a number of advantages:
\begin{itemize}
\item Functions can take functions as arguments and return a function as result (higher-order functions). This is key to defining highly reusable building block functions, which can then be fully parameterised with various functional strategies.
\item Functions can be created ``on the fly'' (anonymous functions).  Among other applications, such functions can be passed into a system from the external environment, as a fields of functions considered as input coming a sensor modelling humans adding new code into a device while the system is operating.
\item Functions can be moved between devices in the same way our calculus allows values to move, which allows one to express complex patterns of code deployment across space and time.
\item Similarly, in our calculus a function value is naturally actually a field of functions (possibly created on the fly and then shared by movement to all devices), and can be used as  an ``aggregate function'' operating over a whole spatial domain.
\end{itemize}

The last feature is critical, and its implications are further illustrated in Figure~\ref{f:fn}.
\begin{figure}[t]
\centering
\subfloat[A field of functions hosted on individual devices]{\includegraphics[width=0.45\textwidth]{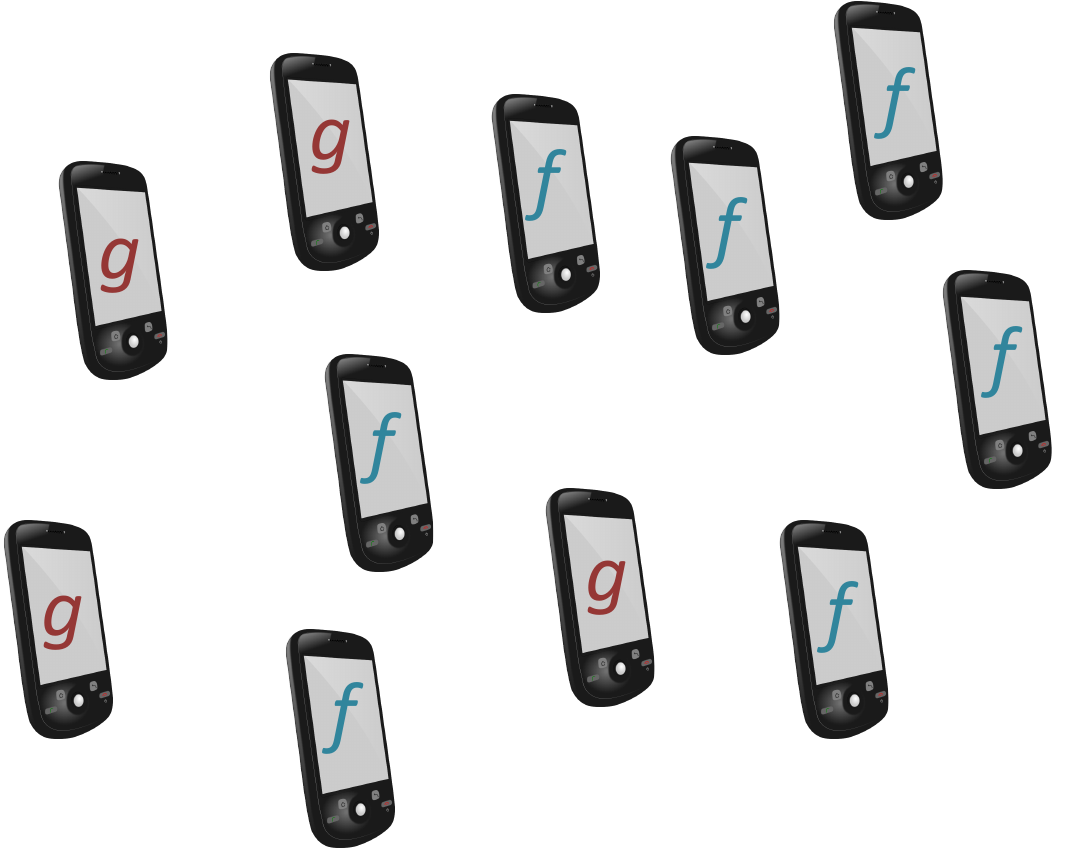}\label{f:fnif}}\hfill
\subfloat[The same field interpreted as an aggregate function defined piecewise over aggregates of devices]{\includegraphics[width=0.45\textwidth]{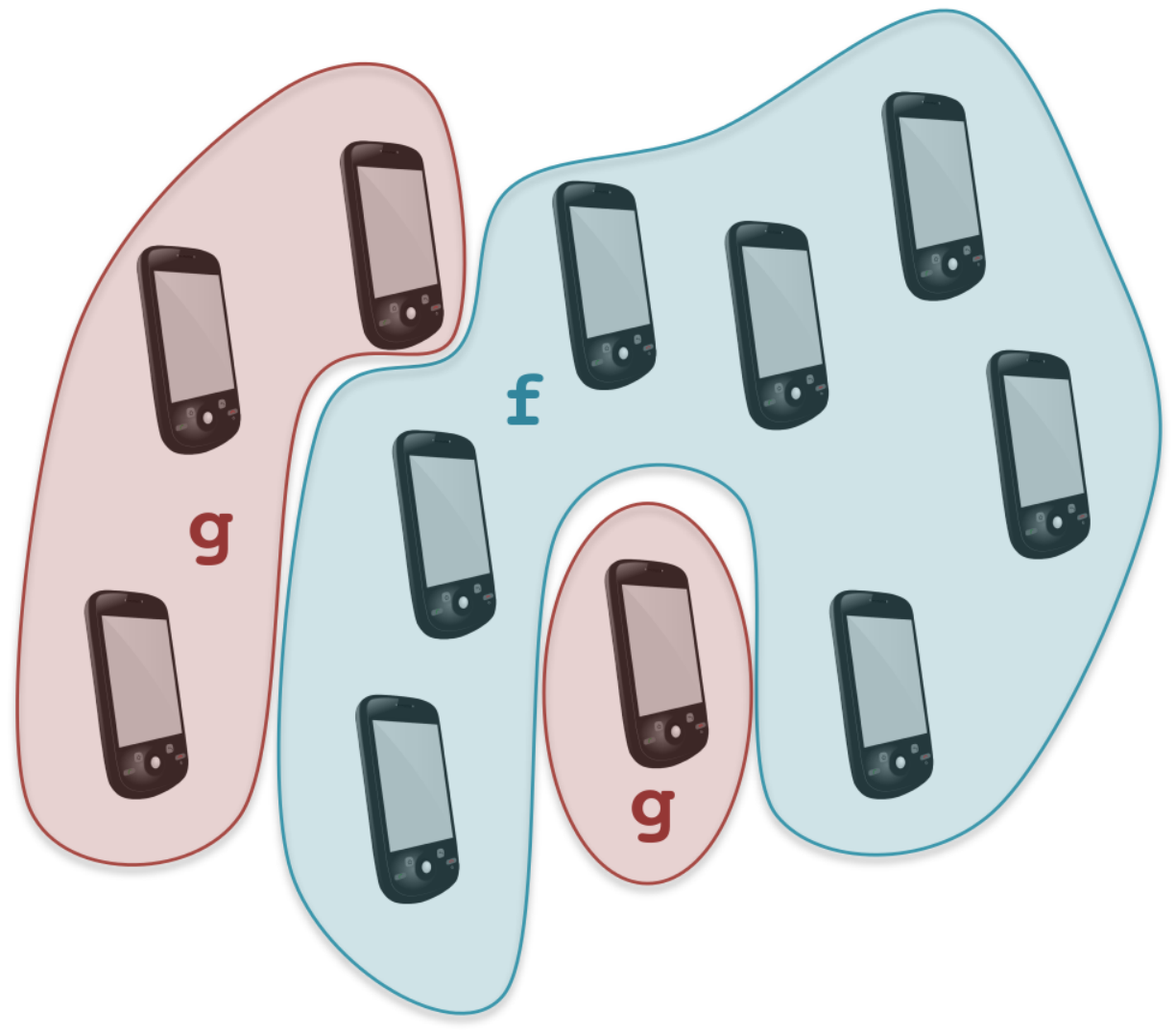}\label{f:iffn}}
\caption{Field calculus functions are evaluated over a domain of
  devices, e.g., in (a) the network reports function {\tt f} in some devices, and function {\tt g} in others.
  When such a field of functions is used in function application it gets reinterpreted as an aggregate function, as shown in (b).
  Namely, the domain is actually partitioned into two subspaces that do not interact with each other: in one domain $D_f$, function {\tt f} is applied to its arguments, seen as fields restricted to $D_f$; in the complementary domain $D_g$, function {\tt g} is applied to its arguments, seen as fields restricted to $D_g$.
}\label{f:fn}
\end{figure}
In considering fields of function values, we take the elegant approach in which making a function call acts as a branch, with each function in the range of the field
applied only on the subspace of devices that hold that function.
When the field of functions is constant, this implicit branch reduces to be precisely equivalent to a standard function call.  
This means that we can view ordinary evaluation of a function $f$ as equivalent to creating
a function-valued field with a constant value $f$, then making a function
call applying that field to its argument fields.
This elegant transformation is one of the key insight of this paper, enabling
first-class functions to be implemented with relatively minimal complexity.

This interpretation of function calls as branching also turns out to be very flexible, as it generally allows the dynamic partitioning of the network into subspaces, each executing a different subprogram, 
and such programs can be even dynamically injected into some device and then be diffused around, as will be exemplified later in this paper.

\section{The Higher-order Field Calculus: Syntax and Typing}\label{sec-calculus-syntax-and-typing}

This section presents the \emph{higher-order field calculus (HFC)}, as extended and refined from~\cite{forte2015},\footnote{The version of the HFC presented in this paper is a minor refinement of the version of HFC presented in~\cite{forte2015}. The new version adopts a different syntax
 (in~\cite{forte2015} a Lisp-like syntax was used), is parametric in the set of the modeled data values (in~\cite{forte2015} Booleans, numbers, and pairs were explicitly modeled), and
does not have a syntactic construct devoted to domain restriction (the $\ifK$-construct in~\cite{forte2015}) since it can be encoded by means of an aggregate function call---the existence of such an encoding is one of the novel contributions of this paper.}
a tiny functional calculus capturing the essential elements of field
computations, much as $\lambda$-calculus~\cite{LambdaCalculus}
captures the essence of functional computation and FJ~\cite{FJ} the
essence of object-oriented programming.
For the key importance of higher-order features, especially in the toolchain under construction \cite{Protelis15}, in the following we sometimes refer to this calculus as simply the \emph{field calculus}, especially when there is no confusion with the work in \cite{DVB-SCP2016}, which did not include higher-order features.

The defining property of fields is that they allow us to see a computation from two different viewpoints.
On the one hand, from the standard ``local'' viewpoint, a computation is seen as occurring in a single device, 
and it hence manipulates data values (e.g., numbers) and communicates such data values with other devices in order to enable coordination.
On the other hand, from the ``aggregate'' (or ``global'') viewpoint
\cite{VDB-FOCLASA-CIC2013}, computation is seen as occurring on the
overall network of interconnected devices: the data abstraction
manipulated is hence a whole distributed {\em field}, a dynamically
evolving data structure mapping devices $\deviceId$ in a domain $D$ (the whole or a subset of it) to associated data values $\anyvalue$ in range $V$.\footnote{%
Note that this viewpoint can embrace both discrete domains (e.g., networks of devices) and continuous domains (e.g., the environment where computation acts upon); in this paper, however, we will restrict ourselves to treating fields with discrete domains.  For a discussion of the relationship between computation on discrete and continuous domains, see~\cite{BVPD-SASO16}}
Field computations then take fields as input (e.g., from sensors) and
produce new fields as outputs (e.g., to feed actuators).
Both input and output may, of course, change over time 
(e.g., as inputs change or the computation progresses).
For example, the input of a computation might be a field of
temperatures, as perceived by sensors at each device in the network,
and its output might be a Boolean field that maps to $\truevalue$ where and when
temperature is greater than 25$^{\circ}$C, and to $\falsevalue$ elsewhere.
In a more involved example, the output might map to $\truevalue$ only those devices whose distance is less than 50 meters from some device where temperature was greater than 25$^{\circ}$C for the last 60 seconds. The operational semantics described in Section~\ref{sec-calculus-operational-semantics} will then describe how such global-level behaviour can be turned into a fully-distributed computation.

\subsection{Syntax}\label{sec-calculus-syntax}

Figure~\ref{fig:source:syntax} presents the syntax of the proposed calculus.
Following~\cite{FJ}, the overbar notation denotes
metavariables over sequences and the empty sequence is denoted  by $\emptyseq$. E.g., for expressions, we let $\overline
\e$ range over sequences of expressions, written
$\e_1,\,\e_2,\,\ldots\,\e_n$ $(n\ge 0)$.

\begin{figure}[!t]
\centering
\centerline{\framebox[\textwidth]{$
\begin{array}{@{\hspace{-0.1cm}}lcl@{\hspace{-4mm}}r}
   \PROGRAM & \BNFcce & \overline{\FUNCTION}  \; \e
                                                                                                 &   {\footnotesize \mbox{program}}    
\\[3pt]
    \FUNCTION & \BNFcce &  \defK \; \fname (\overline{\xname}) \; \{\e\}
 &   {\footnotesize \mbox{function declaration}}
\\[3pt]
\e & \BNFcce &  \xname 
  \; \BNFmid \;  \fvalue 
   \; \BNFmid \; \dcOf{\dc}{\overline{\e}}
\; \BNFmid \; \oname  \; \BNFmid \; \fname 
 \; \BNFmid \; (\overline{\xname}) \; \toSymK \; \e
   \; \BNFmid \; \e(\overline{\e})
  \;\BNFmid \; \repK(\e)\{(\xname) \; \toSymK \; \e\}
     \; \BNFmid \; \nbrK\{\e\}
                                                                                             &   {\footnotesize \mbox{expression}} 
\\
\end{array}
$}
}
\caption{Syntax.}
\label{fig:source:syntax}\vspace{-10pt}
\end{figure}

A program  $\PROGRAM$ consists of a sequence of function declarations and of a main expression $\e$.
A function declaration  $\FUNCTION$ defines a (possibly recursive) function. It consists of the name of the function $\fname$, of $n\ge 0$ variable names $\overline{\xname}$ representing the formal parameters, and of
an expression $\e$ representing the body of the function. 

Expressions $\e$ are the main entities of the calculus; terminologically, an expression can be: a variable $\xname$, used as function formal parameter;
a \emph{neighbouring field value} $\fvalue$;
a \emph{data-expression}  $\dcOf{\dc}{\overline{\e}}$ (where $\dc$ is a \emph{data-constructor} with arity $m\ge 0$ and the $m$ expressions  $\overline{\e}$ are its arguments); a \emph{built-in function} name $\oname$; a \emph{declared function} name $\fname$; an \emph{anonymous function} $(\overline{\xname}) \; \toSymK \; \e$ (where $\overline{\xname}$ are the formal parameters and $\e$ is the body);
a \emph{function call} $\e(\overline{\e})$;
a \emph{$\repK$-expression} $\repK(\e_0)\{(\xname) \; \toSymK \; \e_1\}$, modelling time evolution;
or an \emph{$\nbrK$-expression} $\nbrK\{\e\}$, modelling device-to-neighbourhood interaction.
It should be noted that for the purpose of defining a foundational calculus, data-expressions and built-in representing purely functional operators could be dropped: we have have decided to include them since they simplify using the calculus to formalise non-trivial examples. 
Let the set of free variables in an expression $\e$ be denoted by $\FV{\e}$, say that an expression $\e$ is \emph{closed} if $\FV{\e}=\emptyseq$, and assume the main expression of any program must be closed. 
In section \ref{sec-calculus-informal} we informally describe the meaning of these constructs, before the full formal treatment of denotational and operational semantics will given in the remainder of the paper.

\subsection{Values}\label{sec-calculus-values}

A key aspect of the calculus is that \emph{an expression always models a whole field}.
Let the denumerable set of \emph{device identifiers} $\deviceId$  (which are unique numbers) be $\deviceIdSet$.
A \emph{firing event} $\eventId$ (or simply event) is a point in space-time where a device $\deviceId$ ``fires'', namely, evaluates the main expression of the program.
The outcome of the evaluation of a closed expression $\e$ at a given event gives a \emph{value}, so we can select a (space-time) domain $D_\e$ and collect the values obtained in all events in that domain $D_\e$, and by doing so form a field over $D_\e$, which then represent a time-varying distributed structure.

The syntax of values is given in Figure~\ref{fig:values}.
A value can be either a \emph{local value} $\lvalue$ or a \emph{neighbouring field value} $\fvalue$.
At a given event where a device $\deviceId$ fires, a local value represents an atomic data produced by $\deviceId$, whereas a neighbouring field value is a map associating a local value $\lvalue$ to each neighbour of $\deviceId$. Neighbouring field values are generally used to describe the outcome of some form of device-to-neighbour interaction as described below (sharing of values as of construct $\nbrK$, or sensing of local environment e.g., to estimate distances to neighbours).
While neighbouring field values cannot be denoted in source programs, but only appear dynamically during computations, a local value $\lvalue$ can be denoted (it is in fact parts of the syntax) to represent a constant-valued field that maps each event to that value (e.g., value $1$ represents a field always mapping each device to $1$).
Given that the calculus is higher-order, a local value can be:

\begin{figure}[!t]
\centering
\centerline{\framebox[\textwidth]{$
\begin{array}{ccl@{\hspace{7.2cm}}r}
        \anyvalue & \BNFcce & \lvalue   \; \BNFmid \;  \fvalue   &   {\footnotesize\mbox{value}} \\[3pt]
        \fvalue & \BNFcce & \envmap{\overline{\deviceId}}{\overline{\lvalue}}  &   {\footnotesize\mbox{neighbouring field value}} \\[3pt]
        \lvalue & \BNFcce &  \funvalue  \; \BNFmid \;     \dcOf{\dc}{\overline{\lvalue}}  &   {\footnotesize\mbox{local value}} \\[3pt]
        \funvalue & \BNFcce &  {\oname  \; \BNFmid \; \fname \; \BNFmid \; (\overline{\xname}) \; \toSymK \; \e} \qquad   &   {\footnotesize\mbox{function value}} \\
\end{array}
$}}
\caption{Values, neighbouring field values, local values, and function values for \HFC{}.}
\label{fig:values}\vspace{-10pt}
\end{figure}

\begin{itemize}
\item
either a \emph{data value} $\dcOf{\dc}{\lvalue_1,\cdots,\lvalue_m}$, consisting of a  data-constructor $\dc$ of arity $m\ge 0$ and $m\ge 0$ local value arguments---data values,  simply written $\dc$ when $m=0$, can be Booleans $\truevalue$ and $\falsevalue$, numbers, strings, or structured values like pairs (e.g., $\dcOf{\PairK}{3,\dcOf{\PairK}{\falsevalue,5}}$) or lists (e.g., $\dcOf{\ConsK}{2,\dcOf{\ConsK}{4,\NullK}}$); 
\item
or a \emph{function value} $\funvalue$, consisting of either a built-in operator $\oname$, a declared function $\fname$, or a closed anonymous function expression $(\overline{\xname}) \; \toSymK \; \e$.
\end{itemize}

\subsection{Informal semantics}\label{sec-calculus-informal}

While neighbouring field values, data-expressions and functions (built-in, declared function, and anonymous), trivially result in a constant field, the last three kinds of expression ($\nbrK$-expressions, $\repK$-expressions, and function calls) represent the core field manipulation constructs provided by the calculus. 
Their value at each event may depend on both the particular device that is evaluating it, and also on the last event of its neighbours.
\begin{enumerate}
\item
{\em Time evolution:} $\repK (\e_0)\{(\xname) \; \toSymK \; \e\}$ is a 
  ``repeat'' construct for dynamically changing fields, assuming
  device $\deviceId$ evaluates the application of its anonymous function $(\xname) \; \toSymK \; \e$ repeatedly in
  asynchronous rounds. At the first round $(\xname) \; \toSymK \; \e$ is applied to the value 
  of $\e_0$ at $\deviceId$,
  then at each step $(\xname) \; \toSymK \; \e$ is applied to the value obtained at previous step.
   For instance, $\repK(0)\{(\xname)\; \toSymK \;+(\xname,1))\}$ counts how
  many rounds each device has computed.
\item 
{\em Neighbouring field construction:} $\nbrK \{\e\}$ models
  device-to-neighbour interaction, by 
  returning a field of neighbouring field values: each device is associated to value $\fvalue$, which in turn maps any neighbour $\deviceId$ in the domain $D_\e$ to its most recent available value of $\e$
  (e.g., obtained via periodic broadcast, as discussed in the operational semantics).
  Such neighbouring field values can then be manipulated
  and summarised with built-in operators. For instance, $\minHoodK(\nbrK\{\e\})$ maps each device to the minimum value
  of $\e$ across its neighbourhood.
\item
{\em Function  call:} $\e(\e_1,\ldots,\e_n)$, where $n\ge 0$ and  $\e$ evaluates to a field of function values. If the field is not constant, the application is evaluated separately in each \emph{cluster}, i.e., subdomain of events where $\e$ evaluates to the same function value. Either way, we reduce to the case where the field obtained from $\e$ is a constant function $\funvalue$ over a certain domain, and there can be two cases:

\begin{itemize}

\item If $\funvalue$ is a built-in operator $\oname$, $\e(\e_1,\ldots,\e_n)$ maps an event to the result of applying $\oname$ to the values at the same event of its $n\ge 0$ arguments $\e_1,\dots,\e_n$. Note that $\oname$ can be a \emph{pure operator}, involving neither state nor communication (e.g. mathematical functions like addition, comparison, and sine)---for instance, $+(1,2)$ is the expression evaluating to the constant-valued field $3$, also written $1+2$ for readability as for any other binary built-in operator. Alternatively, $\oname$ can be an \emph{environment-dependent operator}, modelling a sensor---for instance, 0-ary \texttt{sns-temp} is used to map each device $\deviceId$ where the built-in operator call
is evaluated to its local value of temperature, and the  0-ary \texttt{nbr-range} operator returns a neighbouring field value mapping each neighbour of the device $\deviceId$ where the built-in operator call is evaluated to an estimate of its current distance from $\deviceId$.

\item If $\funvalue$ is not a built-in operator, it can be a declared function $\fname$ with corresponding declaration $\defK \; \fname (\xname_1\,\dots\,\xname_n) \; \{\e\}$, or an anonymous function $(\xname_1\,\dots\,\xname_n) \; \toSymK \; \e$; then, expression $\e(\e_1,\ldots,\e_n)$ maps an event to the result of evaluating the closed expression $\e_0$ obtained from the body $\e$ of the function $\funvalue$ by replacing the occurrences of the formal parameters $\xname_1,\dots,\xname_n$ with the values of the 
expressions $\e_1,\ldots,\e_n$.

\end{itemize}

%
\end{enumerate}

\begin{remark}[Function Equality] \label{rmk:equality}
	In the definition of function calls given above, it is necessary to specify when two functional values $\funvalue$ are ``the same''. In the remainder of this paper, we assume that two such values are the same when they are \emph{syntactically equal}. This convention is carried over in the meaning of the builtin operator $\mathtt{=}$ when applied to function values, so that $\mathtt{=}(\funvalue_1, \funvalue_2)$ with $\funvalue_1$, $\funvalue_2$ values holds precisely when $\funvalue_1$, $\funvalue_2$ are syntactically identical  expressions.
\end{remark}

%

According to the explanation given above, calling a declared or anonymous function 
acts as a branch, with each function in the range
applied only on the subspace of devices that hold that function.
Moreover, functional values allow code to be dynamically injected, moved, and executed in
network (sub)domains.
Namely:
\emph{(i)} functions can take functions as arguments and return a function as result (higher-order functions);
\emph{(ii)} (anonymous) functions can be created ``on the fly''; 
\emph{(iii)} functions can be moved between devices (via the $\nbrK{}$ construct); and 
\emph{(iv)} the function one executes can change over time (via the $\repK{}$ construct).

In this section, we have described the various constructs working in isolation: more involved examples dealing with combinations of constructs will be given in later sections, when the denotational and operational semantics are discussed.


\begin{remark}[Syntactic Sugar] \label{rmk:conditional}
Conventional branching can be implemented using a function call: in the examples we give, we will use the $\ifK$-expression $\ifK\;(\e_0)\;\{\e_1\}\; \elseK\;\{\e_2\}$ as syntactic sugar for:
\[
\muxK(\e_0, () \toSymK \sndK(\dcOf{\PairK}{\truevalue,\e_1}), () \toSymK \sndK(\dcOf{\PairK}{\falsevalue,\e_2}))()
\]
where function $\muxK$ is a built-in function multiplexer as defined in Section~\ref{s:conceptual:fields}, and function $\sndK$ extracts the second component of a pair. 
The result of branching is to partition the space-time domain in two subdomains: in the one where $\e_0$ evaluates to $\truevalue$, field $\e_1$ is computed and returned; in the complement where $\e_0$ evaluates to $\falsevalue$ , field $\e_2$ is computed and returned.
\end{remark}

 \subsection{Typing}\label{sec-typing}

In this section, we present a variant of the Hindley-Milner type system~\cite{Damas-Milner:POPL-1982} for the proposed calculus.
This type system has two main kinds of types, \emph{local types} (the types for local values) and \emph{field types} (the types for neighbouring field values), and is designed specifically to ground the denotational semantics presented in next section and to guarantee the following two properties:
\begin{description}
\item[Type Preservation] For every well-typed closed expression $\e$ of type $\type$, if the evaluation of $\e$ on event $\eventId$ (cf.\ the explanation at the beginning of Section~\ref{sec-calculus-values}) yields a result $\anyvalue$, then  $\anyvalue$ is of type $\type$.
\item[Domain Alignment] For every well-typed closed expression $\e$ of type $\type$, if the evaluation of $\e$ at event $\eventId$ yields a neighbouring field value $\fvalue$, then  the domain of $\fvalue$ consists of the device $\deviceId$ of event $\eventId$ and of its aligned neighbours, that is, the neighbours that have calculated the same expression $\e$ before the current evaluation started.
\end{description}
Domain alignment is key to guarantee that the semantics correctly relates the behaviour of $\nbrK$, $\repK$ and function application---namely, two neighbouring field values with different domain are never allowed to be combined. 
In essence, domain alignment is required in order to provide lexical scoping: without it, information may leak unexpectedly between devices that are evaluating different functions, or may be blocked from passing between devices evaluating the same function.

Since the type system is a customisation of the Hindley-Milner type system~\cite{Damas-Milner:POPL-1982} to the field calculus, there is an algorithm (not presented here) that, given an expression $\e$ and type assumptions for its free variables, either fails (if the expression cannot be typed under the given type assumptions)  or returns its \emph{principal type}, i.e., a type  such that all the types that can  be assigned to $\e$ by the type inference rules can be obtained from the principal type by substituting type variables with types. This algorithm is based on an \emph{unification} routine as in \cite{Smolka:SortedUnification} which exists since the type variables $\tvar$, $\ltvar$, $\rtvar$, $\stvar$ form a boolean algebra (i.e. $\stvar$ is exactly the intersection of $\ltvar$ and $\rtvar$, while $\tvar$ is their union).


The syntax of types and local type schemes is given in  Figure~\ref{fig:SurfaceTyping} (top).
A \emph{type}  $\type$ is either a \emph{type variable} $\tvar$, or a local type, or a field type.
A \emph{local type} $\ltype$ is either a \emph{local type variable} $\ltvar$, or a built-in type $\builtintype$ (numbers, booleans, pairs, lists etc.), or the type of a function $(\overline{\type}) \rightarrow \rtype$ (possibly of arity zero). Note that a function always has local type, regardless of the local or field type of its arguments.
A \emph{return type} $\rtype$ is either a \emph{return type variable} $\rtvar$, or a local return type, or a field type.
A \emph{local return type} $\stype$ is either a \emph{local return type variable} $\stvar$, or a built-in type $\builtintype$, or the type of a function $(\overline{\type}) \rightarrow \stype$ (possibly of arity zero).
A \emph{field type} $\ftype$ is the type $\ftypeOf{\stype}$ of a field whose range contains values of local return type $\stype$.
Notice that the type system does not contemplate types of the kind $(\ldots) \to \cdots \to (\ldots) \to F$ (functions that return functions that return neighbouring field values), since expressions involving such types can be unsafe (as exemplified in the next subsection). Notice also that $\ltype$ is equal to $\stype$ together with functions $(\overline{\type}) \rightarrow \ftype$ (thus $\type$ is $\rtype$ together with such functions).

\emph{Local type schemes}, ranged over by $\ltypescheme$, support typing  polymorphic uses of data constructors, built-in operators and user defined-functions. Namely, for each data constructor, built-in operator or user-defined function $\ofname$ there is a \emph{local type scheme} $\forall\overline{\tvar}\overline{\ltvar}\overline{\rtvar}\overline{\stvar}.\ltype$, where $\overline{\tvar}$, $\overline{\ltvar}$, $\overline{\rtvar}$ and $\overline{\stvar}$ are all the type variables occurring in the type $\ltype$, respectively.
Each use of $\ofname$ can be typed with any type obtained from $\forall\overline{\tvar}\overline{\ltvar}\overline{\rtvar}\overline{\stvar}.\ltype$ by replacing the type variables $\overline{\tvar}$ with types, $\overline{\ltvar}$ with local types, $\overline{\rtvar}$ with return types and $\overline{\stvar}$ with local return types. 

\emph{Type environments}, ranged over by $\TtypEnv$ and written $\overline{\xname}:\overline{\type}$, are used to collect type assumptions for program variables (i.e., the formal parameters of the functions and the variables introduced by the $\repK$-construct). 
\emph{Local-type-scheme environments}, ranged over by $\LTStypEnv$ and written $\overline{\ofname}:\overline{\ltypescheme}$, are used to collect the local type schemes for the data constructors and built-in operators together with the local type schemes inferred for the user-defined functions. In particular, the distinguished  \emph{built-in local-type-scheme environment}  $\OStypEnv$ associates a local type scheme to each data constructor $\dc$ and built-in function $\oname$ --- Figure~\ref{fig:typeof-built-in} shows the local type schemes for the data constructors and built-in functions used throughout this paper. We distinguish the built-in functions in \emph{pure} (their evaluation only depends on arguments) and \emph{non-pure} (their evaluation can depend on the specific device and on its physical environment, like e.g. for sensors).

We use the convention that if $\oname$ is a built-in unary operator with local argument and return type, $\oname\texttt{[f]}$ denote the corresponding operator on neighbouring field values which apply $\oname$ pointwise to its argument: in other words, $\oname\texttt{[f]}(\fvalue)$, which is equivalent to $\texttt{map-hood}(\oname,\fvalue)$ (see Figure~\ref{fig:typeof-built-in}), at any device maps a neighbour $\deviceId$ to the result of applying $\oname$ to the value of $\fvalue$ at $\deviceId$. If $\oname$ is a multi-ary operator, a notation such as $\oname\texttt{[f,l]}$ or $\oname\texttt{[l,f,f]}$ is used to specify which parameters have to be promoted to neighbouring field values: for instance, at each device, \texttt{+[f,f]($\fvalue_1$,$\fvalue_2$)} gives a neighbouring field value mapping a neighbour $\deviceId$ to the sum of values of $\fvalue_1$ and $\fvalue_2$ at $\deviceId$.
Notice that the definition of built-in operator $\texttt{map-hood}$ is actually a schema defining such an operator for any positive ariety of its arguments $\overline{\ltvar}$.


 \begin{figure}[t!]{
 \framebox[1\textwidth]{
 $\begin{array}{l}
\textbf{Types:}\\
\begin{array}{rcl@{\hspace{7.5cm}}r}
\type & \BNFcce &  \tvar  \; \BNFmid \;  \ftype \; \BNFmid \;  \ltype        &   {\footnotesize \mbox{type}} \\
\ltype & \BNFcce &  \ltvar  \; \BNFmid \;  \builtintype  \; \BNFmid \;  (\overline{\type}) \rightarrow \rtype    &   {\footnotesize \mbox{local type}} \\
\rtype & \BNFcce &  \rtvar \; \BNFmid \;  \ftype \; \BNFmid \;  \stype  &   {\footnotesize \mbox{return type}} \\
\stype & \BNFcce &  \stvar  \; \BNFmid \;  \builtintype  \; \BNFmid \;  (\overline{\type}) \rightarrow \stype    &   {\footnotesize \mbox{local return type}} \\
\ftype & \BNFcce &  \ftypeOf{\stype}    &   {\footnotesize \mbox{field type}} \\
\end{array}\\
\textbf{Local type schemes:}\\
\begin{array}{rcl@{\hspace{8.5cm}}r}
\ltypescheme & \BNFcce &  \forall\overline{\tvar}\overline{\ltvar}\overline{\rtvar}\overline{\stvar}.\ltype      &   {\footnotesize \mbox{local type scheme}} \\
%
\end{array}\\
\hline\\[-8pt]
\textbf{Expression typing:} 
  \hfill   \hfill \hfill   \hfill \hfill   \hfill \hfill   \hfill \hfill   \hfill \hfill   \hfill \hfill   \hfill \hfill   \hfill
 \hfill   \hfill \hfill   \hfill \hfill   \hfill \hfill   \hfill \hfill   \hfill \hfill   \hfill \hfill   \hfill \hfill   \hfill
\hfill   \hfill \hfill   \hfill \hfill   \hfill \hfill   \hfill \hfill   \hfill \hfill   \hfill \hfill   \hfill \hfill   \hfill
\hfill   \hfill \hfill   \hfill \hfill   \hfill \hfill   \hfill \hfill   \hfill \hfill   \hfill \hfill   \hfill \hfill   \hfill
  \hfill   \hfill \hfill   \hfill \hfill   \hfill \hfill   \hfill \hfill   \hfill \hfill   \hfill \hfill   \hfill \hfill   \hfill
 \hfill   \hfill \hfill   \hfill \hfill   \hfill \hfill   \hfill \hfill   \hfill \hfill   \hfill \hfill   \hfill \hfill   \hfill
\hfill   \hfill \hfill   \hfill \hfill   \hfill \hfill   \hfill \hfill   \hfill \hfill   \hfill \hfill   \hfill \hfill   \hfill
\hfill   \hfill \hfill   \hfill \hfill   \hfill \hfill   \hfill \hfill   \hfill \hfill   \hfill \hfill   \hfill \hfill   \hfill
  \hfill   \hfill \hfill   \hfill \hfill   \hfill \hfill   \hfill \hfill   \hfill \hfill   \hfill \hfill   \hfill \hfill   \hfill
 \hfill   \hfill \hfill   \hfill \hfill   \hfill \hfill   \hfill \hfill   \hfill \hfill   \hfill \hfill   \hfill \hfill   \hfill
\hfill   \hfill \hfill   \hfill \hfill   \hfill \hfill   \hfill \hfill   \hfill \hfill   \hfill \hfill   \hfill \hfill   \hfill
\hfill   \hfill \hfill   \hfill \hfill   \hfill \hfill   \hfill \hfill   \hfill \hfill   \hfill \hfill   \hfill \hfill   \hfill
  \hfill   \hfill \hfill   \hfill \hfill   \hfill \hfill   \hfill \hfill   \hfill \hfill   \hfill \hfill   \hfill \hfill   \hfill
 \hfill   \hfill \hfill   \hfill \hfill   \hfill \hfill   \hfill \hfill   \hfill \hfill   \hfill \hfill   \hfill \hfill   \hfill
\hfill   \hfill \hfill   \hfill \hfill   \hfill \hfill   \hfill \hfill   \hfill \hfill   \hfill \hfill   \hfill \hfill   \hfill
\hfill   \hfill \hfill   \hfill \hfill   \hfill \hfill   \hfill \hfill   \hfill \hfill   \hfill \hfill   \hfill \hfill   \hfill
  \boxed{\expTypJud{\LTStypEnv}{\TtypEnv}{\e}{\type}}
\vspace{0.1cm}
  \\
\begin{array}{c}
\nullsurfaceTyping{T-VAR}{
\expTypJud{\LTStypEnv}{\TtypEnv,\xname:\type}{\xname}{\type}
}
\quad\!
\surfaceTyping{T-VAL}{ \qquad
\funTypJud{\OStypEnv}{\dc}{(\overline{\ltype}) \rightarrow \stype} \qquad
\expTypJud{\LTStypEnv}{\TtypEnv}{\overline{\lvalue}}{\overline{\ltype}} }
{
\expTypJud{\LTStypEnv}{\TtypEnv}{\dcOf{\dc}{\overline{\lvalue}}}{\stype}
}
\skiptransition
\surfaceTyping{T-N-FUN}{  \quad \ltype' = \ltype \substitution{\overline{\tvar}}{\overline{\type}} \substitution{\overline{\ltvar}}{\overline{\ltype}} \substitution{\overline{\rtvar}}{\overline{\rtype}} \substitution{\overline{\stvar}}{\overline{\stype}}
}
{\expTypJud{\LTStypEnv,\ofname:\forall\overline{\tvar}\overline{\ltvar}\overline{\rtvar}\overline{\stvar}.\ltype}
                {\TtypEnv}
                {\ofname}
                {\ltype' }
 }
\skiptransition
\surfaceTyping{T-A-FUN}{ \qquad
\overline{\yname}=\FV{(\overline{\xname}) \toSymK \e} 
\quad
\TtypEnv(\overline{\yname}) \; \mbox{local types}
\quad
\expTypJud{\LTStypEnv}{\;\TtypEnv,\,\overline{\xname}:\overline{\type}}{\e}{\rtype}
}{ \expTypJud{\LTStypEnv}{\TtypEnv}{ (\overline{\xname}) \toSymK \e}{(\overline{\type})\rightarrow\rtype} }
\skiptransition
\surfaceTyping{T-APP}{ \qquad
\expTypJud{\LTStypEnv}{\TtypEnv}{\e}{(\overline{\type})\rightarrow\rtype} \qquad
\expTypJud{\LTStypEnv}{\TtypEnv}{\overline{\e}}{\overline{\type}} }{
\expTypJud{\LTStypEnv}{\TtypEnv}{\e(\overline{\e})}{\rtype} }
\skiptransition
\surfaceTyping{T-REP}{ \qquad
\expTypJud{\LTStypEnv}{\TtypEnv}{\e_1}{\stype}
\qquad \expTypJud{\LTStypEnv}{\TtypEnv,\xname:\ltype}{\e_2}{\stype} }{
\expTypJud{\LTStypEnv}{\TtypEnv}{\repK(\e_1)\{(\xname) \toSymK \e_2\}}{\stype} }
\qquad\qquad
\surfaceTyping{T-NBR}{ \qquad
\expTypJud{\LTStypEnv}{\TtypEnv}{\e}{\stype}
}{ \expTypJud{\LTStypEnv}{\TtypEnv}{\nbrK\{\e\}}{\ftypeOf{\stype}} }
\skiptransition
\end{array}
\\
\textbf{Function typing:} 
  \hfill   \hfill \hfill   \hfill \hfill   \hfill \hfill   \hfill \hfill   \hfill \hfill   \hfill \hfill   \hfill \hfill   \hfill
 \hfill   \hfill \hfill   \hfill \hfill   \hfill \hfill   \hfill \hfill   \hfill \hfill   \hfill \hfill   \hfill \hfill   \hfill
\hfill   \hfill \hfill   \hfill \hfill   \hfill \hfill   \hfill \hfill   \hfill \hfill   \hfill \hfill   \hfill \hfill   \hfill
\hfill   \hfill \hfill   \hfill \hfill   \hfill \hfill   \hfill \hfill   \hfill \hfill   \hfill \hfill   \hfill \hfill   \hfill
  \hfill   \hfill \hfill   \hfill \hfill   \hfill \hfill   \hfill \hfill   \hfill \hfill   \hfill \hfill   \hfill \hfill   \hfill
 \hfill   \hfill \hfill   \hfill \hfill   \hfill \hfill   \hfill \hfill   \hfill \hfill   \hfill \hfill   \hfill \hfill   \hfill
\hfill   \hfill \hfill   \hfill \hfill   \hfill \hfill   \hfill \hfill   \hfill \hfill   \hfill \hfill   \hfill \hfill   \hfill
\hfill   \hfill \hfill   \hfill \hfill   \hfill \hfill   \hfill \hfill   \hfill \hfill   \hfill \hfill   \hfill \hfill   \hfill
  \hfill   \hfill \hfill   \hfill \hfill   \hfill \hfill   \hfill \hfill   \hfill \hfill   \hfill \hfill   \hfill \hfill   \hfill
 \hfill   \hfill \hfill   \hfill \hfill   \hfill \hfill   \hfill \hfill   \hfill \hfill   \hfill \hfill   \hfill \hfill   \hfill
\hfill   \hfill \hfill   \hfill \hfill   \hfill \hfill   \hfill \hfill   \hfill \hfill   \hfill \hfill   \hfill \hfill   \hfill
\hfill   \hfill \hfill   \hfill \hfill   \hfill \hfill   \hfill \hfill   \hfill \hfill   \hfill \hfill   \hfill \hfill   \hfill
  \hfill   \hfill \hfill   \hfill \hfill   \hfill \hfill   \hfill \hfill   \hfill \hfill   \hfill \hfill   \hfill \hfill   \hfill
 \hfill   \hfill \hfill   \hfill \hfill   \hfill \hfill   \hfill \hfill   \hfill \hfill   \hfill \hfill   \hfill \hfill   \hfill
\hfill   \hfill \hfill   \hfill \hfill   \hfill \hfill   \hfill \hfill   \hfill \hfill   \hfill \hfill   \hfill \hfill   \hfill
\hfill   \hfill \hfill   \hfill \hfill   \hfill \hfill   \hfill \hfill   \hfill \hfill   \hfill \hfill   \hfill \hfill   \hfill
  \boxed{\funTypJud{\LTStypEnv}{\FUNCTION}{\ltypescheme}}
  \\
\begin{array}{c}
\surfaceTyping{T-FUNCTION}{
\qquad
\expTypJud{\LTStypEnv,\,\fname:\forall\emptyseq.\overline{\type}\rightarrow\rtype}{\overline{\xname}:\overline{\type}}{\e}{\rtype}
\qquad
\overline{\tvar}\overline{\ltvar}\overline{\rtvar}\overline{\stvar}=\FTV{(\overline{\type})\rightarrow\rtype}
}{ \funTypJud{\LTStypEnv}{\defK \; \fname (\overline{\xname}) \; \{\e\}}{\forall\overline{\tvar}\overline{\ltvar}\overline{\rtvar}\overline{\stvar}.(\overline{\type})\rightarrow\rtype}}
\skiptransition
\end{array}
\\
\textbf{Program typing:} 
  \hfill   \hfill \hfill   \hfill \hfill   \hfill \hfill   \hfill \hfill   \hfill \hfill   \hfill \hfill   \hfill \hfill   \hfill
 \hfill   \hfill \hfill   \hfill \hfill   \hfill \hfill   \hfill \hfill   \hfill \hfill   \hfill \hfill   \hfill \hfill   \hfill
\hfill   \hfill \hfill   \hfill \hfill   \hfill \hfill   \hfill \hfill   \hfill \hfill   \hfill \hfill   \hfill \hfill   \hfill
\hfill   \hfill \hfill   \hfill \hfill   \hfill \hfill   \hfill \hfill   \hfill \hfill   \hfill \hfill   \hfill \hfill   \hfill
  \hfill   \hfill \hfill   \hfill \hfill   \hfill \hfill   \hfill \hfill   \hfill \hfill   \hfill \hfill   \hfill \hfill   \hfill
 \hfill   \hfill \hfill   \hfill \hfill   \hfill \hfill   \hfill \hfill   \hfill \hfill   \hfill \hfill   \hfill \hfill   \hfill
\hfill   \hfill \hfill   \hfill \hfill   \hfill \hfill   \hfill \hfill   \hfill \hfill   \hfill \hfill   \hfill \hfill   \hfill
\hfill   \hfill \hfill   \hfill \hfill   \hfill \hfill   \hfill \hfill   \hfill \hfill   \hfill \hfill   \hfill \hfill   \hfill
  \hfill   \hfill \hfill   \hfill \hfill   \hfill \hfill   \hfill \hfill   \hfill \hfill   \hfill \hfill   \hfill \hfill   \hfill
 \hfill   \hfill \hfill   \hfill \hfill   \hfill \hfill   \hfill \hfill   \hfill \hfill   \hfill \hfill   \hfill \hfill   \hfill
\hfill   \hfill \hfill   \hfill \hfill   \hfill \hfill   \hfill \hfill   \hfill \hfill   \hfill \hfill   \hfill \hfill   \hfill
\hfill   \hfill \hfill   \hfill \hfill   \hfill \hfill   \hfill \hfill   \hfill \hfill   \hfill \hfill   \hfill \hfill   \hfill
  \hfill   \hfill \hfill   \hfill \hfill   \hfill \hfill   \hfill \hfill   \hfill \hfill   \hfill \hfill   \hfill \hfill   \hfill
 \hfill   \hfill \hfill   \hfill \hfill   \hfill \hfill   \hfill \hfill   \hfill \hfill   \hfill \hfill   \hfill \hfill   \hfill
\hfill   \hfill \hfill   \hfill \hfill   \hfill \hfill   \hfill \hfill   \hfill \hfill   \hfill \hfill   \hfill \hfill   \hfill
\hfill   \hfill \hfill   \hfill \hfill   \hfill \hfill   \hfill \hfill   \hfill \hfill   \hfill \hfill   \hfill \hfill   \hfill
  \boxed{\proTypJud{\PROGRAM}{\ltype}}
  \\
\begin{array}{c}
\surfaceTyping{T-PROGRAM}{
\\
\LTStypEnv_0=\OStypEnv
\\
\FUNCTION_i=(\defK \; \fname_i (\_) \; \_)
\qquad
\funTypJud{\LTStypEnv_{i-1}}{\FUNCTION_i}{\ltypescheme_i}
\qquad
\LTStypEnv_i=\LTStypEnv_{i-1},\, \fname_i:\ltypescheme_i
\qquad
 (i \in 1..n)
\\
\expTypJud{\LTStypEnv_n}{\emptyset}{\e}{\ltype}
}{ \proTypJud{\FUNCTION_1\cdots\FUNCTION_n  \;
\e}{\ltype}}
\end{array}
\end{array}$}
} \caption{\HFC: types, local type schemes and type rules for expressions, function
declarations, and programs.} \label{fig:SurfaceTyping}
\end{figure}

\begin{figure}[t]{
\centerline{\framebox[\textwidth]{ $\begin{array}{l}
\textbf{Built-in constructors:}\\
\begin{array}{lcl@{\hspace{0.2cm}}r}
\typeof{\truevalue} & = & ()\to\btype
 \\
\typeof{\falsevalue} & = & ()\to\btype
 \\
\typeof{\zerovalue} & = & ()\to\ntype
 \\
\typeof{\PairK} & = & \forall\stvar_1\stvar_2.(\stvar_1,\stvar_2)\rightarrow \pairltypeOf{\stvar_1}{\stvar_2}  
 \\
\typeof{\NullK} & = & \forall\stvar. ()\to \listltypeOf{\stvar}
 \\
\typeof{\ConsK} & = & \forall\stvar. (\stvar,\listltypeOf{\stvar})\to \listltypeOf{\stvar}
 \\
 \end{array}
\\
\textbf{Pure built-in functions  (independent from  the current device and value-tree environment):}\\
\begin{array}{lcl@{\hspace{0.2cm}}r}
\typeof{\pairK\texttt{[f,f]}} & = & \forall\stvar_1\stvar_2.(\ftypeOf{\stvar_1},\ftypeOf{\stvar_2})\rightarrow\ftypeOf{\pairltypeOf{\stvar_1}{\stvar_2}}  
 \\
\typeof{\pairK\texttt{[l,f]}} & = & \forall\stvar_1\stvar_2.(\stvar_1,\ftypeOf{\stvar_2})\rightarrow\ftypeOf{\pairltypeOf{\stvar_1}{\stvar_2}}
 \\
\typeof{\fstK} & = & \forall\stvar_1\stvar_2.(\pairltypeOf{\stvar_1}{\stvar_2})\rightarrow\stvar_1  
 \\
\typeof{\sndK} & = & \forall\stvar_1\stvar_2.(\pairltypeOf{\stvar_1}{\stvar_2})\rightarrow\stvar_2  
 \\
\typeof{\headK} & = & \forall\stvar.(\listltypeOf{\stvar}) \rightarrow \stvar
 \\
\typeof{\tailK} & = & \forall\stvar.(\listltypeOf{\stvar}) \rightarrow \listltypeOf{\stvar}
 \\
\typeof{\minHoodK} & = & \forall\stvar.(\ftypeOf{\stvar})\to\stvar
 \\
\typeof{\texttt{min-hood+}} & = & \forall\stvar.(\ftypeOf{\stvar})\to\stvar
 \\
\typeof{\texttt{pick-hood}} & = & \forall\stvar.(\ftypeOf{\stvar})\to\stvar
 \\
\typeof{\texttt{map-hood}} & = & \forall \overline{\stvar}. \forall \stvar' ((\overline{\stvar}) \to \stvar', \ftypeOf{\overline{\stvar}}) \to \ftypeOf{\stvar'}
 \\
\typeof{\texttt{fold-hood}} & = & \forall\stvar.((\stvar, \stvar) \to \stvar, \ftypeOf{\stvar}) \to \stvar
\\
\typeof{\texttt{mux}} & = & \forall\stvar.(\btype,\stvar,\stvar)\to\stvar  
 \\
\typeof{\texttt{mux[f,f,l]}} & = & \forall\stvar.(\ftypeOf{\btype},\ftypeOf{\stvar},\stvar)\to\ftypeOf{\stvar}
\\
\typeof{\texttt{and}} & = & (\btype,\btype)\to\btype
\\
\typeof{\texttt{*}} & = & (\ntype,\ntype)\to\ntype
\\
\typeof{\texttt{-}} & = & (\ntype,\ntype)\to\ntype
\\
\typeof{\texttt{+}} & = & (\ntype,\ntype)\to\ntype
 \\
\typeof{\texttt{+[f,f]}} & = & (\ftypeOf{\ntype},\ftypeOf{\ntype})\to\ftypeOf{\ntype}
 \\
\typeof{\texttt{<[f,l]}} & = & \forall\stvar.(\ftypeOf{\stvar},\stvar)\to\ftypeOf{\btype}
 \\
\typeof{\texttt{=}} & = & \forall\tvar.(\tvar,\tvar) \to \btype
\\
 \end{array}
\\
\textbf{Non-pure built-in functions (depend from  the current device and value-tree environment):}\\
\begin{array}{lcl@{\hspace{0.2cm}}r}
\typeof{\texttt{sns-range}} & = & ()\to\ntype
\\
\typeof{\texttt{sns-injection-point}} & = & ()\to\btype
\\
\typeof{\texttt{sns-injected-function}} & = & ()\to(()\to\ntype)
 \\
\typeof{\texttt{nbr-range}} & = & ()\to\ftypeOf{\ntype}
 \\
\typeof{\selfK} & = & () \to \ntype  
 \\
 \end{array}
\end{array}
 $}}}
\caption{Local type schemes for the built-in functions used throughout this paper.}
\label{fig:typeof-built-in}
\end{figure}

The type rules are given in Figure~\ref{fig:SurfaceTyping} (bottom).
The typing judgement for expressions is of the form ``$\expTypJud{\LTStypEnv}{\TtypEnv}{\e}{\type}$'', to be read: ``$\e$ has type $\type$ under the local-type-scheme assumptions 
$\LTStypEnv$ (for data constructors, built-in operators and user-defined functions) and the type assumptions
$\TtypEnv$ (for the program variables occurring in $\e$)". 
As a standard syntax in type systems~\cite{FJ}, given
$\overline{\type}=\type_1,\ldots,\type_n$ and
$\overline{\e}=\e_1,\ldots,\e_n$ ($n\ge 0$), we write
$\expTypJud{\LTStypEnv}{\TtypEnv}{\overline{\e}}{\overline{\type}}$
as short for $\expTypJud{\LTStypEnv}{\TtypEnv}{\e_1}{\type_1}$
$\cdots$ $\expTypJud{\LTStypEnv}{\TtypEnv}{\e_n}{\type_n}$.
Note that the type rules are syntax-directed, so they straightforwardly describe a type inference algorithm.

Rule \ruleNameSize{[T-VAR]} (for  variables) lookups the type assumptions for $\xname$ in $\TtypEnv$. 

Rule \ruleNameSize{[T-VAL]} (for values) lookups the specifications of the data constructor $\dc$ in $\typeofNAME$ and checks that are met by its arguments, which need to be \emph{values}. This request allows us to recognize whether a well-typed expression is a value by only checking its outermost syntactic block, which is convenient for later proofs. For convenience of presentation, in the following we shall assume that every constructor $\dc : (\overline{\ltype}) \to \stype$ comes with an associated function $\dc'$ defined as $(\overline{\xname}) \; \toSymK \; \dcOf{\dc}{\overline{\xname}}$. Notice that specifications for constructors are only allowed to involve local types.


Rule \ruleNameSize{[T-N-FUN]} (for built-in functions and user-defined function names) ensures that the local type scheme $\forall\overline{\tvar}\overline{\ltvar}\overline{\rtvar}\overline{\stvar}.\ltype$ associated to the built-in function or user-defined function name $\ofname$ is instantiated by substituting the type variables $\overline{\tvar}$, $\overline{\ltvar}$, $\overline{\rtvar}$, $\overline{\stvar}$ correctly with types in $\type$, $\ltype$, $\rtype$, $\stype$ respectively.

Rule \ruleNameSize{[T-A-FUN]} (for anonymous functions) ensures that anonymous functions  $(\overline{\xname}) \toSymK \e$ have return type in $\rtype$ and do not contain free variables of field type in order to avoid domain alignment errors (as exemplified in the next subsection).

Rule \ruleNameSize{[T-APP]} (for function applications) is standard.

Rule \ruleNameSize{[T-REP]} (for $\repK$-expressions) ensures that both the variable $\xname$, its initial value $\e_1$ and the body $\e_2$ have (the same) local return type. In fact, allowing field types might produce domain mismatches, while a $\repK$-expression of type $(\overline{\type})\to\ftype$ would be a non-constant (thus not safely applicable) function returning neighbouring field values.

Rule \ruleNameSize{[T-NBR]} (for $\nbrK$-expressions) ensures that the body $\e$ of the expression has a local return type. This prevents the attempt to create a ``field of fields'' (i.e., a neighbouring field value that maps device identities to neighbouring field values). 

Function declaration typing (represented by judgement ``$\funTypJud{\LTStypEnv}{\FUNCTION}{\ltypescheme}$'') and program typing (represented by judgement ``$\proTypJud{\PROGRAM}{\ltype}$'') are almost standard.


We say that a program $\PROGRAM$  is \emph{well-typed} to mean that $\proTypJud{\PROGRAM}{\ltype}$ holds for some local type $\ltype$.

\begin{remark}[On Termination] \label{rmk:termination}
Termination of a device firing is clearly not decidable. In the rest of the paper we assume without loss of generality for the results of this paper that a decidable subset of the termination fragment has been identified  by applying some static analysis technique for termination.
\end{remark}

\subsection{Examples}

Though the type system mostly trivially assigns a local or field type to an expression, it enforces some peculiar restrictions that we now clarify with the help of few examples. 
In the code to come, syntax is coloured to increase readability: grey for comments, red for field calculus keywords, and blue
for functions (both user defined and built in).
In the first-order type system presented in \cite{VDB-FOCLASA-CIC2013} one peculiar check was introduced in the type system: a conditional \texttt{if} expression could not have field type. In fact, such an expression produces a field combining two subfields where values, which are neighbouring field values, are restricted in different ways, each to the neighbours that evaluated the conditional guard in the same way. Such a combination, hence, is shown to contradict domain alignment as described by the following example---in the present language conditional branching is modeled by the branching implicit in function application (see Remark~\ref{rmk:conditional}), thus similar issues apply to functions returning neighbouring field values. 
Consider the expression $\e_{\mathtt{wrong}}$:
\begin{Verbatim}[%fontsize=\footnotesize,
	                  frame=single,
	                  commandchars=\\\{\}]
(\km{if} (\pr{uid}=1) \{(x)=>x\} \km{else} \{(x)=>x \pr{+[f,f]} \km{nbr}\{\pr{uid}\}\} )(\km{nbr}\{\pr{0}\}) \pr{+[f,f]} \km{nbr}\{\pr{uid}\}
\end{Verbatim}
This expression violates domain alignment, thus provoking conflicts between field domains. When the conditional expression is evaluated on a device with \texttt{uid} (unique identifier) equals to $1$, function $\funvalue = (\xname) \toSymK \xname$ is obtained whence applied to $\nbrK\{0\}$ \emph{in the restricted domain of devices who computed the same $\funvalue$ in their last evaluation round}, that is the domain $\{1\}$. Thus at device $1$ the function application returns the neighbouring field value $\fvalue = \envmap{1}{0}$ which cannot be combined with $\nbrK\{\selfK\}$ whose (larger) domain consists of \emph{all} neighbours of device $1$.  
A complementary violation occurs for all neighbors of device $1$, which compute a neighbouring field value whose domain lacks device $1$.

However, not all expressions involving functions returning fields are unsafe. For instance, consider the similar expression $\e_{\mathtt{safe}}$:
\begin{Verbatim}[%fontsize=\footnotesize,
	                  frame=single,
	                  commandchars=\\\{\}]
((x)=>x)(\km{nbr}\{\pr{0}\}) \pr{+[f,f]} \km{nbr}\{\pr{uid}\}
\end{Verbatim}
In this case, on every device the same function $\funvalue = (\xname) \toSymK \xname$ is obtained whence applied to $\nbrK\{0\}$, that is, no alignment is required thus the function application returns the whole neighbouring field value $\fvalue = \nbrK\{0\}$ which can be safely combined with $\nbrK\{\selfK\}$. This suggest that functions returning neighbouring field values are safe \emph{as long as they evaluate to the same function regardless of the device and surrounding environment}.

The type system presented in the previous Section \ref{sec-typing} ensures 
 this distinction. Besides performing standard checks, the type system perform the following additional checks in order to ensure domain alignment:
\begin{itemize}
	\item
	\emph{Functions returning neighbouring field values are not allowed as return type}. That is, all functions (both user defined and built-ins, and as a consequence also $\repK$ statements) don't return a ``function returning neighbouring field values''. This prevents the possibility of having a well-typed expression $\e$ which evaluates to different functions returning neighbouring field values on different devices, thus allowing undesired behaviours such as in the example described above. In fact, if we expand the conditional in $\e_{\mathtt{wrong}}$ according to Remark \ref{rmk:conditional}, we obtain\\
	
	\begin{Verbatim}[%fontsize=\footnotesize,
	                  frame=single,
	                  commandchars=\\\{\}]
\pr{mux}(\pr{uid}=1, ()=>\pr{snd}(\pr{Pair}(\pr{True},(x)=>x)),
           ()=>\pr{snd}(\pr{Pair}(\pr{False},(x)=>x\pr{+[f,f]}\km{nbr}\{\pr{uid}\})))(...)
	\end{Verbatim}
	in which both the entire \texttt{mux} expression and both of its two branches have the disallowed type $()\to \ftypeOf{\ntype} \to \ftypeOf{\ntype}$.
	
	\item
	\emph{In an anonymous function $(\overline{\xname}) \toSymK \e$, the free variables $\overline{\yname}$ of $\e$ that are not in $\overline{\xname}$ have local type}. This prevents a device $\deviceId$ from creating a closure $\e'=\applySubstitution{(\overline{\xname}) \toSymK \e}{\substitution{\overline{\yname}}{\overline{\fvalue}}}$ containing neighbouring field values $\overline{\fvalue}$ (whose domain is by construction equal to the subset of the aligned  neighbours of  $\deviceId$). The closure $\e'$ may lead to a domain alignment error since it  may be shipped (via the $\nbrK$  construct) to another device $\deviceId'$ that may use it (i.e., apply $\e'$  to some arguments); and the evaluation of the body of $\e'$ may involve use of a neighbouring field value $\fvalue$ in $\overline{\fvalue}$  such that the set of aligned neighbours of $\deviceId'$ is different from the domain of $\fvalue$.
	For instance, the expression $\e'_{\mathtt{wrong}}$:\\ 
	\begin{Verbatim}[%fontsize=\footnotesize,
	                  frame=single,
	                  commandchars=\\\{\}]
((x) => \pr{pick-hood}(\km{nbr}\{() => \pr{min-hood}(x \pr{+[f,f]} \km{nbr}\{\pr{0}\})\}))(\km{nbr}\{\pr{0}\})()
	\end{Verbatim}
	(where \texttt{pick-hood} is a built-in function that returns the value of  a randomly chosen device among  a device neighbours) that should have type $\ntype$, is ill-typed. Its body will fail to type-check since it contains the function $() \toSymK \minHoodK(\xname + \nbrK\{0\})$ with free variable $\xname$ of field type. This prevents conflicts between field domains since:
	\begin{itemize}
		\item
		when the expression is evaluated on a device $\deviceId$, the closure
		\[
		\lvalue=\applySubstitution{\texttt{() => \pr{min-hood}(x \pr{+[f,f]} \km{nbr}\{\pr{0}\})}}{\substitution{\xname}{\phi_2}}
		\]
		where $\phi_2$ is the neighbouring field value produced by the evaluation of \texttt{nbr\{0\}} on $\deviceId$ (whose domain consists of the aligned neighbours of the device $\deviceId$---i.e., the neighbours that have evaluated a corresponding occurrence of $\e'_{\mathtt{wrong}}$ in their last evaluation round), will be made available to other devices; and 
		
		\item
		when the expression is evaluated on a device $\deviceId'$ that has $\deviceId$ as neighbour  and the evaluation of the application of \texttt{pick-hood} returns the closure $\lvalue$ received from $\deviceId$; then the neighbouring field value $\phi'_1$ produced by the evaluation of the subexpression \texttt{nbr\{0\}} of $\lvalue$ on $\deviceId'$ would contain the aligned neighbours of the device $\deviceId'$ (i.e., the neighbours that have evaluated a corresponding occurrence of $\e'_{\mathtt{wrong}}$ in their last evaluation round) hence may have a domain different from the domain of $\phi_2$, leaving the neighbouring field values mismatched in domain at the evaluation of the sum occurring in $\lvalue$ on $\deviceId'$.
	\end{itemize}
	
	\item
	\emph{In  a $\repK$-expression  $\repK(\e_1)\{(\xname) \toSymK \e_2\}$ it holds that $\xname$, $\e_1$ and $\e_2$ have (the same) local return type}. This prevents a device $\deviceId$  from storing in $\xname$ a neighbouring field value $\fvalue$ that may be reused in the next computation round of $\deviceId$, when the set of the set of aligned neighbours may be different from the domain of $\fvalue$.
	For instance, the expression $\e''_{\mathtt{wrong}}$:\\
	\begin{Verbatim}[%fontsize=\footnotesize,
	                  frame=single,
	                  commandchars=\\\{\}]
\pr{min-hood}(\km{rep}(\km{nbr}\{0\})\{(x) => x \pr{+[f,f]} \km{nbr}\{\pr{uid}\}\})
	\end{Verbatim}
	that should have type $\ntype$, is ill-typed.
		
	\item
	\emph{In a $\nbrK$-expression $\nbrK\{\e\}$ the expression $\e$ has local type}. This prevents the attempt to create a ``field of fields'' (i.e., a neighbouring field value that maps device identifiers to neighbouring field values)---which is pragmatically often overly costly to maintain and communicate, as well as further complicating the issues involved in ensuring domain alignment.
\end{itemize}

\section{Denotational Semantics} \label{sec-denotational-semantics}

We now introduce a denotational semantics for the field calculus. 
In this semantics, we are posed with an additional challenge with respect to the denotational semantics for lambda calculus or ML-like programming languages (see among many \cite{Ohori:Polymorphism,Smolka:SortedUnification,Winskel:FormalSemantics}): 
several devices and firing events are involved in the computation, possibly influencing each other's outcomes. Even though this scenario seems similar to classical concurrent programming (see e.g. \cite{Baeten:ProcessAlgebras,Bakker:DenotationalConcurrency}), it cannot be treated using the same tools because of two crucial differences:
\begin{itemize}
	\item communication within a computational round is strongly connected with the underlying space-time properties of the physical world in which devices are located;
	\item the whole outcome of the computation is not a single value, but a \emph{field} of spatially and temporally distributed values.
\end{itemize}
In order to reflect these characteristics of field calculus, the denotational semantics of an expression is given in terms of the resulting space-time field, 
formalised as a partial mapping from the ``evolving'' domain (the set of participating devices may change over time) to values. 
The domain is defined as a set of (firing) events, each carrying a node identifier and equipped with a neighbour relationship modelling causality, i.e., reachability of communications. 
Values are defined analogously to ML-like languages, with the addition of fields and formulating functions as mathematical operators on computational fields instead of single values 
(which is necessary because the computation of a function might involve state communication among devices or events).

This semantics is nicely compositional, allowing one to formalise each construct separately. 
Furthermore, it represents the global result of computation as a single ``space-time object'', thus giving a perspective that is more proper for the designer of a computation and more convenient in proving certain properties of the calculus---the operational semantics (Section \ref{sec-calculus-operational-semantics}) is instead more useful in designing a platform for the equivalent distributed execution of field computations on actual devices.

\subsection{Preliminary definitions}

Recall that we let $\deviceIdSet$ be the set of \emph{devices}, ranged over by meta-variable $\deviceId$; we now also let $\EventS$ be the set of \emph{events}, ranged over by meta-variable $\eventId$. An event models a firing in a network, and is labeled by a device identifier $\devF{\eventId}$. We use $\eventS$ to range over subsets of $\EventS$.
 
We model the neighbour relationship as a global-level, fixed predicate $\neighbour{\eventId}{\eventId'}$ which holds if the device at $\eventId$ is aware of the result of computation at $\eventId'$. This relationship is based on the topology and time evolution of involved devices, hence we require that $\neighbour{\eventId}{\eventId'}$ satisfies the following properties:
 \begin{enumerate}
 	\item the graph on $\EventS$ induced by $\textit{neigh}$ is a DAG (directed acyclic graph);
 	\item every $\eventId'$ linked from $\eventId$ has a different label $\devF{\eventId'}$, that is, there exists no $\eventId'$, $\eventId''$ such that $\devF{\eventId'} = \devF{\eventId''}$ and both $\neighbour{\eventId}{\eventId'}$ and $\neighbour{\eventId}{\eventId''}$ hold;
 	\item every $\eventId'$ is linked from at most one $\eventId$ with the same label $\devF{\eventId} = \devF{\eventId'}$, that is, there exists no $\eventId_0$, $\eventId_1$ such that $\devF{\eventId_0} = \devF{\eventId_1} = \devF{\eventId'}$ and both $\neighbour{\eventId_0}{\eventId'}$ and $\neighbour{\eventId_1}{\eventId'}$ hold;
 \end{enumerate}
Property 1 ensures that $\textit{neigh}$ is causality-driven, property 2 that the neighbours of an event $\eventId$ are indexed by devices, and property 3 that restricting $\textit{neigh}$ to a single device we obtain a set of directed paths (thus modeling that the firing of a device is aware of either its immediate predecessor on the same device or nothing).

\begin{figure}[t]
\centering
\includegraphics[width=0.4\textwidth]{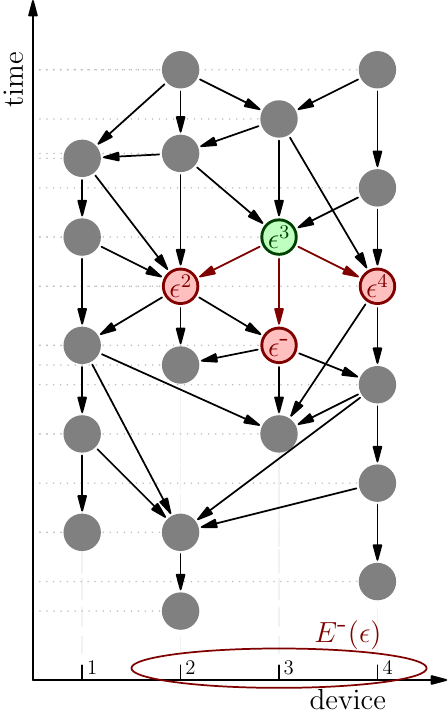}
\caption{A sample neighbourhood graph on four asynchronously firing devices.  The neighbours of the green event $\eventId = \eventId^3$ are shown in red, forming a neighbourhood over devices 2, 3, and 4.}\label{f:neigh}
\end{figure}

Figure~\ref{f:neigh} shows a sample neighbourhood graph involving four devices, each firing from four to six times. 
Notice that device 2 is restarted after its second firing, and that the set of neighbouring devices changes over time for each device, in particular, device 4 drops its connection with device 2 from its fourth firing on. This can be explained by assuming devices to be moving in space---though movement is not represented in Figure \ref{f:neigh}. Nonetheless, the depicted graph satisfies all of the three above mentioned properties.
 
For all $\eventId$ and $\deviceId$, we define $\nbrdevice{\eventId}{\deviceId}$ as the latest event at $\deviceId$ that $\eventId$ can be aware of, namely the one satisfying $\neighbour{\eventId}{\nbrdevice{\eventId}{\deviceId}}$ if $\deviceId \neq \deviceId_\eventId$ or $\eventId$ itself in case $\deviceId = \deviceId_\eventId$. Notice that if $\nbrdevice{\eventId}{\deviceId}$ exists, it is unique by property 2. 
We use $\repdevice{\eventId}$ to denote the previous event of $\eventId$ at the same device \emph{if it exists}.
These notations are exemplified in the picture above, where $\eventId = \eventId^3$.
We also define $\predevices{\eventS}(\eventId)$ where $\eventS \subseteq \EventS$ as the neighbourhood of $\eventS$, namely, the set of devices $\deviceId$ such that $\nbrdevice{\eventId}{\deviceId}$ exists in $\eventS$. For example, $\predevices{\eventS}(\eventId) = \{2,3,4\}$ for the green event $\eventId$ in the picture above.

We chose such a generic approach to model neighbouring in order to abstract from the particular conditions and implementations that might occur in practice when an execution platform has to handle device to device communication. 

\begin{example}[Unit-Disc Communication] \label{ex:denotational}
A typical scenario for the computations we aim at modelling and designing is that of a mobile set of wirelessly communicating devices, such that the neighbourhood relationship depends primarily on physical position---e.g., devices within a certain range can communicate. 
In this case, the predicate $\neighbour{\eventId}{\eventId'}$ could be defined from a set of \emph{paths} $\PathS$ for moving devices, labels $\timeF{\eventId}$ modelling passage of time, a timeout value $\decay$, and a predicate $\neighbour{\posS}{\posS'}$ between positions, where:
\begin{itemize}
 \item $\neighbour{\posS}{\posS'}$ is a global-level, fixed reflexive and symmetric predicate which holds if the two devices at positions $\posS$ and $\posS'$ are neighbours.
 \item $\PathS$ is a mapping from device identifiers $\deviceId$ to space-time paths $\pathS$. A path $\pathS$ is a continuous function from $\mathbb{R}^+$ (times) to the set of possible positions, defined on the union of a finite number of disjoint closed intervals (the time intervals in which the device is turned on).\footnote{We remark that this definition for paths allows (in an extension of the present language) consideration of devices in which some stored values are preserved while turned off.}
 \item $\decay \in \mathbb{R}^+$ models a timeout expiration after which non-communicating devices are considered ``removed,'' 
 allowing adaptation of the network to device removal and topology changes.
 \item $\neighbour{\eventId}{\eventId'}$ holds if and only if:
 \begin{enumerate}
 	\item $\timeF{\eventId'} \in [\timeF{\eventId}-\decay,\timeF{\eventId})$ (i.e. $\eventId'$ happened in the time interval of size $\decay$ before $\eventId$);
 	\item $\PathS(\devF{\eventId})$ is defined in the interval $[\timeF{\eventId'}, \timeF{\eventId}]$ (i.e. $\devF{\eventId}$ was constantly turned on during the time between events $\eventId'$ and $\eventId$);
 	\item $\neighbour{\PathS(\devF{\eventId})(\timeF{\eventId'})}{\PathS(\devF{\eventId'})(\timeF{\eventId'})}$ holds (i.e. the two devices where neighbours when $\eventId'$ happened);
 	\item there exists no further event $\eventId''$ with $\devF{\eventId''} = \devF{\eventId'}$ and $\timeF{\eventId''} > \timeF{\eventId'}$ satisfying the above conditions (i.e. $\eventId'$ is the last firing of $\devF{\eventId'}$ recorded by $\devF{\eventId}$ before $\eventId$).
 \end{enumerate} 
\end{itemize}
\end{example}

\subsection{Denotational semantics of types}

A necessary preliminary step in the definition of denotational semantics for the field calculus is to clarify the denotation of types.
As usual, the denotation of a type gives a set over which the denotation of expressions that are given that type range---denotation of expressions will be presented in next section.
The denotational semantics of a type is given by two intertwined functions: a function $\denotval{\cdot}$ mapping a type $\type$ without type variables to a set of local value denotations (i.e. values at individual devices), and a function $\denottype{\cdot}$ mapping $\type$ to a set of \emph{field evolutions}, ranged over by meta-variable $\dvalue$, assigning local values to every device in each firing event.

If $\builtintype$ is a built-in local type, we assume that $\denotval{\builtintype}$ is given. For derived types, $\denotval{\cdot}$ and $\denottype{\cdot}$ are altogether defined by rules:
\[\begin{array}{rcl}
\denottype{\type} & = & \EventS \pto \denotval{\type} \\
\denotval{\ftypeOf{\ltype}} & = & \deviceIdSet \pto \denotval{\ltype}\\
\denotval{(\type_1,\ldots,\type_n)\rightarrow\type} & = & \setFS \times (\denottype{\type_1}\times\ldots\times\denottype{\type_n}) \rightarrow \denottype{\type}
\end{array}\]
where $\EventS \pto \denotval{\type}$ (resp. $\deviceIdSet \pto \denotval{\type}$) is the set of partial functions from $\EventS$ (resp. $\deviceIdSet$) to $\denotval{\type}$, and $\setFS$ is a set of function tags uniquely characterizing each function.

The denotation of a type $\type$ is a set of field evolutions, that is, partial maps from events $\EventS$ to local value denotations $\denotval{\type}$. This reflects the fact that an expression $\e$ evaluates to (possibly different) local values in each device and event of the computation.

The local value denotation of a field type $\ftypeOf{\ltype}$ is the set of partial functions from devices to local value denotations, which are intended to map a neighbourhood (or an ``aligned subset'' of it: in both cases, a subset of $\deviceIdSet$) to local value denotations of the corresponding local type.

The denotation of a function $(\type_1,\ldots,\type_n)\rightarrow\type$ is instead a set of pairs with the following two components:

\begin{itemize}
  \item The function tag in $\setFS$ (e.g. a syntactic function value as in Figure \ref{fig:values}), needed in order to reflect the choice to compare functions by \emph{syntactic} equality instead of \emph{semantic} equality, which would not allow a computable operational semantics (see Remark~\ref{rmk:equality}). In fact, the presence of such tags is used to grant that two differently specified but identically behaving functions $\funvalue, \funvalue'$ get distinct denotations.

  \item A mapping from input field evolutions in $\denottype{\type_1}, \ldots \denottype{\type_n}$ to an output field evolution in $\denottype{\type}$.
\end{itemize}

The local execution environment under which the computation of the function is assumed to happen is implicitly determined as the (common) domain of its input field evolutions; and the same domain will be retained for the output. 
This environment can influence the outcome of the computation through $\repK$ and $\nbrK$ statements and through non-pure built-in functions.
  
Since local denotational values are not connected to specific events or domains, the common domain of the input field evolutions can be any subset of $\EventS$. In particular, this fact implies that a field evolution $\dvalue$ of function type and domain $\eventS$ is built of functions $\sndK(\dvalue(\eventId))$ which can take arguments of arbitrary domain, \emph{including domains $\eventS' \nsubseteq \eventS$}. Notice that this property grants that a local denotational function value can be meaningfully moved around devices (through operators $\nbrK$, $\repK$).
  
Notice that the definition of $\denotval{(\type_1,\ldots,\type_n)\rightarrow\type}$ by means of a function on whole field evolutions instead of local denotational values is required by the nature of the basic blocks of the language ($\nbrK$, $\repK$), which cannot be computed pointwise event by event. We also remark that the denotation of a function type consists of \emph{total} functions: this reflects the assumption that every function call is guaranteed to terminate (see Remark~\ref{rmk:termination}).

In the remainder of this paper, we use $\denotf{\xname \in D}{\funvalue}$ to denote the mathematical function with domain $D$ assigning each $\xname \in D$ to the corresponding value of expression $\funvalue$. We use $\proj{\dvalue}{\eventS}$ for the restriction of the field evolution $\dvalue$ to $\eventS$, defined by $\denotf{\eventId \in \eventS}{} \dvalue(\eventId)$ for denotational values $\dvalue$ of local type and by $\denotf{\eventId \in \eventS}{} \proj{\dvalue(\eventId)}{\predevices{\eventS}(\eventId)}$ for denotational values of field type. Whenever a sequence of field evolutions $\overline\dvalue$ is assumed to share a common domain, we use $\domof{\overline\dvalue}$ with abuse of notation to denote their common domain.

Notice that we have not given an interpretation for parametric types containing free type variables $\overline{\tvar}, \overline{\ltvar}, \overline{\rtvar}, \overline{\stvar}$, even though these types are contemplated in the present system. Since the denotational semantics of parametric types is a well-understood topic (see e.g. \cite{Ohori:Polymorphism}) and it is entirely orthogonal to the core semantics of field computations, we prefer for sake of clarity to give the definitions only for monomorphic types---extending those definition to polymorphic types would lead to a much heavier, though straightforward extension.\footnote{In \cite{Ohori:Polymorphism} this is achieved by setting $\denottype{\forall\overline{\tvar}\overline{\ltvar}\overline{\rtvar}\overline{\stvar}.\type} = \prod_{\overline\tvar \in \mathbf{T}} \prod_{\overline\ltvar \in \mathbf{L}} \prod_{\overline\rtvar \in \mathbf{R}} \prod_{\overline\ltvar \in \mathbf{S}} \denottype{\type}$ where $\mathbf{T}$, $\mathbf{L}$, $\mathbf{R}$, $\mathbf{S}$ are the involved set of types. In other words, the elements of $\denottype{\forall\overline{\tvar}\overline{\ltvar}\overline{\rtvar}\overline{\stvar}.\type}$ are functions mapping all possible concrete instantiations of the type parameters to the denotational function values of the corresponding type.}

\subsection{Denotational semantics of expressions}

The denotational semantics of a well-typed expression $\e$ of type $\type$ in domain $\eventS$ under assumptions $\VarS = \envmap{\overline{\xname}}{\overline{\dvalue}}$ is written $\denotexp{\e}{\VarS}{\eventS}$ and yields a field evolution in $\denottype{\type}$ with domain $\eventS$. As for the denotation of types, we assume that the denotations of built-in functions and constructors are given. In particular, this is represented by the function $\builtindenot{C}{\dc}$ in $(\denotval{\ltype_1}\times\ldots\times\denotval{\ltype_n}) \rightarrow \denotval{\ltype}$ translating the behaviour of built-in constructors $\dc$ of type $(\overline{\ltype})\to\ltype$;\footnote{Since a constructor does not depend on the environment, we do not need an element of $\denotval{(\overline{\ltype}) \to \ltype}$ in this case.} and by the function $\builtindenot{B}{\oname}$ in $\denottype{(\overline{\type}) \to \type}$ translating the behaviour of built-in operators $\oname$ of type $(\overline{\type}) \to \type$ to denotational values (and possibly implicitly depending on sensor values and global environment status).

The interpretation function $\denotexp{\cdot}{}{}$ is then defined by the following rules:
\[
\begin{array}{rcl}
	\denotexp{\xname}{\VarS}{\eventS} & = & \proj{\VarS(\xname)}{\eventS} \\[6pt]

	\denotexp{\envmap{\overline\deviceId}{\overline\anyvalue}}{\VarS}{\eventS} & = & \denotf{\eventId \in \eventS}{} \envmap{\overline\deviceId}{\denotexp{\overline\anyvalue}{\VarS}{\eventS}(\eventId)} \\[6pt]

	\denotexp{\dcOf{\dc}{\overline\lvalue}}{\VarS}{\eventS} & = & \denotf{\eventId \in \eventS}{} \builtindenot{C}{\dc}(\denotexp{\overline\lvalue}{\VarS}{\eventS}(\eventId)) \\[6pt]

	\denotexp{\oname}{\VarS}{\eventS} & = & \denotf{\eventId \in \eventS}{} \langle \oname, ~ \builtindenot{B}{\oname} \rangle \\[6pt]

	\denotexp{\fname}{\VarS}{\eventS} & = & \denotf{\eventId \in \eventS}{} \langle \fname, ~ \lim_{n} \builtindenot{D}{\fname}_n \rangle \\[6pt]

	\denotexp{(\overline{\xname}) \; \toSymK\;\e}{\VarS}{\eventS} & = & \denotf{\eventId \in \eventS}{} \langle \texttt{(}\overline{\xname}\texttt{)} \; \toSymK\;\e, ~ \denotf{\overline\dvalue \in \denottype{\overline\type}}{} \denotexp{\e}{\VarS \cup \envmap{\overline{\xname}}{\overline\dvalue}}{\domof{\overline\dvalue}} \rangle \\[6pt]

	\denotexp{\e'(\overline\e)}{\VarS}{\eventS} & = & \denotf{\eventId \in \eventS}{} \sndK\left(\denotexp{\e'}{\VarS}{\eventS}(\eventId)\right)\left(\proj{\denotexp{\overline\e}{\VarS}{\eventS}}{\eventS(\e',\eventId)}\right)(\eventId) \\[6pt]

	\denotexp{\nbrK\{\e\}}{\VarS}{\eventS} & = & \denotf{\eventId \in \eventS}{} \denotf{\deviceId \in \predevices{\eventS}(\eventId)}{} \denotexp{\e}{\VarS}{\eventS}(\nbrdevice{\eventId}{\deviceId}) \\[6pt]

	\denotexp{\repK(\e_1)\{(\xname) \; \toSymK\; \e_2\}}{\VarS}{\eventS} & = & \lim_n \builtindenot{R}{\repK(\e_1)\{(\xname) \; \toSymK\; \e_2\}}_n
\end{array}
\]
Where:
\begin{itemize}
	\item
	$\builtindenot{D}{\fname}_n$ is the partial function translating the behaviour of $\fname$ when recursion is bounded to depth $n$, and is defined by rules:
	\[
	\begin{array}{lcl}
		\builtindenot{D}{\fname}_0 & = & \emptyset \\
		\builtindenot{D}{\fname}_{n+1} & = & \denotf{\overline\dvalue \in \denottype{\overline\type}}{} \denotexp{\body{\fname}}{\VarS \cup \envmap{\args{\fname}}{\overline\dvalue}, \envmap{\fname}{\denotf{\eventId \in \eventS}{} \builtindenot{D}{\fname}_n}}{\domof{\overline\dvalue}}
	\end{array}
	\]
	
	\item
	$\eventS(\e', \eventId)$ is equal to $\eventS$ if $\denotexp{\e'}{\VarS}{\eventS}(\eventId)$ is a built-in operator, and to $\{\eventId' \in \eventS: ~ \denotexp{\e'}{\VarS}{\eventS}(\eventId') = \denotexp{\e'}{\VarS}{\eventS}(\eventId)\}$ otherwise (that is, the set of events in $\eventS$ aligned with $\eventId$ with respect to the computation of $\e'$).
	
	\item
	$\builtindenot{R}{\e}_n$ with $\e = \repK(\e_1)\{(\xname) \; \toSymK\; \e_2\}$ denotes the $\repK$ construct as bounded to $n$ loop steps, and it is defined by rules:
	\[
	\begin{array}{lcl}
		\builtindenot{R}{\e}_0 & = & \denotexp{\e_1}{\VarS}{\eventS} \\
		\builtindenot{R}{\e}_{n+1} & = & \denotexp{\e_2}{\VarS \cup \envmap{\xname}{\shift{\builtindenot{R}{\e}_n, \builtindenot{R}{\e}_0}}}{\eventS}
	\end{array}
	\]
	where $\shift{\dvalue, \dvalue_0} = \denotf{\eventId \in \eventS}{} \dvalue(\repdevice{\eventId}) \ifK \; \repdevice{\eventId} \text{ exists } \elseK \; \dvalue_0(\eventId)$ pushes each value in $\dvalue$ to the next future event, while falling back to $\dvalue_0$ for starting events. In the remainder of this paper, for ease of notation we shall drop the reference to $\dvalue_0$ and keep it implicit in the definition of $\shift{\cdot}$.
\end{itemize}

The rules above provide a definition of $\denotexp{\cdot}{}{}$ by induction on the structure of the expressions. In the remainder of this paper we shall feel free to omit the subscript $\VarS$ whenever $\VarS = \emptyset$ and the superscript $\eventS$ whenever $\eventS = \EventS$. Notice that syntactic values are always denoted by constant field evolutions, and can be reconstructed from their denotation (with possibly the exception of constructors).

The denotation of variables is straightforward, while the denotation of constructors and built-in operators is abstracted away assuming that corresponding $\builtindenot{C}{\dc}$ and $\builtindenot{B}{\oname}$ are given. In order to produce neighbouring field values with the correct domain, we require that in case the return type $\type$ of $\oname$ is a \emph{field} type, then $\domof{\builtindenot{B}{\oname}(\overline{\dvalue})(\eventId)} = \predevices{\eventS}(\eventId)$ for all possible $\overline{\dvalue}$ in $\denottype{\overline{\type}}$ of domain $\eventS$ and $\eventId \in \eventS$. Even though most built-in operators (pure operators, local sensors) could be defined pointwise in the same way constructors are defined, this is not possible for relational sensors (as \texttt{nbr-range}) thus we opted for a more general and simpler formulation. The denotation of neighbouring field values is given for convenience, but since neighbouring field values does not occur in source programs is not needed for their denotation.

The denotation of defined functions, as usual, is defined as a \emph{fixpoint} of an iterated process starting from the function $\builtindenot{D}{\fname}_0$ with empty domain. At each subsequent step $n+1$, the body of $\fname$ is evaluated with respect to the context that associates the name $\fname$ itself with the previously obtained function $\builtindenot{D}{\fname}_n$ (and the arguments of $\fname$ with the respective values). We assume that the resulting function $\builtindenot{D}{\fname}_{n+1}$ is undefined if it calls $\builtindenot{D}{\fname}_n$ with arguments outside of its domain. It follows by easy induction that each such step is a conservative extension, i.e., $\builtindenot{D}{\fname}_n \subseteq \builtindenot{D}{\fname}_{n+1}$, hence the limit of the process is well-defined. 
Since function calls are guaranteed to terminate, this limit will be a total function as required by the denotation of function types (see Remark~\ref{rmk:termination}).

The denotation of a function application $\e'(\overline\e)$ is given pointwise by event, and applies the second coordinate of $\denotexp{\e'}{\VarS}{\eventS}(\eventId)$ (that is, the mathematical function corresponding to $\e'$) interpreted in the restricted domain $\eventS(\e', \eventId)$ (computed through the first coordinate of $\denotexp{\e'}{\VarS}{\eventS}$) containing $\eventId$ to the arguments $\denotexp{\overline\e}{\VarS}{\eventS}$. Such domain restriction is used in order to prevent interference among non-aligned devices: in fact, no restriction is needed for built-in operators while restriction to devices computing the same function $\e'$ is needed otherwise. The importance of this aspect shall be further clarified in the following sections. As the rule above is formulated, it seems that a whole field evolution $\dvalue$ is calculated for each event $\eventId$, while being used only to produce the local value $\dvalue(\eventId)$. However, the whole field evolution is actually used since each event in its domain $\eventS(\e', \eventId)$ computes the same function on the same arguments, hence producing the same output field evolution. Thus the rule could also be reformulated as follows:
\[
\denotexp{\e'(\overline\e)}{\VarS}{\eventS} = \bigcup_{\eventId \in \eventS} \sndK\left(\denotexp{\e'}{\VarS}{\eventS}(\eventId)\right)\left(\proj{\denotexp{\overline\e}{\VarS}{\eventS}}{\eventS(\e',\eventId)}\right) \\[5pt]
\]

The denotation of operator $\nbrK$ yields in each event $\eventId$ a neighbouring field of domain $\predevices{\eventS}(\eventId)$ mapping to the values of expression $\e$ in the corresponding events.

The denotation of operator $\repK$ is carried out by a fixpoint process as for recursive functions. First, a field evolution $\builtindenot{R}{\cdot}_0$ is computed holding the initial values computed by $\e_1$ in each event. At each subsequent step, the results computed by $\builtindenot{R}{\cdot}_n = \dvalue$ in each event are made available to their subsequent events through the new assumption $\envmap{\xname}{\shift{\dvalue}}$ in $\VarS$. It follows that once the value at each event in $\predevices{\eventS}(\eventId)$ stabilizes, the value at $\eventId$ also stabilizes in one more iteration. Since the events form a DAG and values at source events (events without predecessor) are steadily equal to the initial value by construction, the whole process stabilizes after a number of iterations at most equal to the cardinality of $\EventS$, hence the limit of the process is well-defined.

\subsection{Example}

We now illustrate the denotational semantics by applying it to representative example expressions.

\begin{figure}[p]
	\centering
	\vspace{-10pt}
	\includegraphics[height=1.00\textheight]{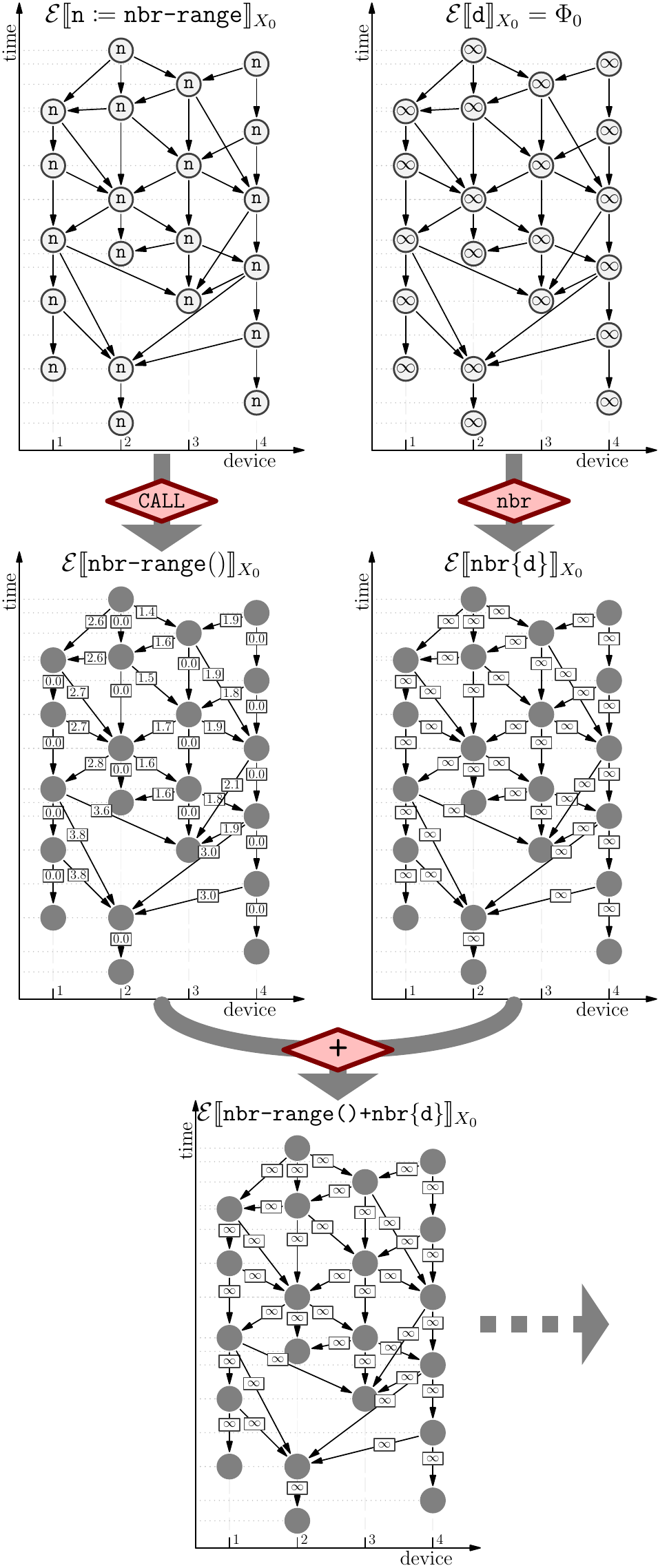}
	\vspace{-4pt}
	\caption{(a) The denotational semantics of a distance-to calculation.}\label{f:gradient}
\end{figure}
\begin{figure}[p]
	\ContinuedFloat
	\centering
	\vspace{-10pt}
	\includegraphics[height=1.00\textheight]{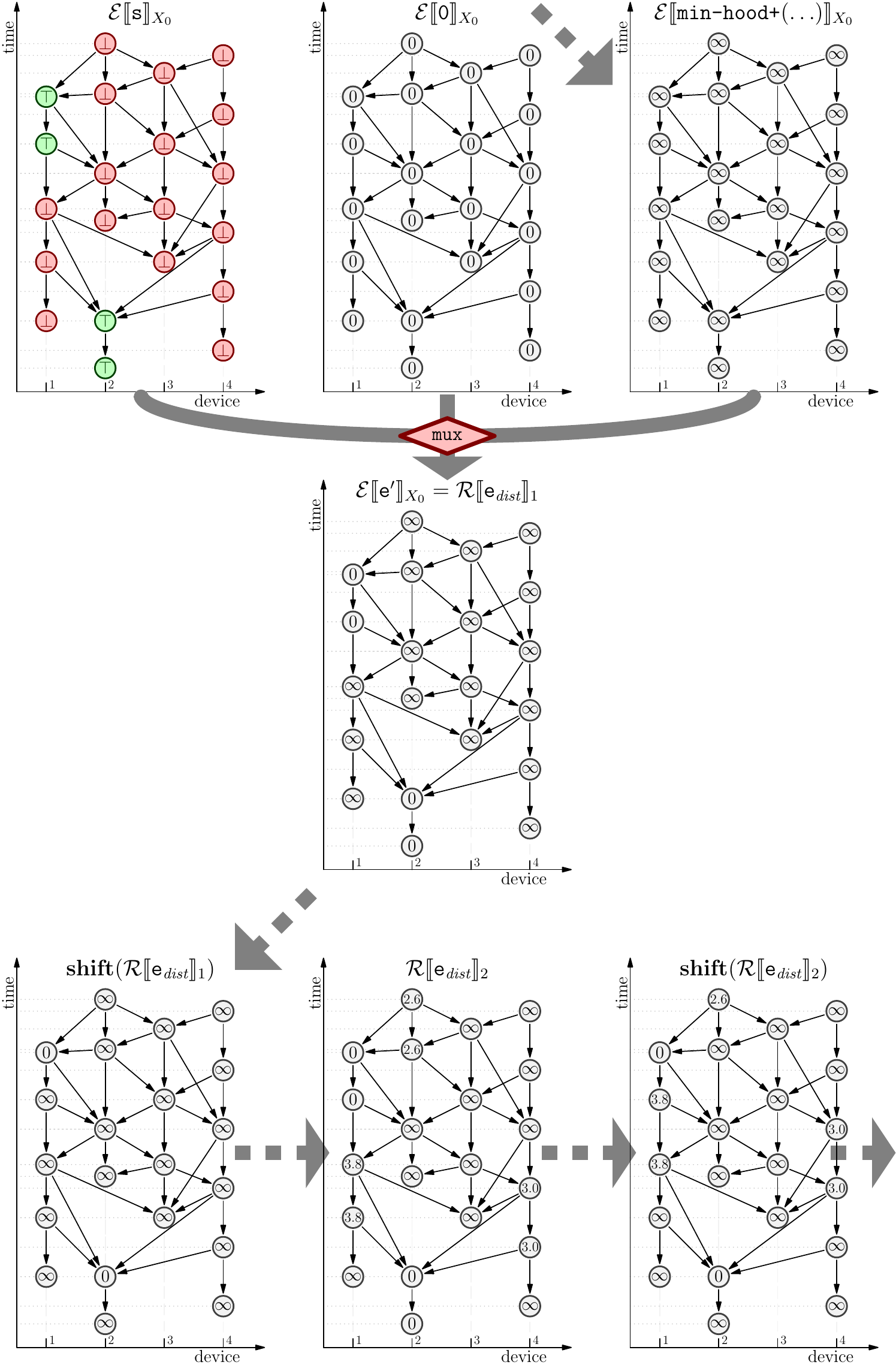}
	\vspace{-4pt}
	\caption{(b) The denotational semantics of a distance-to calculation (cont.)}
\end{figure}
\begin{figure}[t]
	\ContinuedFloat
	\centering
	\vspace{-10pt}
	\includegraphics[width=0.99\textwidth]{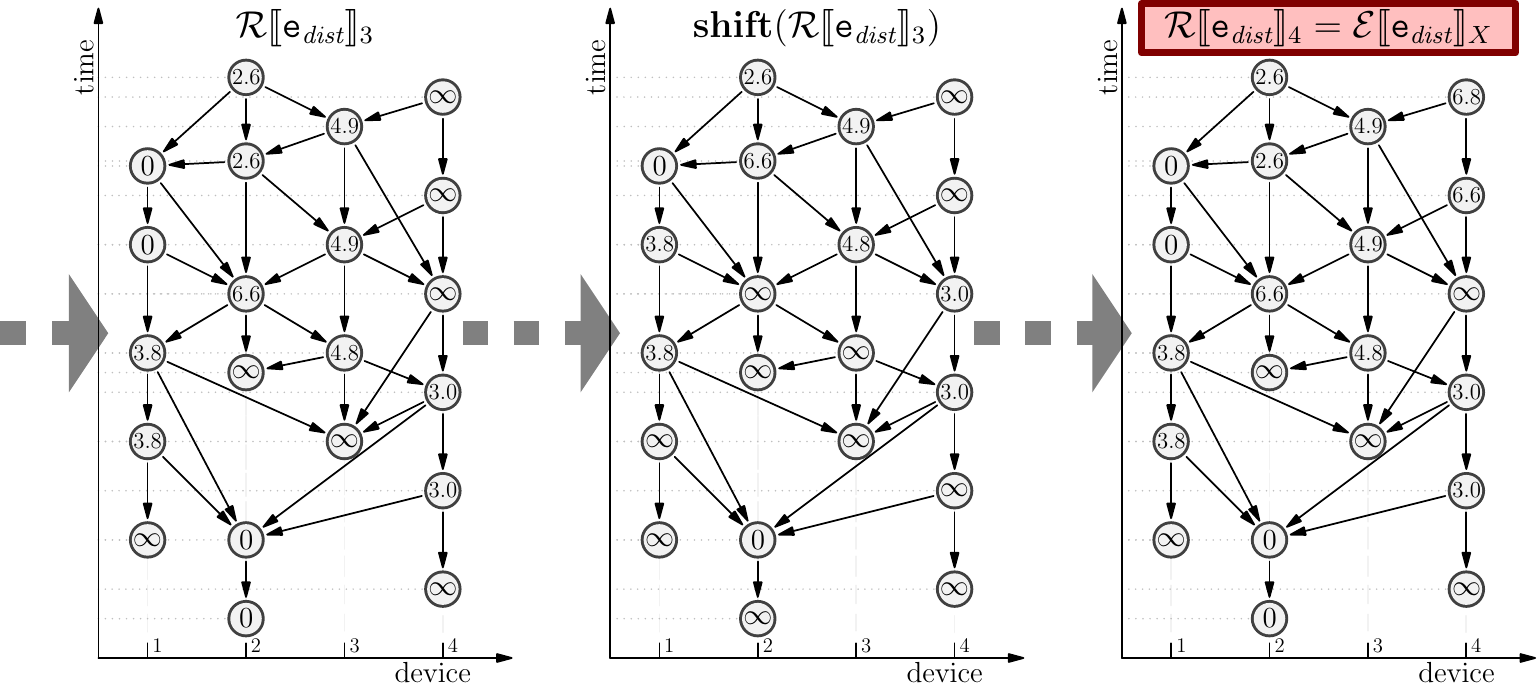}
	\vspace{-4pt}
	\caption{(c) The denotational semantics of a distance-to calculation (cont.)}
\end{figure}

\subsubsection*{Distance-To}

For a first example, consider the expression $\e_\textit{dist}$, computing the distance of every device from devices in a given source set indicated by the Boolean-valued field \texttt{s}:
\begin{Verbatim}[%fontsize=\footnotesize,
	                  frame=single,
	                  commandchars=\\\{\}]
\pr{mux}( s, 0, \pr{min-hood+}(\pr{nbr-range}() + \km{nbr}\{d\}) )  \il{\texttt{ e'}}
\km{rep} (infinity) \{ (d) => e' \}  \il{\texttt{ e\_dist}}
\end{Verbatim}
where \texttt{min-hood+} is a built-in function that returns the minimum value amongst a device's neighbours, excluding itself.
Figure~\ref{f:gradient} shows the evaluation of the denotational semantics for this expression, as evaluated with respect to the neighbourhood graph shown in Figure~\ref{f:neigh}.
We consider as input a source set \texttt{s} consisting of device 2 before its reboot, and device 1 beginning at its fourth firing,
represented by the Boolean field evolution $\dvalue_\texttt{s}$, with corresponding environment $X = \envmap{\texttt{s}}{\dvalue_\texttt{s}}$, shown in the top left of Figure \ref{f:gradient} (b). We assume the devices to be moving\footnote{Note that the $x$-axis in Figure \ref{f:gradient} is indexed by \emph{device} and not by \emph{position}.} so that their relative distance changes over time as depicted in the center left of Figure \ref{f:gradient} (a).

The outermost component of expression $\e_\textit{dist}$ is a $\repK$ operator, thus $\denotexp{\e_\textit{dist}}{X}{}$ is calculated via the following procedure:
\begin{itemize}
	\item First, $\builtindenot{R}{\e_\textit{dist}}_0$ is calculated as $\denotexp{\infty}{X}{} = \dvalue_0$ (since $\infty$ is the initial value of the $\repK$-expression), a constant field evolution.

	\item This value is then shifted in time (in this case leaving it unchanged, $\shift{\dvalue_0} = \dvalue_0$) and incorporated in the substitution $X_0 = X \cup \envmap{\texttt{d}}{\dvalue_0}$. Thus $\builtindenot{R}{\e_\textit{dist}}_1$ is calculated as $\denotexp{\e'}{X_0}{}$ giving a new field evolution $\dvalue_1$. This evaluation is illustrated step-by-step in Figure \ref{f:gradient} (a) and (b) top and center, breaking $\e'$ into all its subexpressions.

	\item The process of shift and evaluation is then repeated: $\dvalue_1$ is shifted in time (Figure \ref{f:gradient} (b) bottom left), incorporated in a new substitution $X_1 = X \cup \envmap{\texttt{d}}{\shift{\dvalue_1}}$ and $\denotexp{\e'}{X_1}{}$ is expanded into another field evolution $\dvalue_2$ (Figure \ref{f:gradient} (b) bottom center).  Shifting and evaluation continues until a fixed point is reached, which in the case of this example happens at stage $n=4$ (Figure \ref{f:gradient} (c) right).
\end{itemize}

Notice that due to the characteristics of the $\repK$ operator, the values $\nbrK\{\texttt{d}\}$ collected from neighbour devices are not the latest but instead the ones before them (that is, the values fed to the update function in the latest event). 
The latest outcome of a $\repK$ operator could instead be obtained via $\nbrK\{ \repK (\cdot) \{\cdot\} \}$, but this construct is not used here as that would not allow the distance calculation to propagate across multiple hops in the network.
This additional delay sometimes leads to counterintuitive behaviours: for example, the third firing $\eventId$ of device 3 calculates gradient $4.9$ obtained through device 4, which however holds value $\infty$ in its latest firing available to $\eventId$. In fact, the value to which $\eventId$ refers to is the previous one, which is equal to $3.0$. This behaviour can slow down the propagation of updates through a network, but appears necessary for ensuring both safe and general composition.

\subsubsection*{Distance avoiding obstacles}

\begin{figure}[p]
	\centering
	\vspace{-10pt}
	\includegraphics[height=1.00\textheight]{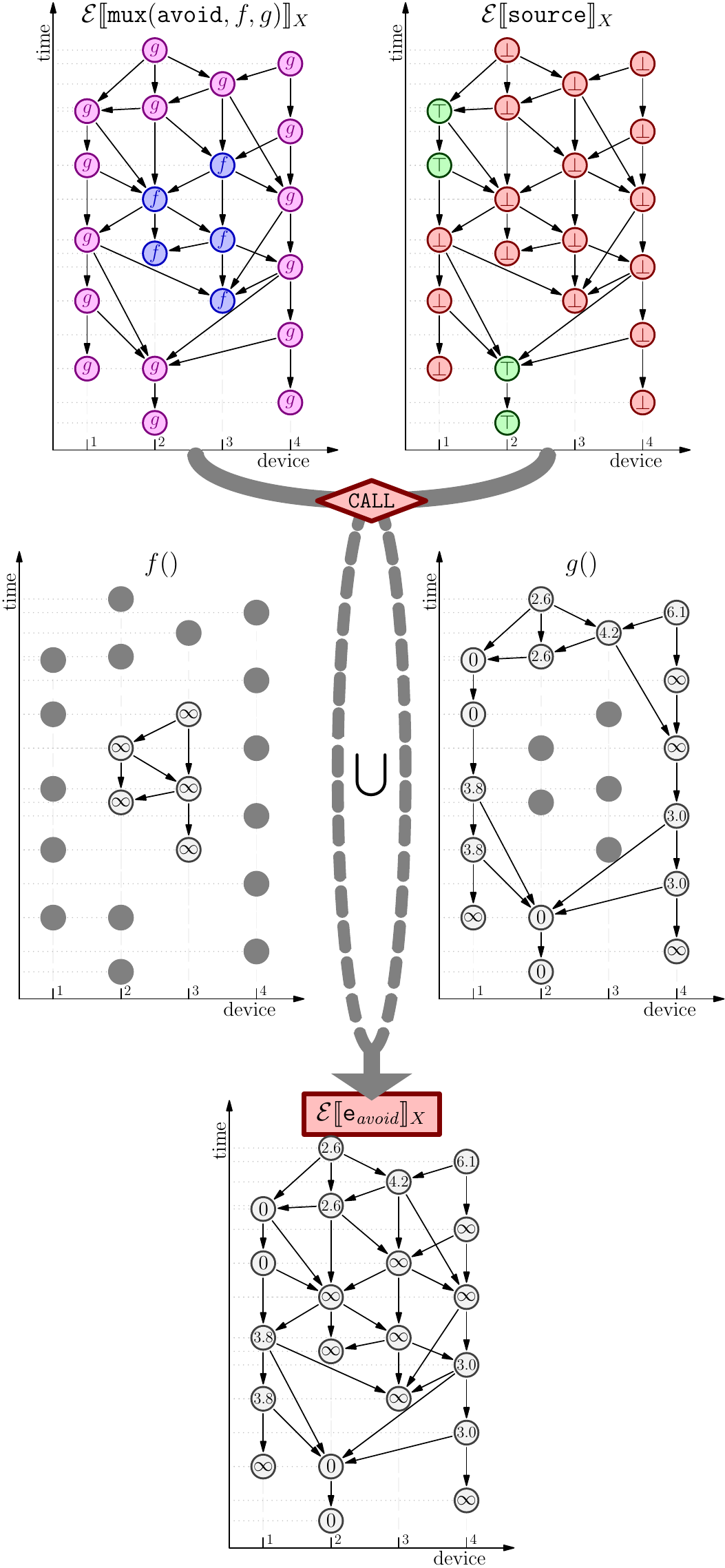}
	\vspace{-4pt}
	\caption{The denotational semantics of a gradient calculation avoiding obstacles.} \label{f:avoid}
\end{figure}

The previous example allowed us to show the denotation of data values (\texttt{nbr-range}, $0$), variable lookups (\texttt{d}, \texttt{s}),  builtin functions (\texttt{nbr-range}, \texttt{+[f,f]}, \texttt{mux}, \texttt{min-hood+}), $\nbrK$ and $\repK$ constructs. It did not, however, show an example of \emph{branching}, which in the present calculi is modeled by function calls $\e(\overline\e)$ where the functional expression $\e$ is \emph{not} constant.

To illustrate branching, we now expand on the previous example by considering the expression $\e_\textit{avoid}$ which computes the distance of each device from a given source set \emph{avoiding some obstacles}:
\begin{Verbatim}[%fontsize=\footnotesize,
	                  frame=single,
	                  commandchars=\\\{\}]
\km{def} f(s) \{ \pr{infinity} \}
\km{def} g(s) \{ e_dist \}
\texttt{e\_avoid} \pr{mux}(avoid, f, g)(source)  \il{\texttt{e\_avoid}}
\end{Verbatim}
The denotational semantics of $\e_\textit{avoid}$ on the sample network in Figure \ref{f:neigh} is shown in Figure \ref{f:avoid}. We consider as input the same source set \texttt{source} (with corresponding field evolution $\dvalue_\texttt{s}$) as in the previous example, and a set of obstacles \texttt{avoid} corresponding to the first firings of device 3 and device 2 after its reboot (blue nodes in Figure \ref{f:avoid} top left), together enclosed in environment $X$.

When the field of functions (top left) is called on the argument (top right), the computation branches in two parts:
\begin{itemize}
	\item the events holding $f$, which just compute the constant value $\infty$ (center left);
	\item the events holding $g$, which compute a distance from the source set, as in the previous example (center right).
\end{itemize}
Notice that since the events holding $g$ compute their distances in isolation, their final values differ from the ones obtained in the previous example.

Finally, the two different branches are merged together in order to form the final outcome of the function call, in Figure \ref{f:avoid} bottom.

\subsection{Properties}\label{ssec:properties}

The denotational semantics can be used in order to formulate conveniently intuitive facts about the calculus.

\subsubsection*{Alignment}

Given any expression $\e$ of field type, $\denotexp{\e}{}{\eventS}(\eventId)$ is a neighbouring field with domain $\predevices{\eventS}(\eventId)$. This fact can be easily checked in the rule for $\nbrK$. In the case of function application, notice that $\e'$ has to be a value since it has type $(\overline{\type}) \to \ftype$. Thus $\eventS(\e', \eventId) = \eventS$ and the thesis follows by inductive hypothesis using $\nbrK$ and built-in operators as base case.

\subsubsection*{Restriction}

Given any expression $\e_0$ executed in domain $\eventS$ and $V = \denotexp{\e_0}{}{\eventS}(\eventId)$ for some $\eventId \in \eventS$, we say that $({\denotexp{\e_0}{}{\eventS}})^{-1}(V) = \eventS(\e_0, \eventId)$ is a cluster.
We say that function call has the restriction property to mean that well-typed expression $\e_0(\e_1,\ldots,\e_n)$ computes in isolation in each such cluster.

Namely, given $\e_0$ and any of its clusters $\eventS(\e_0, \eventId)$, let $\e_1, \ldots, \e_n$ and $\e'_1, \ldots, \e'_n$ be such that the denotation of $\e_i$ coincides with that of $\e'_i$ on $\eventS(\e_0, \eventId)$, that is $\proj{\denotexp{\e_i}{}{\eventS}}{\eventS(\e_0, \eventId)}  = \proj{\denotexp{\e'_i}{}{\eventS}}{\eventS(\e_0, \eventId)}$. Then $\denotexp{\e_0(\e_1,\ldots,\e_n)}{}{\eventS}(\eventId)=\denotexp{\e_0(\e'_1,\ldots,\e'_n)}{}{\eventS}(\eventId)$ for any $\eventId \in \eventS(\e_0, \eventId)$.

This means that computation of $\e_0$ inside cluster $\eventS(\e_0, \eventId)$ is independent of the outcome of computation outside it, and of events outside it. This fact reflects the intuition beyond function application given in Section \ref{sec-calculus-syntax-and-typing}, and implies the analogous property for the first-order calculus with conditionals (since the conditional expression can be simulated by function application).

\subsubsection*{Repeating statements}

Consider now an expression $\e = \repK(\e_1)\{(\xname) \; \toSymK\; \e_2\}$ and suppose that $\eventId$ is a ``source'' event in $\eventS$, that is, there exists no $\eventId'$ in $\eventS$ such that $\neighbour{\eventId}{\eventId'}$. Then $\denotexp{\e}{}{\eventS}(\eventId) = \denotexp{((\xname) \; \toSymK \; \e_2)(\e_1)}{}{\eventS}$.

Furthermore, assume now that $\e$ does not contain $\nbrK$ statements or non-pure operators and $\eventId$ is any event in $\eventS$. Then $\denotexp{\e}{}{\eventS}(\eventId) = \denotexp{\e}{}{\eventS'}(\eventId)$ where $\eventS'$ is the subset of $\eventS$ containing only the events on device $\deviceId_\eventId$.

\section{Operational Semantics}\label{sec-calculus-operational-semantics}

Having established a denotational semantics describing the behavior of the aggregate of devices in the previous section, 
we now develop an operational semantics that describes the equivalent computation as carried out by an individual device at a particular event, and hence, by a whole network of devices over time.
In particular, this section presents a formal semantics that can serve as a specification for implementation of programming languages based on the calculus.

\newcommand{\piIofOv}[1]{\overline{\pi}(#1)}

\subsection{Device Semantics (Big-Step Operational Semantics)}\label{sec-calculus-device-semantics}

According to the ``local'' viewpoint, individual devices undergo computation
in rounds.
In each round, a device sleeps for some time, wakes up,
gathers information about messages received from neighbours while sleeping, performs an evaluation of the program, and finally emits a message to all
neighbours with information about the outcome of computation before
going back to sleep.
The scheduling of such rounds across the network is fair and
non-synchronous.

We base operational semantics on the syntax introduced in Section~\ref{sec-calculus-syntax} (Figure~\ref{fig:source:syntax}),
%
%
To simplify the notation, we shall assume a fixed program $\PROGRAM$. 
We say that ``device $\deviceId$ \emph{fires}'', to mean that the main expression of $\PROGRAM$ is evaluated on $\deviceId$ at a particular round.

We model device computation by a big-step 
operational semantics where the result of evaluation is a \emph{value-tree} $\vtree$, which is an ordered tree of values that tracks the results of all evaluated subexpressions.
Intuitively, the evaluation of an expression at a given time in a device $\deviceId$ is performed against the recently-received
value-trees of neighbours, namely, its outcome depends on those value-trees.
The result is a new value-tree that is conversely made available to $\deviceId$'s neighbours (through a broadcast) for their
firing; this includes $\deviceId$ itself, so as to support a form of state across computation rounds (note that any implementation might massively compress the value-tree, storing only informations about sub-expressions which are relevant for the computation).
A \emph{value-tree environment} $\Trees$ is a map from device identifiers  to value-trees, collecting the outcome of the last evaluation on the neighbours.
This is written $\envmap{\overline\deviceId}{\overline\vtree}$ as short for $\envmap{\deviceId_1}{\vtree_1},\ldots,\envmap{\deviceId_n}{\vtree_n}$.

The syntax of value-trees and value-tree environments is given in Figure~\ref{fig:deviceSemantics} (first frame).
Figure~\ref{fig:deviceSemantics} (second frame) defines: the auxiliary functions $\rho$ and $\pi$ for extracting the root value and a subtree of a value-tree, respectively (further explanations about function $\pi$ will be given later); the extension of functions $\rho$ and $\pi$ to value-tree environments; and the auxiliary functions $\argsNAME$ and $\bodyNAME$ for extracting the formal parameters and the body of a (user-defined or anonymous) function, respectively.
The computation that takes place on a single device is formalised  by the big-step operational semantics rules given in
Figure~\ref{fig:deviceSemantics} (fourth frame). The derived judgements are of the form $\bsopsem{\deviceId}{\Trees}{\e}{\vtree}$,
to be read ``expression $\e$ evaluates to  value-tree $\vtree$
on device $\deviceId$  with respect to\ the value-tree environment
$\Trees$'', where:
\emph{(i)} $\deviceId$ is the identifier of the current device;
\emph{(ii)} $\Trees$ is the neighbouring field of the value-trees produced by the most recent evaluation of (an expression corresponding to) $\e$ on 
$\deviceId$'s neighbours;
\emph{(iii)}
 $\e$ is a run-time expression (i.e., an expression that may contain neighbouring field values);
\emph{(iv)} the value-tree $\vtree$  represents the values
    computed for all the expressions encountered during the
    evaluation of $\e$---in particular $\vrootOf{\vtree}$
    is the resulting value of expression $\e$. 

\begin{figure}[!t]{
 \framebox[1\textwidth]{
 $\begin{array}{l}
\textbf{Value-trees  
and value-tree environments:}\\
\begin{array}{rcl@{\hspace{8.5cm}}r}
\vtree
& \BNFcce &  \mkvtree{E-Rule}{\anyvalue}{\overline{\vtree}}   
 &   {\footnotesize \mbox{value-tree}} \\
%
%
\Trees & \BNFcce & \envmap{\overline{\deviceId}}{\overline{\vtree}}  
 &   {\footnotesize \mbox{value-tree environment}} \\
%
\end{array}\\
\hline\\[-8pt]
\textbf{Auxiliary functions:}\\
\begin{array}{l}
\begin{array}{l@{\hspace{1.4cm}}l}
\vrootOf{\mkvtree{E-Rule}{\anyvalue}{\overline{\vtree}}}  =   \anyvalue
&
\\
\piIof{i}{\mkvtree{E-Rule}{\anyvalue}{\vtree_1,\ldots,\vtree_n}}  =   \vtree_i
\quad \mbox{if} \; 1\le i \le n
                         &    
\piBof{\lvalue}{\mkvtree{E-Rule}{\anyvalue}{\vtree_1,\ldots,\vtree_{n+2}}}  =   \vtree_{n+2}
\quad \mbox{if} \;  \vrootOf{\vtree_{n+1}} = \lvalue
\\

\piIof{i}{\vtree}  =   \emptyseq \quad \mbox{otherwise} 
                         &  
 \piBof{\lvalue}{\vtree}  =   \emptyseq \quad \mbox{otherwise}
\\  
\end{array}
\\
\mbox{For } \auxNAME\in\rho,\piI{i},\piB{\lvalue}: 
\quad 
\left\{\begin{array}{lcll}
 \aux{\envmap{\deviceId}{\vtree}}  & =  & \envmap{\deviceId}{\aux{\vtree}} & \quad \mbox{if} \; \aux{\vtree} \not=\emptyseq  
\\
\aux{\envmap{\deviceId}{\vtree}}  & =   & \emptyseq  & \quad \mbox{if} \; \aux{\vtree}=\emptyseq  
\\
\aux{\Trees,\Trees'}  & =  &  \aux{\Trees},\aux{\Trees'}
\end{array}\right.   
\\
\begin{array}{l@{\hspace{2.28cm}}l}
\args{\fname} = \overline{\xname} \quad \mbox{if } \, \defK \; \fname (\overline{\xname}) \; \{\e\}
&
\body{\fname} = \e  \quad \mbox{if } \, \defK \; \fname (\overline{\xname}) \; \{\e\}
\\
\args{ (\overline{\xname}) \; \toSymK \; \e} = \overline{\xname}
&
\body{(\overline{\xname}) \; \toSymK \; \e} = \e
\end{array}
\end{array}\\
\hline\\[-10pt]
\textbf{Syntactic shorthands:}\\
\begin{array}{l@{\hspace{20pt}}l@{\hspace{20pt}}l}
\bsopsem{\deviceId}{\piIofOv{\Trees}}{\overline{\e}}{\overline{\vtree}}
&
  \textrm{where~~} |\overline{\e}|=n
&
  \textrm{for~~}
  \bsopsem{\deviceId}{\piIof{1}{\Trees}}{\e_1}{\vtree_1}
  \quad
    \cdots
    \quad
    \bsopsem{\deviceId}{\piIof{n}{\Trees}}{\e_n}{\vtree_n}\\
\vrootOf{\overline{\vtree}}
&
  \textrm{where~~} |\overline{\vtree}|=n
  & \textrm{for~~}
\vrootOf{\vtree_1},\ldots,\vrootOf{\vtree_n}\\
\substitution{\overline{\xname}}{\vrootOf{\overline{\vtree}}}
&   \textrm{where~~} |\overline{\xname}|=n
  &
  \textrm{for~~}
\substitution{\xname_1}{\vrootOf{\vtree_1}}~\ldots\quad\substitution{\xname_n}{\vrootOf{\vtree_n}}
\end{array}\\
\hline\\[-10pt]
\textbf{Rules for expression evaluation:} \hspace{5.7cm} 
  \boxed{\bsopsem{\deviceId}{\Trees}{\e}{\vtree}}
\skiptransition
\begin{array}{c}
\nullsurfaceTyping{E-LOC}{
\bsopsem{\deviceId}{\Trees}{\lvalue}{\mkvtree{E-LOC}{\lvalue}{}}
}
\qquad\qquad
\surfaceTyping{E-FLD}{\qquad \fvalue' = \proj{\fvalue}{\domof{\Trees}\cup\{\deviceId\}}}{
\bsopsem{\deviceId}{\Trees}{\fvalue}{\mkvtree{E-FLD}{\fvalue'}{}}
}
\skiptransition\\[-3pt]
\surfaceTyping{E-B-APP}{  \quad
\begin{array}{c}
  \bsopsem{\deviceId}{\piIofOv{\Trees}}{\overline{\e},\e}{\overline{\vtree},\vtree}
  \qquad \anyvalue=\builtinop{\vrootOf{\vtree}}{\deviceId}{\Trees}(\vrootOf{\overline{\vtree}})
\end{array}
 }{
\bsopsem{\deviceId}{\Trees}{\e(\overline{\e})}{\mkvtree{E-B-APP}{\anyvalue}{\overline{\vtree},\vtree}}
}
\skiptransition\\
\surfaceTyping{E-D-APP}{ \quad
\begin{array}{c}
  \bsopsem{\deviceId}{\piIofOv{\Trees}}{\overline{\e},\e}{\overline{\vtree},\vtree} \quad 
  \funvalue = \vrootOf{\vtree} \quad
  \bsopsem{\deviceId}{\piBof{\funvalue}{\Trees}}{\applySubstitution{\body{\funvalue}}{\substitution{\args{\funvalue}}{\vrootOf{\overline{\vtree}}}}}{\vtree'}
\end{array}
 }{
\bsopsem{\deviceId}{\Trees}{\e(\overline{\e})}{\mkvtree{E-D-APP}{\vrootOf{\vtree'}}{\overline{\vtree},\vtree,\vtree'}}
}
\skiptransition\\
\surfaceTyping{E-NBR}{
         \qquad
\Trees_1=\piIof{1}{\Trees}
\qquad
     \bsopsem{\deviceId}{\Trees_1}{\e}{\vtree_1}
\qquad
 \fvalue=\mapupdate{\vrootOf{\Trees_{1}}}{\envmap{\deviceId}{\vrootOf{\vtree_1}}}
 }{
\bsopsem{\deviceId}{\Trees}{\nbrK\{\e\}}{\mkvtree{E-NBR}{\fvalue}{\vtree_1}}
}
\skiptransition\\
\surfaceTyping{E-REP}{
	\quad
	\begin{array}{ll}
     \bsopsem{\deviceId}{\piIof{1}{\Trees}}{\e_1}{\vtree_1} & \lvalue_1=\vrootOf{\vtree_{1}}\\
     \bsopsem{\deviceId}{\piIof{2}{\Trees}}{\applySubstitution{\e_2}{\substitution{\xname}{\lvalue_0}}}{\vtree_2}~~& \lvalue_2=\vrootOf{\vtree_{2}}
	\end{array}
	\quad
	\lvalue_0 = \left\{\begin{array}{lc}
                             \vrootOf{\piIof{2}{\Trees}}(\deviceId) & \mbox{if} \;  \deviceId \in \domof{\Trees} \\
                             \lvalue_1 & \mbox{otherwise}
                           \end{array}\right.
 }{
\bsopsem{\deviceId}{\Trees}{\repK(\e_1)\{(\xname) \; \toSymK \; \e_2\}}{\mkvtree{E-REP}{\lvalue_2}{\vtree_1,\vtree_2}}
}
\end{array}
\end{array}$}
} 
\vspace{-0.1cm}
 \caption{Big-step operational semantics for expression evaluation.} \label{fig:deviceSemantics}
\end{figure}

The operational semantics  rules  are based on rather standard rules for functional languages, extended so as to be able to evaluate a subexpression $\e'$ of $\e$ with respect to\ the value-tree environment $\Trees'$ obtained from $\Trees$ by extracting the corresponding subtree (when present) in the value-trees in the range of $\Trees$. This process, called \emph{alignment}, is modeled by the 
 auxiliary function $\pi$, defined in Figure~\ref{fig:deviceSemantics} (third frame). The function $\pi$ has two different behaviours (specified by its subscript or superscript): $\piIof{i}{\vtree}$ extracts the $i$-th subtree of $\vtree$, if it is present; and $\piBof{\lvalue}{\vtree}$ extracts the last subtree of $\vtree$, if it is present and the root of the second last subtree of $\vtree$ is equal to the local value $\lvalue$.

Rules \ruleNameSize{[E-LOC]} and \ruleNameSize{[E-FLD]} model the evaluation of expressions that are either a local value or a neighbouring field value, respectively. For instance, evaluating the expression $\mathtt{1}$ produces (by rule  \ruleNameSize{[E-LOC]}) the value-tree $\mkvtree{E-LOC}{1}{}$, while evaluating the expression $\mathtt{+}$ produces the value-tree $\mkvtree{E-LOC}{+}{}$. Note that, in order to ensure that domain restriction is obeyed (cf.\ Section~\ref{sec-concepts}), rule  \ruleNameSize{[E-FLD]} restricts the domain of the neighbouring field value $\fvalue$ to the domain of $\Trees$ augmented by $\deviceId$.

Rule \ruleNameSize{[E-B-APP]} models the application of built-in functions. It is used to evaluate expressions of the form $\e_{n+1}(\e_1 \cdots \e_n)$ such that the evaluation of $\e_{n+1}$ produces a value-tree $\vtree_{n+1}$ whose root $\rho(\vtree_{n+1})$ is a built-in function $\oname$. It produces the value-tree $\mkvtree{E-B-APP}{\anyvalue}{\vtree_{1},\ldots,\vtree_{n},\vtree_{n+1}}$, where  $\vtree_{1},\ldots,\vtree_{n+1}$ are the value-trees produced by the evaluation of the actual parameters and functional expression $\e_{1},\ldots,\e_{n+1}$ ($n\ge 0$) and $\anyvalue$ is the value returned by the function.
Rule \ruleNameSize{[E-B-APP]}  exploits the  special auxiliary function $\builtinop{\oname}{\deviceId}{\Trees}$, whose actual definition is abstracted away. This is such that $\builtinop{\oname}{\deviceId}{\Trees}(\overline\anyvalue)$ computes the result of applying built-in function $\oname$ to values $\overline\anyvalue$ in the current environment of the device $\deviceId$.
As in the denotational case, we require that $\builtinop{\oname}{\deviceId}{\Trees}$ always yields values of the expected type $\type$ where $\oname : (\overline{\type}) \to \type \in \OStypEnv$.

In particular, for the examples in this paper, we assume that the built-in 0-ary function $\selfK$ gets evaluated to the current device identifier (i.e.,  $\builtinop{\selfK}{\deviceId}{\Trees}() =\deviceId$), and that mathematical operators have their standard meaning, which is independent from $\deviceId$ and $\Trees$ (e.g., $\builtinop{+}{\deviceId}{\Trees}(1,2)=3$). We also assume that \texttt{map-hood}, \texttt{fold-hood} reflect the rules for function application, so that for instance $\texttt{map-hood}(\funvalue, \envmap{\overline\deviceId}{\overline\anyvalue}) = \envmap{\overline\deviceId}{\funvalue(\overline\anyvalue)}$ (where $\funvalue(\anyvalue_i)$ is computed w.r.t. the empty value-tree environment $\Trees = \emptyset$). The $\builtinop{\oname}{\deviceId}{\Trees}$ function also encapsulates measurement variables such as \texttt{nbr-range} and interactions with the external world via sensors and actuators.

In order to ensure that domain restriction is obeyed, for each built-in function $\oname$ we assume that $\builtinop{\oname}{\deviceId}{\Trees}(\anyvalue_1,\cdots,\anyvalue_n)$ is defined only if all the neighbouring field values in $\anyvalue_1,\ldots,\anyvalue_n$ have domain $\domof{\Trees}\cup\{\deviceId\}$; and if $\builtinop{\oname}{\deviceId}{\Trees}(\anyvalue_1,\cdots,\anyvalue_n)$ returns a neighbouring field value $\phi$, then  $\domof{\phi}=\domof{\Trees}\cup\{\deviceId\}$.
For example, evaluating the expression $\mathtt{+(1\; 2)}$ produces the value-tree 
$\mkvtree{E-B-APP}{3}{\mkvtree{E-LOC}{1}{},\mkvtree{E-LOC}{2}{},\mkvtree{E-LOC}{+}{}}$.
The value of the whole expression, $3$, has been computed by using rule \ruleNameSize{[E-B-APP]} to evaluate the application of the sum operator $+$  (the root of the third subtree of the value-tree) to the values $1$  (the root of the first subtree of the value-tree) and $2$  (the root of the second subtree of the value-tree). In the following, for sake of readability, we sometimes write the value $\anyvalue$ as short for the value-tree $\mkvtree{E-LOC}{\anyvalue}{}$. Following this convention, the value-tree 
$\mkvtree{E-B-APP}{3}{\mkvtree{E-LOC}{1}{},\mkvtree{E-LOC}{2}{},\mkvtree{E-LOC}{+}{}}$
is shortened to $\mkvtree{E-B-APP}{3}{1,2,+}$.

Rule \ruleNameSize{[E-D-APP]} models the application of user-defined or  anonymous functions, i.e., it is used to evaluate expressions of the form $\e_{n+1}(\e_1 \cdots \e_n)$ such that the evaluation of $\e_{n+1}$ produces a value-tree $\vtree_{n+1}$ whose root  $\lvalue=\vrootOf{\vtree_{n+1}}$ is a user-defined function name or an anonymous function.
It is similar to rule \ruleNameSize{[E-B-APP]}, however it produces a value-tree which has one more  subtree, $\vtree_{n+2}$, which is produced by  evaluating the body of the function $\lvalue$ with respect to\ the value-tree environment $\piBof{\lvalue}{\Trees}$ containing only the value-trees associated to the evaluation of the body of the same function $\lvalue$.

To illustrate rule \ruleNameSize{[E-REP]} ($\repK$ construct), as well as computational rounds, we consider program \mbox{\texttt{rep(0)\{(x) => +(x, 1)\}}} (cf.\ Section~\ref{sec-calculus-syntax}).
The first firing of a device $\deviceId$ after activation or reset is performed against the empty tree environment. Therefore, according to rule  \ruleNameSize{[E-REP]}, to evaluate  \mbox{\texttt{rep(0)\{(x) => +(x, 1)\}}} means to evaluate the subexpression \mbox{\texttt{+(0, 1)}}, obtained from \mbox{\texttt{+(x, 1)}} by replacing \mbox{\texttt{x}} with \mbox{\texttt{0}}. This produces the  value-tree $\vtree=\mkvtree{E-REP}{1}{0, \mkvtree{E-B-APP}{1}{0,1,+}}$, where root $1$ is the overall result as usual, while its sub-trees are the result of evaluating the first and second argument respectively.
Any subsequent firing of the device $\deviceId$ is performed with respect to\ a tree environment $\Trees$ that associates to $\deviceId$ the outcome of the most recent firing of $\deviceId$. Therefore, evaluating    \mbox{\texttt{rep(0)\{(x) => +(x, 1)\}}} at the second firing means to evaluate the subexpression \mbox{\texttt{+(1, 1)}}, obtained from \mbox{\texttt{+(x, 1)}} by replacing \mbox{\texttt{x}} with \mbox{\texttt{1}}, which is the root of $\vtree$. Hence the results of computation are $1$, $2$, $3$, and so on.

Notice that in both rules \ruleNameSize{[E-REP]}, \ruleNameSize{[E-NBR]} we do not assume that $\Trees$ is empty whenever it does not contain $\deviceId$. This might seem unnatural at a first glance, since every time a device is restarted its first firing is computed with respect to the empty value-tree environment, and all the subsequent firings will contain $\deviceId$ their domains. However, this fact is not inductively true for the sub-expressions of $\emain$: for example, the first time a conditional guard evaluates to $\truevalue$ the if-expression will be evaluated w.r.t. an environment not containing $\deviceId$ but possibly containing other devices whose guard evaluated to $\truevalue$ in their previous round of computation.

Value-trees also support modelling information exchange through the $\nbrK$ construct, as of rule \ruleNameSize{[E-NBR]}. Consider 
the program $\e'=\minHoodK(\nbrK\{\snsNumK()\})$, where the 1-ary built-in function $\minHoodK$ returns the lower limit of values in the range of its neighbouring field argument, and  the 0-ary built-in function $\snsNumK$ returns the numeric value measured by a sensor. Suppose that the program runs  on a network of three fully connected devices $\deviceId_A$,  $\deviceId_B$, and $\deviceId_C$ where $\snsNumK$ returns  \texttt{1} on $\deviceId_A$,  \texttt{2} on $\deviceId_B$, and \texttt{3} on $\deviceId_C$. Considering an initial empty tree-environment $\emptyset$ on all devices, we have the following:
the evaluation of  $\snsNumK()$ on $\deviceId_A$ yields $\mkvtree{E-B-APP}{1}{\snsNumK}$ (by rules  \ruleNameSize{[E-LOC]} and \ruleNameSize{[E-B-APP]}, since $\builtinop{\snsNumK}{\deviceId_A}{\emptyset}()=1$); the evaluation of  $\nbrK\{\snsNumK()\}$ on $\deviceId_A$ yields $\mkvtree{E-NBR}{(\envmap{\deviceId_A}{1})}{\mkvtree{E-B-APP}{1}{\snsNumK}}$ (by rule \ruleNameSize{[E-NBR]}); and the evaluation of  $\e'$ on $\deviceId_A$ yields
\[
\begin{array}{l@{\quad}c@{\quad}l}
	\vtree_A & = & \mkvtree{E-B-APP}{1}{\mkvtree{E-NBR}{(\envmap{\deviceId_A}{1})}{\mkvtree{E-B-APP}{1}{\snsNumK}},\minHoodK}
\end{array}
\] 
(by rule \ruleNameSize{[E-B-APP]}, since $\builtinop{\minHoodK}{\deviceId_A}{\emptyset}(\envmap{\deviceId_A}{1})=1$). Therefore, after its first firing, device $\deviceId_A$ produces the value-tree $\vtree_A$.
 Similarly, after their first firing,  devices $\deviceId_B$ and $\deviceId_C$ produce the value-trees
\[
\begin{array}{l@{\quad}c@{\quad}l}
	\vtree_B & = & \mkvtree{E-B-APP}{2}{\mkvtree{E-NBR}{(\envmap{\deviceId_B}{2})}{\mkvtree{E-B-APP}{2}{\snsNumK}},\minHoodK} \\
	\vtree_C & = & \mkvtree{E-B-APP}{3}{\mkvtree{E-NBR}{(\envmap{\deviceId_C}{3})}{\mkvtree{E-B-APP}{3}{\snsNumK}},\minHoodK}
\end{array}
\]
respectively. Suppose that device $\deviceId_B$ is the first device that fires a second time. Then the  evaluation of $\e'$ on $\deviceId_B$ is now performed with respect to\ the value tree environment \mbox{$\Trees_{B} = (\envmap{\deviceId_A}{\vtree_A},\;\envmap{\deviceId_B}{\vtree_B},\;\envmap{\deviceId_C}{\vtree_C})$} and the evaluation of its subexpressions  $\nbrK\{\snsNumK()\}$  and $\snsNumK()$ is performed, respectively, with respect to\ the  following  value-tree environments  obtained from $\Trees_{B}$ by alignment:
\[
\begin{array}{l}
	\Trees'_{B}  \; = \;  \piIof{1}{\Trees_{B}} \; = \; (\envmap{\deviceId_A}{\mkvtree{E-NBR}{(\envmap{\deviceId_A}{1})}{\mkvtree{E-B-APP}{1}{\snsNumK}}},\;\;\envmap{\deviceId_B}{\cdots},\;\;\envmap{\deviceId_C}{\cdots})
	\\
	\Trees''_{B}     \; = \;  \piIof{1}{\Trees'_{B}}  \; = \;   (\envmap{\deviceId_A}{\mkvtree{E-B-APP}{1}{\snsNumK}},\;\;\envmap{\deviceId_B}{\mkvtree{E-B-APP}{2}{\snsNumK}},\;\;\envmap{\deviceId_C}{\mkvtree{E-B-APP}{3}{\snsNumK}})
\end{array}
\] 
We have that $\builtinop{\snsNumK}{\deviceId_B}{\Trees''_B}()=2$; the evaluation of  $\nbrK\{\snsNumK()\}$ on $\deviceId_B$ with respect to\ $\Trees'_B$ yields $\mkvtree{E-NBR}{\phi}{\mkvtree{E-B-APP}{2}{\snsNumK}}$ where $\phi = (\envmap{\deviceId_A}{1},\envmap{\deviceId_B}{2},\envmap{\deviceId_C}{3})$; and $\builtinop{\minHoodK}{\deviceId_B}{\Trees_B}(\phi)=1$. Therefore the  evaluation of $\e'$ on $\deviceId_B$ produces the value-tree $\mkvtree{E-B-APP}{1}{\mkvtree{E-NBR}{\phi}{\mkvtree{E-B-APP}{2}{\snsNumK}},\minHoodK}$.
Namely, the computation at device $\deviceId_B$ after the first round yields $1$, which is the minimum of $\snsNumK$ across neighbours---and similarly for $\deviceId_A$ and $\deviceId_C$.

We now present an example illustrating first-class functions. 
Consider the program $\pickHoodK(\nbrK\{\snsFunK()\})$, where the 1-ary built-in function $\pickHoodK$ returns at random a value in the range of its neighbouring field argument, and  the 0-ary built-in function $\snsFunK$ returns a 0-ary function returning a value of type $\ntype$. Suppose that the program runs again on a network of three fully connected devices $\deviceId_A$,  $\deviceId_B$, and $\deviceId_C$ where $\snsFunK$ returns  $\lvalue_0=() \toSymK 0$ on $\deviceId_A$ and  $\deviceId_B$, and returns  $\lvalue_1=() \toSymK \e'$ on $\deviceId_C$, where $\e'=\minHoodK(\nbrK\{\snsNumK()\})$ is the program illustrated in the previous example. Assume that $\snsNumK$ returns  \texttt{1} on $\deviceId_A$,  \texttt{2} on $\deviceId_B$, and \texttt{3} on $\deviceId_C$. Then after its first firing, device $\deviceId_A$ produces the value-tree
\[
\begin{array}{l}
	\vtree'_A \; = \; \mkvtree{E-D-APP}{0}{
	\mkvtree{E-B-APP}{\lvalue_0}{
	\mkvtree{E-NBR}{(\envmap{\deviceId_A}{\lvalue_0})}{\mkvtree{E-B-APP}{\lvalue_0}{\snsFunK}},\pickHoodK}, 0}
\end{array}
\] 
where the root  of the first subtree of $\vtree'_A$ is the anonymous function $\lvalue_0$ (defined above), and the second subtree  of $\vtree'_A$, $0$, has been produced by the evaluation of the body $0$ of $\lvalue_0$.
After their first firing, devices $\deviceId_B$ and $\deviceId_C$ produce the value-trees
\[
\begin{array}{l}
	\vtree'_B \; = \; \mkvtree{E-D-APP}{0}{
	\mkvtree{E-B-APP}{\lvalue_0}{
	\mkvtree{E-NBR}{(\envmap{\deviceId_B}{\lvalue_0})}{\mkvtree{E-B-APP}{\lvalue_0}{\snsFunK}},\pickHoodK}, 0}
	\\
	\vtree'_C \; = \; \mkvtree{E-D-APP}{3}{
	\mkvtree{E-B-APP}{\lvalue_1}{
	\mkvtree{E-NBR}{(\envmap{\deviceId_C}{\lvalue_1})}{\mkvtree{E-B-APP}{\lvalue_1}{\snsFunK}},\pickHoodK}, \vtree_C}
\end{array}
\] 
respectively, where $\vtree_C$ is the value-tree for $\e$ given in the previous example.

Suppose that device $\deviceId_A$ is the first device that fires a second time, and its $\pickHoodK$ selects the function shared by device $\deviceId_C$. The computation is performed with respect to\ the value tree environment
$\Trees'_{A} =  (\envmap{\deviceId_A}{\vtree'_A},\;\envmap{\deviceId_B}{\vtree'_B},\;\envmap{\deviceId_C}{\vtree'_C})$
and produces the value-tree
$\mkvtree{E-D-APP}{1}{
\mkvtree{E-B-APP}{\lvalue_1}{
\mkvtree{E-NBR}{\phi'}{\mkvtree{E-B-APP}{\lvalue_1}{\snsFunK}},\pickHoodK}, \vtree''_A}$, where $\phi' = (\envmap{\deviceId_A}{\lvalue_1},\envmap{\deviceId_C}{\lvalue_1})$ and $\vtree''_A  =  
 \mkvtree{E-B-APP}{1}{\mkvtree{E-NBR}{(\envmap{\deviceId_A}{1},\envmap{\deviceId_C}{3})}{\mkvtree{E-B-APP}{1}{\snsNumK}},\minHoodK}$, since, according to rule \ruleNameSize{[E-D-APP]},  the evaluation of the body $\e'$ of $\lvalue_1$  (which produces the value-tree $\vtree''_A$) is performed with respect to the value-tree environment
$
\piBof{\lvalue_1}{\Trees'_A} =  (\envmap{\deviceId_C}{\vtree_C}).
$
Namely, device $\deviceId_A$ executed the anonymous function $\lvalue_1$ received from $\deviceId_C$, and this was able to correctly align with execution of $\lvalue_1$ at $\deviceId_C$, gathering values perceived by $\snsNumK$ of $1$ at $\deviceId_A$ and $3$ at $\deviceId_C$.

\subsection{Network Semantics (Small-Step Operational Semantics)}\label{sec-calculus-network-semantics}

\begin{figure}[!t]{
 \framebox[1\textwidth]{
 $\begin{array}{l}
 \textbf{System configurations and action labels:}\\
\begin{array}{lcl@{\hspace{8cm}}r}
\Field & \BNFcce &  \envmap{\overline\deviceId}{\overline\Trees}    &   {\footnotesize \mbox{status field}} \\
\Topo & \BNFcce &  \envmap{\overline\deviceId}{\overline\devset}    &   {\footnotesize \mbox{topology}} \\
\Sens & \BNFcce &  \envmap{\overline\deviceId}{\overline\senstate}    &   {\footnotesize \mbox{sensors-map}} \\
\Envi & \BNFcce &  \EnviS{\Topo}{\Sens}    &   {\footnotesize \mbox{environment}} \\
\Cfg & \BNFcce &  \SystS{\Envi}{\Field}    &   {\footnotesize \mbox{network configuration}} \\
\act & \BNFcce &  \deviceId \;\BNFmid\; \envact    &   {\footnotesize \mbox{action label}} \\
\end{array}\\
\hline\\[-8pt]
\textbf{Environment well-formedness:}\\
\begin{array}{l}
\wfn{\EnviS{\Topo}{\Sens}} \textrm{~~holds if $\Topo,\Sens$ have same domain, and $\Topo$'s values do not escape it.}
\\
\end{array}\\
\hline\\[-8pt]
\textbf{Transition rules for network evolution:} \hfill
  \boxed{\nettran{\Cfg}{\act}{\Cfg}}
  \\[0.2cm]
\vspace{0.5cm}
\begin{array}{c}
\netopsemRule{N-FIR}
                 {\quad \Envi=\EnviS{\Topo}{\Sens}
                   \quad \Topo(\deviceId)= \overline\deviceId 
                  \quad \bsopsem{\deviceId}{\filter(\Field)(\deviceId)}{\emain}{\vtree} \; (\mbox{w.r.t.} \; \Sens(\deviceId))
                  \quad
                 \Field_1=\envmap{\overline\deviceId}{\{\envmap{\deviceId}{\vtree}\}}}
                 {\nettran{\SystS{\Envi}{\Field}}{\deviceId}{\SystS{\Envi}{\mapupdate{\filter(\Field)}{\Field_1}}}
                 }
\skiptransition
\netopsemRule{N-ENV}
                 {\qquad \wfn{\Envi'}\qquad \Envi'=\EnviS{\Topo}{\envmap{\overline\deviceId}{\overline\senstate}} \qquad
                  \Field_0=\envmap{\overline\deviceId}{\emptyset}
                 }
                 {\nettran{\SystS{\Envi}{\Field}}{\envact}{\SystS{\Envi'}{\mapupdate{\Field_0}{\Field}}}
                 }\\[-10pt]
\end{array}\\
%
%
%
\end{array}$}
} 
 \caption{Small-step operational semantics for network evolution.} \label{fig:networkSemantics}
\end{figure}

We now provide an operational semantics for the evolution of whole networks, namely,
for modelling the distributed evolution of computational fields over time.
Figure \ref{fig:networkSemantics} (top) defines key syntactic elements to this end.
$\Field$ models the overall status of the devices in the network at a given time, as a map from device identifiers to value-tree environments. From it we can define the state of the field at that time by summarising the current values held by devices as the partial map from device identifiers to values defined by $\fvalue(\deviceId) = \vrootOf{\Field(\deviceId)(\deviceId)}$ if $\Field(\deviceId)(\deviceId)$ exists.
$\Topo$ models \emph{network topology}, namely, a directed neighbouring graph, as a map from device identifiers to set of identifiers.
$\Sens$ models \emph{sensor (distributed) state}, as a map from device identifiers to (local) sensors (i.e., sensor name/value maps).
Then, $\Envi$ (a couple of topology and sensor state) models the system's environment.
So, a whole network configuration $\Cfg$ is a couple of a status field and environment.

$\filter(\cdot)$ in Figure \ref{fig:networkSemantics} (bottom), rule \ruleNameSize{[N-FIR]}, models a filtering operation that clears out old stored values from the value-tree environments in $\Field$, implicitly based on space/time tags.
We use the following notation for status fields. Let $\envmap{\overline\deviceId}{\Trees}$ denote the map sending each device identifier in $\overline\deviceId$ to the same value-tree environment $\Trees$. Let $\mapupdate{\Trees_0}{\Trees_1}$ denote the value-tree environment with domain $\domof{\Trees_0} \cup \domof{\Trees_1}$ coinciding with $\Trees_1$ in the domain of $\Trees_1$ and with $\Trees_0$ otherwise. Let $\mapupdate{\Field_0}{\Field_1}$ denote the status field with the \emph{same domain} as $\Field_0$ made of $\envmap{\deviceId}{\mapupdate{\Field_0(\deviceId)}{\Field_1(\deviceId)}}$ for all $\deviceId$ in the domain of $\Field_1$, $\envmap{\deviceId}{\Field_0(\deviceId)}$ otherwise.

We define network operational semantics in terms of small-steps transitions of the kind $\nettran{\Cfg}{\act}{\Cfg'}$, where $\act$ is either a device identifier in case it represents its firing, or label $\envact$ to model any environment change.
This is formalised in Figure \ref{fig:networkSemantics} (bottom).
Rule \ruleNameSize{[N-FIR]} models a computation round (firing) at device $\deviceId$: it takes the local value-tree environment filtered out of old values $\filter(\Field)(\deviceId)$; then by the single device semantics it obtains the device's value-tree $\vtree$, which is used to update the system configuration of $\deviceId$'s neighbours--the local sensors $\Sens(\deviceId)$ are implicitly used by the auxiliary function
$\builtinop{\oname}{\deviceId}{\Trees}$ that gives the semantics to the built-in functions.
Rule \ruleNameSize{[N-ENV]} takes into account the change of the environment to a new well-formed environment $\Envi'$. Let $\overline\deviceId$ be the domain of $\Envi'$. We first construct a status field $\Field_0$ associating to all the devices of $\Envi'$ the empty context $\emptyset$. Then, we adapt the existing status field $\Field$ to the new set of devices: $\mapupdate{\Field_0}{\Field}$ automatically handles removal of devices, map of new devices to the empty context, and retention of existing contexts in the other devices.

\begin{example} \label{ex:operational}
	In a possible implementation (by adding a ``time tag'' to every value tree $\timeF{\vtree}$ and action label $\timeF{\act}$) we can define $\filter(\Field)$ as the mapping from $\overline\deviceId$ to $\overline{\filter(\Trees)}$ where
	\[
	\filter(\Trees) = \{\envmap{\deviceId}{\vtree} \in \Trees: ~ \timeF{\vtree} \geq \timeF{\act} - \decay\}
	\]
	(recall that $\decay$ is the decay parameter).
	
	Furthermore, we can proceed in analogy with Example \ref{ex:denotational} and define a sequence of network evolution rules from a set of paths $\PathS$ together with a neighbouring predicate between positions and time tags for firings. In particular:
	\begin{itemize}
		\item We introduce an occurrence of Rule \ruleNameSize{[N-ENV]} updating with the topology $\Topo$ given by
		\[
		\Topo(\deviceId) = \{ \deviceId' : ~ \neighbour{\PathS(\deviceId)(t)}{\PathS(\deviceId')(t)}\}
		\]
		for any time $t$ corresponding to an activation change for a device (i.e. a border of an interval in which a path $\PathS(\delta)$ is defined);
		\item We also introduce an occurrence of Rule \ruleNameSize{[N-ENV]} as above for any time $t_\eventId$ corresponding to the firing $\eventId$ of a device, each of them followed by
		\item an occurrence of Rule \ruleNameSize{[N-FIR]} on device $\deviceId_\eventId$.
	\end{itemize}
	The above sequence of rules is to be intended as sorted time-wise.
\end{example}

\newcommand{\tcomp}{t_{\texttt{c}}}

\section{Properties of HFC}\label{sec-properties}

In this section we present the main properties of HFC, namely:
1) type preservation and domain alignment, intuitively meaning that HFC is ``safe'' in that it maintains lexical scoping in its handling of fields and computations with fields, and
2) adequacy and full abstraction, intuitively meaning that any aggregate-level program described within the denotational semantics is correctly implemented by the corresponding local actions of the operational semantics.

We shall defer all the proofs to Appendix \ref{apx-properties}, focusing here on statements and definitions. 
We remark that throughout this paper we assume that the execution of every firing event $\eventId$ terminates and is instantaneous and that every event happens at a distinct moment in time. 
Termination of the execution of a firing event can be enforced by means of several different techniques (see among others \cite{Gasarch:Termination}), whose specific details fall outside the scope of this paper.

\subsection{Device Computation Type Preservation and Domain Alignment}\label{sec-preservation}

Here we formally state the device computation type preservation and domain alignment properties (cf. Section~\ref{sec-typing}) for the HFC calculus. In order to state these properties we introduce the notion of well-formed value-tree environment for an expression.
 
Given a closed expression $\e$, a local-type-scheme environment $\LTStypEnv$, a type environment $\TtypEnv=\overline{\xname}:\overline{\type}$, and a type $\type$ such that $\expTypJud{\LTStypEnv}{\TtypEnv}{\e}{\type}$ holds, the set $\WFVT{\LTStypEnv}{\TtypEnv}{\e}$ of the \emph{well-formed value-trees} for $\e$ is inductively defined as follows. $\vtree \in \WFVT{\LTStypEnv}{\TtypEnv}{\e}$ if and only if $\anyvalue = \vrootOf{\vtree}$ has type $T$ and
\begin{itemize}
	\item if $\e$ is a value, $\vtree$ is of the form $\mkvtree{}{\anyvalue}{}$;
	\item if $\e = \nbrK\{\e_1\}$, $\vtree$ is of the form $\mkvtree{}{\anyvalue}{\vtree_1}$;
	\item if $\e = \repK(\e_1)\{(\xname) \; \toSymK \; \e_2\}$, $\vtree$ is of the form $\mkvtree{}{\anyvalue}{\vtree_1,\vtree_2}$ where $\vtree_2$ is well-formed for $\e_2$ with the additional assumption $\xname : \type$;
	\item if $\e = \e_{n+1}(\overline{\e})$, $\vtree$ is of one of the following two forms:
	\begin{itemize}
		\item $\mkvtree{}{\anyvalue}{\overline{\vtree}, \vtree_{n+1}}$ where $\funvalue = \vrootOf{\vtree_{n+1}}$ is a built-in operator,
		\item $\mkvtree{}{\anyvalue}{\overline{\vtree}, \vtree_{n+1}, \vtree_{n+2}}$ where $\funvalue$ is not a built-in operator and $\vtree_{n+2}$ is well-formed with respect to $\e_{n+2} = \body{\funvalue}$ with the additional assumptions that $\args{\funvalue} : \overline{\type}'$ where $\overline{\e} : \overline{T}'$.
	\end{itemize}
\end{itemize}
In the above definition, $\vtree_i$ is always assumed to be in the corresponding $\WFVT{\LTStypEnv}{\TtypEnv}{\e_i}$. The set $\WFVTE{\LTStypEnv}{\TtypEnv}{\e}$ of the \emph{well-formed value-tree environments} for $\e$ is such that $\Trees \in \WFVTE{\LTStypEnv}{\TtypEnv}{\e}$ if and only if $\Trees = \envmap{\overline{\deviceId}}{\overline{\vtree}}$ where each $\vtree_i$ is in $\WFVT{\LTStypEnv}{\TtypEnv}{\e}$.

As these notions are defined we can now formally state the type preservation and domain alignment properties (cf.\ Section~\ref{sec-typing}).

\begin{thm}[Device computation type preservation and domain alignment] \label{the-DeviceTypePreservationAndDomainAlignment}
	Let $\TtypEnv=\overline{\xname}:\overline{\type}$, $\expTypJud{\LTStypEnv}{\emptyset}{\overline{\anyvalue}}{\overline{\type}}$, so that $\lengthOf{\overline{\anyvalue}}=\lengthOf{\overline{\xname}}$.
	If $\expTypJud{\LTStypEnv}{\TtypEnv}{\e}{\type}$, $\Trees\in\WFVTE{\LTStypEnv}{\TtypEnv}{\e}$ and $\bsopsem{\deviceId}{\Trees}{\applySubstitution{\e}{\substitution{\overline{\xname}}{\overline{\anyvalue}}}}{\vtree}$, then:
	\begin{enumerate}
		\item 
		$\vtree \in \WFVT{\LTStypEnv}{\TtypEnv}{\e}$,
		\item
		$\expTypJud{\LTStypEnv}{\emptyset}{\vrootOf{\vtree}}{\type}$, and
		\item
		if $\vrootOf{\vtree}$ is a neighbouring field value $\fvalue$ then $\domof{\fvalue}=\domof{\Trees}\cup\{\deviceId\}$. 
	\end{enumerate}
\end{thm}

\begin{proof}
	See Appendix \ref{apx-properties}.
\end{proof}


\subsection{Adequacy and full abstraction}\label{sec-adequacy}

We are now able to prove that the denotational semantics introduced in Section \ref{sec-denotational-semantics} is adequate and fully abstract with respect to the operational semantics introduced in Section \ref{sec-calculus-operational-semantics}. We shall prove these properties for monomorphic types, which will imply that given any parametric type \emph{all of its possible instantiations} will satisfy the same property.

The notions of \emph{adequacy} and \emph{full-abstraction} are presented in literature in many (slightly) different forms (see e.g., \cite{Curien:FullAbstraction,Stoughton:FullyAbstract}), none of them fitting without modifications into the setting of HFC. In order to fix a reference, for an ML-like calculus we say that:
\begin{itemize}
	\item
	The denotational semantics is \emph{adequate} iff for any expression $\e : \type$,
	\begin{itemize}
		\item $\e$ converges operationally iff it converges denotationally,
		\item if $\e$ evaluates to $\anyvalue$ then their denotations are equal ($\llbracket \e \rrbracket = \llbracket \anyvalue \rrbracket$),
		\item and the converse of the last property holds whenever $\type$ is a built-in type.
	\end{itemize}
	\item
	\emph{Full abstraction} holds iff given any two expressions $\e_1$, $\e_2$ of a certain type $\type$, their denotations coincide ($\llbracket \e_1 \rrbracket = \llbracket \e_2 \rrbracket$) if and only if for all contexts $\e'$ of built-in type with free variable $\xname : \type$, $\applySubstitution{\e'}{\substitution{\xname}{\e_1}}$ and $\applySubstitution{\e'}{\substitution{\xname}{\e_2}}$ evaluate to the same value.
\end{itemize}
Remark that the forward direction of full abstraction is implied by adequacy.

These notions need to be adjusted for the field calculus in order to accommodate for the influence of the environment and of the previous rounds of computations. In general, we achieve this roughly by adding a ``for all environments'' quantifier inside any operational or denotational statement.

Consider a program $\emain$ interpreted both operationally and denotationally, where each event $\eventId$ is in bijection with one occurrence of rule \ruleNameSize{[N-FIR]}. Let $\Field_\eventId$, $\Topo_\eventId$ denote the status of the status field and topology just before the occurrence corresponding to $\eventId$. Let $\Trees_\eventId = \filter(\Field_\eventId)(\deviceId_\eventId)$ denote the value-tree environment used in the computation of the firing corresponding to $\eventId$, and let $\vtree_\eventId$ denote the outcome of this computation.

We say that the denotational environment is \emph{coherent} with the operational evolution of a network if and only if:
\begin{enumerate}
	\item For each $\eventId$ in $\EventS$, $\Trees_\eventId = \left\{\envmap{\deviceId_{\eventId'}}{\vtree_{\eventId'}}: ~ \neighbour{\eventId}{\eventId'} \right\}$. This is equivalent to the assertion that for each $\eventId$, $\eventId'$ in $\EventS$, $\neighbour{\eventId}{\eventId'}$ holds if and only if:
	\begin{itemize}
		\item $\deviceId_\eventId \in \Topo_{\eventId'}(\deviceId_{\eventId'})$;
		\item there is no further $\eventId''$ between $\eventId'$ and $\eventId$ such that $\deviceId_{\eventId'} = \deviceId_{\eventId''}$ and $\deviceId_\eventId \in \Topo_{\eventId''}(\deviceId_{\eventId''})$;
		\item $\filter(\Field_\eventId)$ does not filter out $\eventId'$ (referring to Example \ref{ex:operational}: $t_{\eventId'} \geq t_\eventId - \decay$).
	\end{itemize}
	\item $\builtindenot{B}{\oname}$ operates pointwise on its arguments (i.e. does not incorporate communication between events) and correctly translates the behaviour of $\builtinop{\oname}{\deviceId_\eventId}{\Trees}$, that is:
	\[
	\denotexp{ \builtinop{\oname}{\deviceId_\eventId}{\Trees}(\overline\anyvalue) }{}{\eventS} = \builtindenot{B}{\oname}\left( \denotexp{\overline\anyvalue}{}{\eventS} \right)
	\]
	for all values $\overline\anyvalue$ of the correct type and $\eventS$ matching $\domof{\Trees}$.
\end{enumerate}

We remark that the possible implementations outlined in Examples \ref{ex:denotational} and \ref{ex:operational} are coherent.

\begin{thm}[Adequacy] \label{thm:Adequacy}
	Assume that the denotational environment is coherent with the operational evolution of the network and $\e$ is well-typed (that is, $\expTypJud{\LTStypEnv}{\emptyset}{\e}{\type}$).
	
	Then $\denotexp{\e}{}{\EventS}(\eventId) = \denotexp{\anyvalue_\eventId}{}{\EventS}(\eventId)$ for each $\eventId$ in $\EventS$, where $\anyvalue_\eventId = \vrootOf{\vtree_\eventId}$ is the operational outcome of fire $\eventId$.
\end{thm}

\begin{proof}
	See Appendix \ref{apx-properties}.
\end{proof}

We say that \emph{full abstraction} holds for a field calculus iff given any two expression $\e_1$, $\e_2$ of a certain type $\type$, $\denotexp{\e_1}{}{\eventS} = \denotexp{\e_2}{}{\eventS}$ for all $\eventS$ if and only if for all environments $\Envi$ and contexts $\e'$ of built-in local type $\builtintype$ with free variable $\xname : \type$, $\applySubstitution{\e'}{\substitution{\xname}{\e_1}}$ and $\applySubstitution{\e'}{\substitution{\xname}{\e_2}}$ evaluate to the same value trees in each firing.

The operational and denotational semantics hereby presented satisfy the full abstraction property, \emph{provided that it holds for values of built-in local type}: we say that built-in constructors are faithful iff any two syntactically different expressions $\dcOf{\dc}{\overline{\lvalue}}$, $\dcOf{\dc'}{\overline{\lvalue}'}$ necessarily denote different objects.\footnote{This property fails for example if we include constructors \texttt{Succ}, \texttt{Pred} for type $\ntype$. On the other hand, it can hold for example if we assume to have a distinguished constructor $n$ for every integer $n$.}

\begin{thm}[Full Abstraction] \label{thm:FullAbstraction}
	Suppose that constructors for built-in local types are faithful. Then the full abstraction property holds.
\end{thm}

\begin{proof}
	See Appendix \ref{apx-properties}.
\end{proof}

\section{A Pervasive Computing Example}\label{sec-examples}

We now illustrate the application of field calculus, with a focus on first-class functions, using a pervasive computing example.
In this scenario, people wandering a large environment (like an
outdoor festival, an airport, or a museum) each carry a personal
device with short-range point-to-point ad-hoc capabilities (e.g. a
smartphone sending messages to others nearby via Bluetooth or Wi-Fi).
All devices run a minimal ``virtual machine'' that allows runtime
injection of new programs: any device can initiate a new distributed
process (in the form of a 0-ary anonymous function), which the virtual
machine spreads to all other devices within a specified range (e.g.,
30 meters).
For example, a person might inject a process that estimates crowd
density by counting the number of nearby devices or a process that
helps people to rendezvous with their friends, with such processes likely
implemented via various self-organisation mechanisms.
The virtual machine then executes these using the first-class function
semantics above, providing predictable deployment and execution of an
open class of runtime-determined processes.

\subsection{Virtual Machine Implementation}

\begin{figure}[t]
\begin{Verbatim}[fontsize=\fontsize{8pt}{9pt}, frame=single, commandchars=\\\{\}, codes={\catcode`$=3\catcode`^=7\catcode`_=8}]
\il{Computes a field of minimum distance from 'source' devices}
\km{def} \fn{distance-to} (source) \{  \il{has type:  $(\btype) \to \ntype$}
  \km{rep}(infinity) \{ (d) $\toSymK$ \fn{mux}(source, 0, \fn{min-hood+}( \fn{+[f,f]}(\km{nbr}\{d\}, \fn{nbr-range}()))) \}
\}

\il{Computes a field of pairs of distance to nearest 'source' device, and the most recent value of 'v' there}
\km{def} \fn{gradcast} (source, v) \{   \il{has type:  $\forall \ltvar. (\btype, \ltvar) \to  \ltvar$} 
  \fn{snd}( (x) $\toSymK$
          \km{rep}(x) \{
            (t) $\toSymK$ \fn{mux}(source, \fn{Pair}(0, v),
                       \fn{min-hood+}(\fn{Pair[f,f]}(\fn{+[f,f]}( \fn{nbr-range}(), \km{nbr}\{\fn{fst}(t)\}), \km{nbr}\{\fn{snd}(t)\})))
          \}
       (\fn{Pair}(infinity, v)))
\}

\il{Evaluate a function field, running 'f' from 'source' within 'range' meters, and 'no-op' elsewhere}
\km{def} \fn{deploy} (range, source, g, no-op) \{  \il{has type:  $\forall \ltvar. (\ntype, \btype, ()\to\ltvar, ()\to\ltvar) \to \ltvar$}
    \km{if} (\fn{distance-to}(source) \fn{<} range) \{\fn{gradcast}(source, g)\} \km{else} \{no-op\} ()
\}
    
\il{The entry-point function executed to run the virtual machine on each device}
\km{def} \fn{virtual-machine} () \{  \il{has type:  $() \to \ntype$}
    \fn{deploy}( \fn{sns-range}(), \fn{sns-injection-point}(), \fn{sns-injected-fun}(), () $\toSymK$ 0)
\}
\end{Verbatim}
\begin{Verbatim}[fontsize=\fontsize{8pt}{9pt}, frame=single, commandchars=\\\{\}, codes={\catcode`$=3\catcode`^=7\catcode`_=8}]
\il{Sums values of 'summand' into a minimum of 'potential', by descent}
\km{def} \fn{converge-sum} (potential, summand) \{  \il{has type:  $(\ntype, \ntype) \to \ntype$}
  \km{rep}(summand) \{
    (v) $\toSymK$ summand \fn{+} 
           \fn{sum-hood+}( \fn{mux[f,f,l]}( \km{nbr}\{\fn{parent}(potential)\} \fn{=[f,l]} \fn{uid}(), \km{nbr}\{v\}, 0))
  \}				          
\}

\il{Maps each device to the uid of the neighbour with minimum value of 'potential'}
\km{def} \fn{parent} (potential) \{  \il{has type:  $(\ntype) \to \ntype$}
  \fn{snd}( \fn{min-hood}( \fn{Pair[l,f]}( potential, 
                            \fn{mux[f,f,l]}( \km{nbr}\{potential\} \fn{<[f,l]} potential, \km{nbr}\{\fn{uid}()\}, NaN))))
\}

\il{Simple low-pass filter for smoothing noisy signal 'value' with rate constant 'alpha'}
\km{def} \fn{low-pass} (alpha, value) \{  \il{has type:  $(\ntype, \ntype) \to \ntype$}
  \km{rep}(value) \{ (filtered) $\toSymK$ \fn{*}(value, alpha) \fn{+} \fn{*}(filtered, \fn{-}(1, alpha)) \}
\}
\end{Verbatim}
\caption{Virtual machine code (top) and application-specific code (bottom).}\label{f:code}\vspace{-10pt}
\end{figure}

The complete code for our example is listed in Figure~\ref{f:code}.
We use the following naming conventions for built-ins:
functions \texttt{{sns-*}} embed sensors that return a value
perceived from the environment (e.g., \texttt{sns-injection-point}
returns a Boolean indicating whether a device's user wants to inject a
function); 
functions \texttt{{*-hood}} yield a local value $\lvalue$ obtained
by aggregating over the neighbouring field value $\fvalue$ in the input (e.g.,
\texttt{{sum-hood}} sums all values in each neighbourhood);
functions \texttt{{*-hood+}} behave the same but exclude the value
associated with the current device;
and built-in functions \texttt{{Pair}}, \texttt{{fst}}, and \texttt{{snd}}
respectively create a pair of locals and access a pair's first and
second component.
%

The first two functions in Figure~\ref{f:code} implement frequently
used self-organisation mechanisms.
As already discussed, function \texttt{{distance-to}}, also known as \emph{gradient}
\cite{clement2003self,original-gradient}, computes a field of minimal
distances from each device to the nearest ``source'' device (those
mapping to {\em true} in the Boolean input field).
%
%
Note that the process of estimating distances self-stabilises into the desired field of distances, regardless of any transient perturbations or faults~\cite{GradientSelfStabilization}.
The second self-organisation mechanism, \texttt{{gradcast}}, is a
directed broadcast,
achieved by a computation identical to that of
\texttt{{distance-to}}, except that the values are pairs 
(note that \texttt{{Pair[f,f]}} produces a neighbouring field of pairs, not a pair of neighbouring fields), with the
second element set to the value of {\tt v} at the source: \texttt{{min-hood}}
operates on pairs by applying lexicographic ordering, so the second
value of the pair is automatically carried along shortest paths from
the source.  
The result is a field of pairs of distance and most recent value of {\tt v} at
the nearest source, of which only the value is returned.

The latter two functions in Figure~\ref{f:code} use these
self-organisation methods to implement our simple virtual machine.
Code mobility is implemented by function \texttt{{deploy}}, which
spreads a 0-ary function \texttt{g} via \texttt{{gradcast}},
keeping it bounded within distance \texttt{range} from sources, and
holding 0-ary function \texttt{{no-op}} elsewhere.  The corresponding 
field of functions is then executed (note the double parenthesis).
The \texttt{{virtual-machine}} then simply calls
\texttt{{deploy}}, linking its arguments to sensors configuring
deployment range and detecting who wants to inject which functions
(and using \texttt{()$\toSymK$0} as \texttt{{no-op}} function).

In essence, this virtual machine implements a code-injection model
much like those used in a number of other pervasive computing
approaches (e.g.,~\cite{tota,linda,butera})---though of course it has
much more limited features, since it is only an illustrative example.
With these previous approaches, however, all code shares the same (global) lexical scope and
cannot have its network domain externally controlled: this means that injected
code may spread through the network unpredictably and may interact
unpredictably with other injected code that it encounters.  The
extended field calculus semantics that we have presented, however, thanks to the restriction property (Section \ref{ssec:properties}), 
ensures that injected code moves only within the range specified to the
virtual machine and remains lexically isolated from different injected
code, so that no variable can be unexpectedly affected by interactions
with neighbours.

\subsection{Simulation of Example Application}

We further illustrate the example in a simulated scenario, considering a museum whose docents
monitor their efficacy in part by tracking the number of patrons
nearby while they are working.  
To monitor the number of nearby patrons, each docent's device injects
the following anonymous function (of type: $() \to \ntype$):
\begin{Verbatim}[fontsize=\fontsize{8pt}{9pt}, frame=single, commandchars=\\\{\}, codes={\catcode`$=3\catcode`^=7\catcode`_=8}]
() $\toSymK$ \fn{low-pass}(0.5,\fn{converge-sum}( \fn{distance-to}(\fn{sns-injection-point}()), \fn{mux}(\fn{sns-patron()},1,0)))
\end{Verbatim}
This function is an anonymous version of the \texttt{{track-count}} function example in Section~\ref{s:conceptual:fields},
using the same low-pass filtering of summation of a potential field to the docent, 
except that since the function cannot have any arguments, the Boolean fields indicating locations of patrons and docents are instead acquired via virtual sensors.
In particular, in the \texttt{{converge-sum}} function, each device's local value is summed with those
identifying it as their parent (their closest neighbour to the source,
breaking ties with device unique identifiers from built-in function
\texttt{{uid}}), resulting in a relatively balanced spanning tree
of summations with the source at its root.
This very simple version of summation is somewhat noisy on a moving
network of devices \cite{VBDP-SASO2015}, so its output is passed through a simple low-pass
filter, the function \texttt{{low-pass}}, also defined in
Figure~\ref{f:code}(bottom), in order to smooth its output and improve
the quality of estimate.

Figure~\ref{f:snap} shows a simulation of a docent and 250 patrons in a large 100x30 meter museum gallery.
Of the patrons, 100 are a large group of school-children moving together past the stationary docent from one side of the gallery to the other (thus causing a coherent rise and fall in local crowd density), while the rest are wandering randomly.
In this simulation, people move at an average 1 m/s, the docent and all patrons carry personal devices running the virtual machine, executing asynchronously at 10Hz, and communicating via low-power Bluetooth to a range of 10 meters---hence, hop-by-hop communication is needed for longer range interaction.
The simulation was implemented using the {\sc Alchemist}~\cite{PMV-JOS2013} simulation framework and the Protelis~\cite{Protelis15} incarnation of field calculus, updated to the extended version of the calculus presented in this paper.

In this simulation, at time 10 seconds, the docent injects the patron-counting function with a range of 25 meters, and at time 70 seconds removes it.
Figure~\ref{f:snap} shows two snapshots of the simulation, at times 11 (top) and 35 (bottom) seconds, while Figure~\ref{f:count} compares the estimated value returned by the injected process with the true value.
Note that upon injection, the process rapidly disseminates and begins producing good estimates of the number of nearby patrons, then cleanly terminates upon removal.

\begin{figure}[t]
\centering
\subfloat[Simulation snapshots]{%
\begin{minipage}[]{0.403\textwidth}%
\fbox{\includegraphics[width=\textwidth]{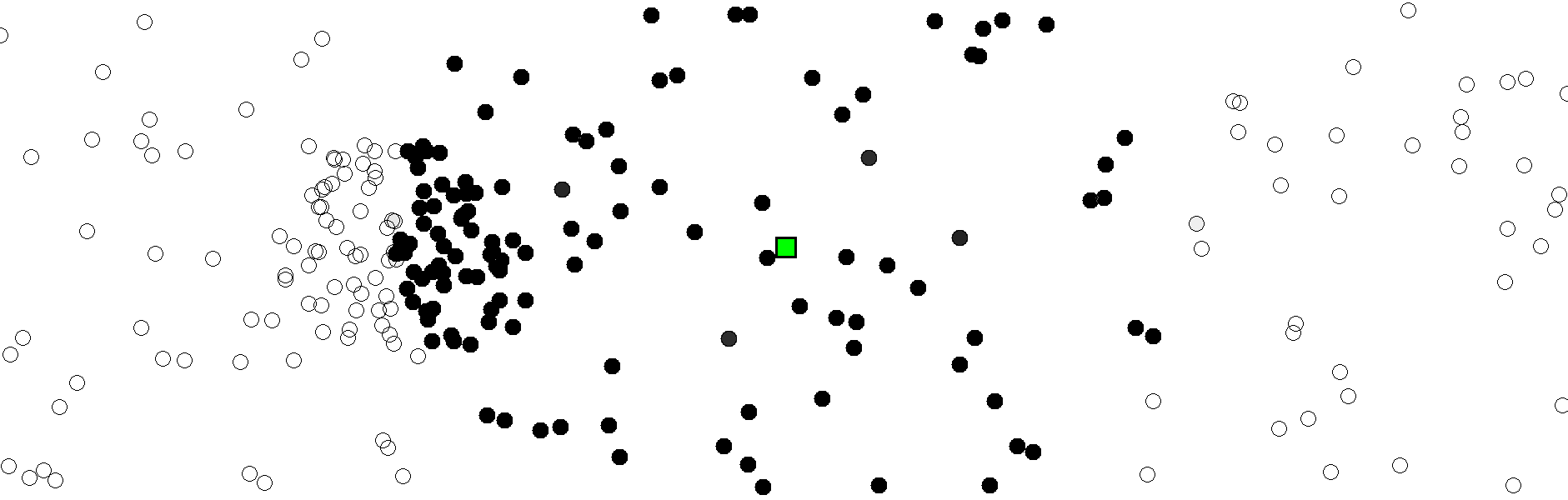}}\\[5pt]
\fbox{\includegraphics[width=\textwidth]{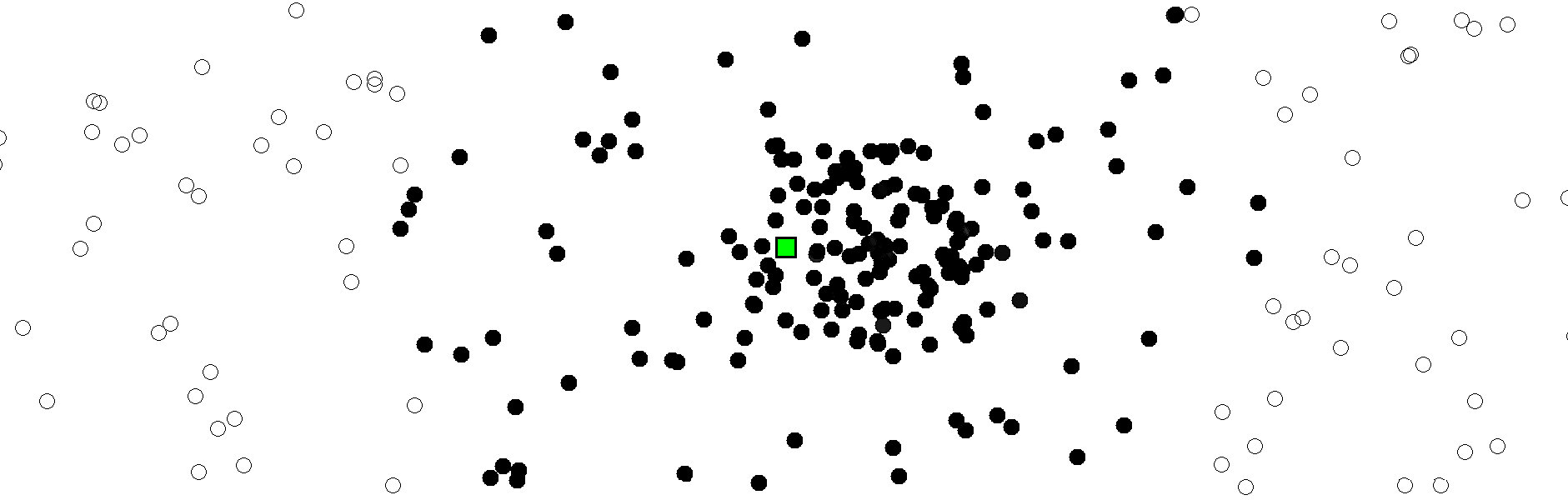}}%
\end{minipage}\label{f:snap}}
\hspace{10pt}
\subfloat[Estimated vs. True Count]{%
\begin{minipage}[]{0.497\textwidth}%
\fbox{\includegraphics[width=\textwidth]{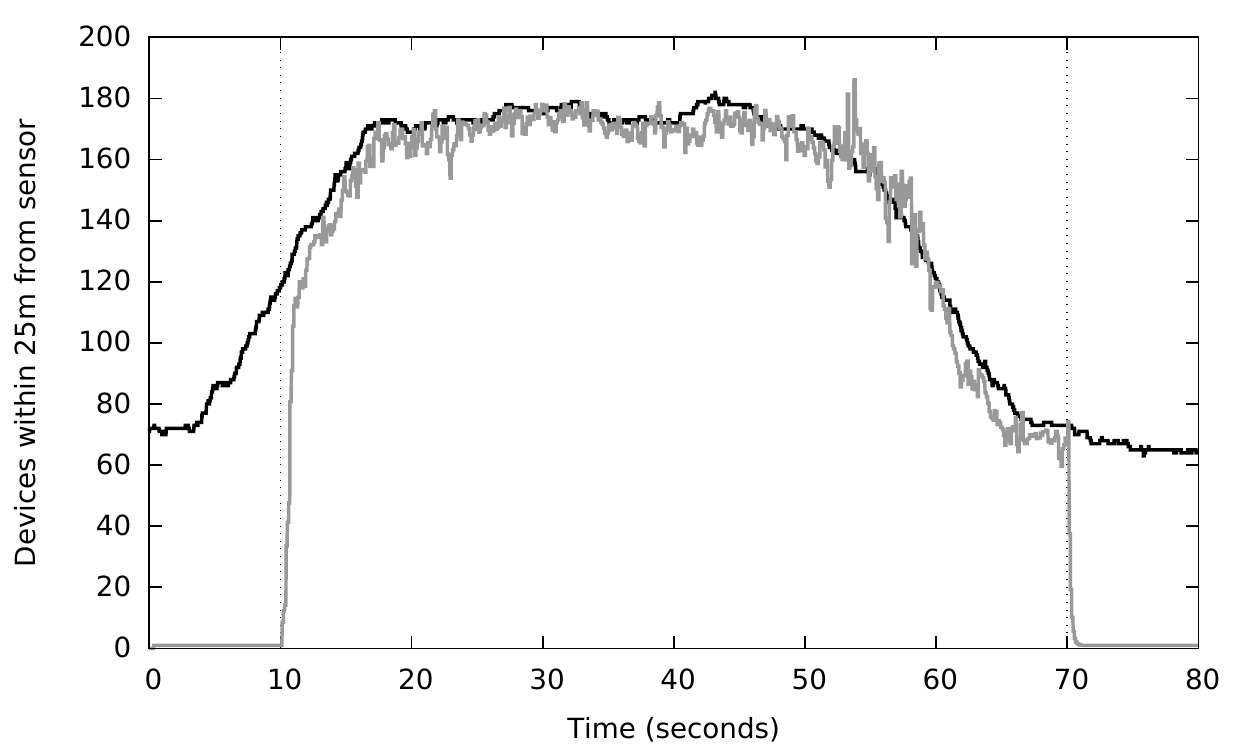}}
\end{minipage}\label{f:count}}
\caption{(a) Two snapshots of museum simulation: patrons (grey) are
  counted (black) within 25 meters of the docent (green). (b)
  Estimated number of nearby patrons (grey) vs. actual number (black)
  in the simulation, during the period where the docent is running the patron counting function (between two dashed lines).}
\end{figure}

All together, these examples illustrate how the coherence between global denotational and local operational semantics in field calculus allows complex distributed applications to be safely and elegantly implemented with quite compact code,
including higher-order operations like process management and runtime deployment of applications in an open environment.

\section{Conclusion and Future Work}
\label{sec-conclusion}

Conceiving emerging distributed systems in terms of computations involving aggregates of devices,
and hence adopting higher-level abstractions for system development, is a thread that has recently received a good deal of attention, as discussed in Section~\ref{sec-concepts}.
Those that best support self-organisation approaches to robust and
environment-independent computations, however, have generally lacked
well-engineered mechanisms to support openness and code mobility
(injection, update, etc.).
Our contribution has been to develop a core calculus,
building on the work presented in \cite{VDB-FOCLASA-CIC2013}, that for the first time smoothly combines
self-organisation and code mobility, by means of the abstraction of a ``distributed function
field''.
This combination of first-class functions with the
domain-restriction mechanisms of field calculus allows the predictable
and safe composition of distributed self-organisation mechanisms at
runtime, thereby enabling robust operation of open pervasive systems.
Furthermore, the simplicity of the calculus enables it to easily serve
as both an analytical framework and a programming framework, and we
have already incorporated this into Protelis~\cite{Protelis15},
thereby allowing these mechanisms to be deployed both in simulation
and in actual distributed systems.

Future plans include consolidation of this work, by extending the
calculus and its conceptual framework, to support an analytical
methodology and a practical toolchain for system development, as
outlined in~\cite{BV-FOCAS2014} and~\cite{BPV-COMPUTER2015}.
We aim to apply our approach to support various application needs
for dynamic management of distributed
processes~\cite{spatialprocesses}, which may also impact the methods
of alignment for anonymous functions.
We also plan to isolate fragments of the calculus that satisfy
behavioural properties such as self-stabilisation, quasi-stabilisation
to a dynamically evolving field, or density independence, following
the approach of \cite{VD-COORD2014-LNCS2014} and \cite{VBDP-SASO2015}.
Finally, these foundations can be applied in developing
APIs enabling the simple construction of complex distributed
applications, building on the work in \cite{BV-FOCAS2014}, \cite{BPV-COMPUTER2015} and~\cite{VBDP-SASO2015} to define a
layered library of self-organisation patterns, and applying these APIs
to support a wide range of practical distributed applications.


\appendix
\section{Appendix: Proofs of the Main Theorems} \label{apx-properties}

\subsection{Device Computation Type Preservation and domain Alignment}

\begin{repthm}[Type Preservation and Domain Alignment]{the-DeviceTypePreservationAndDomainAlignment}
	Let $\TtypEnv=\overline{\xname}:\overline{\type}$, $\expTypJud{\LTStypEnv}{\emptyset}{\overline{\anyvalue}}{\overline{\type}}$, so that $\lengthOf{\overline{\anyvalue}}=\lengthOf{\overline{\xname}}$.
	If $\expTypJud{\LTStypEnv}{\TtypEnv}{\e}{\type}$, $\Trees\in\WFVTE{\LTStypEnv}{\TtypEnv}{\e}$ and $\bsopsem{\deviceId}{\Trees}{\applySubstitution{\e}{\substitution{\overline{\xname}}{\overline{\anyvalue}}}}{\vtree}$, then:
	\begin{enumerate}
		\item 
		$\vtree \in \WFVT{\LTStypEnv}{\TtypEnv}{\e}$,
		\item
		$\expTypJud{\LTStypEnv}{\emptyset}{\vrootOf{\vtree}}{\type}$, and
		\item
		if $\vrootOf{\vtree}$ is a neighbouring field value $\fvalue$ then $\domof{\fvalue}=\domof{\Trees}\cup\{\deviceId\}$. 
	\end{enumerate}
\end{repthm}

Observe that the typing rules (in Figure~\ref{fig:SurfaceTyping}) and the evaluation rules  (in Figure~\ref{fig:deviceSemantics}) are syntax directed. Then the proof can be carried out by induction on the syntax of expressions, while using the following standard lemmas.

\begin{lem}[Substitution]\label{lem-substitution}
	Let $\TtypEnv=\overline{\xname}:\overline{\type}$, $\expTypJud{\OStypEnv}{\emptyset}{\overline{\anyvalue}}{\overline{\type}}$. If $\expTypJud{\LTStypEnv}{\TtypEnv}{\e}{\type}$, then $\expTypJud{\LTStypEnv}{\emptyset}{\applySubstitution{\e}{\substitution{\overline{\xname}}{\overline{\anyvalue}}}}{\type}$.
\end{lem}
\begin{proof}
	Straightforward by induction on application of the typing rules for expressions in Figure~\ref{fig:SurfaceTyping}. 
\end{proof}

\begin{lem}[Weakening]\label{lem-weakening}
	Let $\LTStypEnv'\supseteq\LTStypEnv$, $\TtypEnv'\supseteq\TtypEnv$ be such that $\domof{\LTStypEnv'}\cap\domof{\TtypEnv'}=\emptyset$. If $\expTypJud{\LTStypEnv}{\TtypEnv}{\e}{\type}$, then $\expTypJud{\LTStypEnv'}{\TtypEnv'}{\e}{\type}$.
\end{lem}
\begin{proof}
	Straightforward by induction on application of the typing rules for expressions in Figure~\ref{fig:SurfaceTyping}. 
\end{proof}

\begin{proof}[of Theorem \ref{the-DeviceTypePreservationAndDomainAlignment}]
	We proceed proving points (1)-(3) by simultaneous induction on the syntax of expression $\e$ (given in Figure~\ref{fig:source:syntax}).
	\begin{itemize}
		\item
		$\e \; = \; \fvalue$: This case is not allowed to appear in source programs.
		
		\item
		$\e \; = \; \dcOf{\dc}{\overline{\lvalue}} \; \BNFmid \; \oname \; \BNFmid \; \fname \; \BNFmid \; (\overline{\xname}) \; \toSymK \; \e$: In this case, $\applySubstitution{\e}{\substitution{\overline{\xname}}{\overline{\anyvalue}}} = \anyvalue$ is a local value hence $\vtree = \mkvtree{}{\anyvalue}{}$ by rule \ruleNameSize{[E-LOC]}. Thus $\vtree$ is well-formed for $\e$ and $\vrootOf{\vtree} = \anyvalue$ has type $\type$ by the Substitution Lemma.
		
		\item
		$\e \; = \; \xname$: In this case, $\applySubstitution{\e}{\substitution{\overline{\xname}}{\overline{\anyvalue}}}$ is trivially a value of the correct type $\type$, hence $\vtree = \mkvtree{}{\anyvalue}{}$ is well-formed for $\e$. If $\type$ is a local type, we are done. If $\type$ is a field type, we know by induction hypothesis that $\fvalue = \anyvalue$ has domain $\domof{\Trees'} \cup \{\deviceId\}$ for some $\Trees'$ including $\Trees$ as a subtree (pointwise). Thus $\domof{\fvalue} \supseteq \domof{\Trees}\cup\{\deviceId\}$, hence we can apply rule \ruleNameSize{[E-FLD]} to obtain that $\fvalue$ has domain exactly $\domof{\Trees}\cup\{\deviceId\}$.
		
		\item
		$\e \; = \; \e_{n+1}(\overline{\e})$: In this case, either rule \ruleNameSize{[E-B-APP]} or \ruleNameSize{[E-D-APP]} applies, depending on whether $\e_{n+1}$ evaluates to a built-in operator or not. In both cases, the resulting value-tree $\vtree$ is easily checked to be well-formed for $\e$, and type preservation holds by induction hypothesis together with standard arguments on typing of function applications.
		
		If $\e$ has field type, we also need to check that $\domof{\vrootOf{\vtree}} = \domof{\Trees} \cup \{\deviceId\}$. We have two possibilities:
		\begin{itemize}
			\item $\e_{n+1}$ evaluates to a built-in operator $\oname$: In this case, domain alignment follows from the assumptions on $\builtinop{\oname}{\deviceId}{\Trees}$.
			\item $\e_{n+1}$ evaluates to a user-defined or anonymous function $\funvalue$: In this case, $\e_{n+1}$ cannot be of the form $\repK$, $\nbrK$, or $\dcOf{\dc}{\ldots}$. Furthermore, it cannot neither be of the form $\e'(\overline{\e}')$ since in that case $\e'$ would have type $(\overline{\type}') \to (\overline{\type}) \to \ftype$ which is not allowed by the type system. It follows that $\e_{n+1}$ is in fact equal to $\funvalue$, hence no alignment with neighbours is required.
		\end{itemize}
		
		\item
		$\e \; = \; \nbrK\{\e_1\}$: In this case rule \ruleNameSize{[E-NBR]} applies, thus $\vtree = \mkvtree{}{\fvalue}{\vtree_1}$ where $\fvalue = \mapupdate{\vrootOf{\piIof{1}{\Trees}}}{\envmap{\deviceId}{\vrootOf{\vtree_1}}}$. Then $\vtree$ is well-formed for $\e$, and $\fvalue$ has type $\type = \ftypeOf{\type_1}$ where $\e_1 : \type_1$. Furthermore, $\domof{\fvalue} = \domof{\Trees} \cup \{\deviceId\}$ as required.
		
		\item
		$\e \; = \; \repK(\e_1)\{(\xname) \; \toSymK \; \e_2\}$: In this case rule \ruleNameSize{[E-REP]} applies, thus $\vtree = \mkvtree{}{\lvalue_2}{\vtree_1,\vtree_2}$ is well-formed for $\e$ (where $\lvalue_i$, $\vtree_i$ follows the notation in Figure~\ref{fig:deviceSemantics}, rule \ruleNameSize{[E-REP]}). By induction hypothesis, $\lvalue_1 = \vrootOf{\vtree_1}$ has the same type as $\e_1$ which is $\type$. If $\deviceId \notin \domof{\Trees}$, $\lvalue_0 = \lvalue_1$ also has type $\type$. Otherwise, since $\Trees$ is well-formed for $\e$, $\piIof{2}{\Trees}$ is well-formed for $\e_2$ with the additional assumption that $\xname : \type$, and by induction hypothesis $\lvalue_0 = \vrootOf{\piIof{2}{\Trees}}(\deviceId)$ has the same type as $\e_2$ which is also $\type$. Then by the Substitution Lemma, $\applySubstitution{\e_2}{\substitution{\xname}{\lvalue_0}}$ also has type $\type$ hence by induction hypothesis the same does $\lvalue_2 = \vrootOf{\vtree_2}$, concluding the proof.
	\end{itemize}
\end{proof}

\subsection{Adequacy and full abstraction}

\begin{repthm}[Adequacy]{thm:Adequacy}
	Assume that the denotational environment is coherent with the operational evolution of the network and $\e$ is well-typed (that is, $\expTypJud{\LTStypEnv}{\emptyset}{\e}{\type}$).
	
	Then $\denotexp{\e}{}{\EventS}(\eventId) = \denotexp{\anyvalue_\eventId}{}{\EventS}(\eventId)$ for each $\eventId$ in $\EventS$, where $\anyvalue_\eventId = \vrootOf{\vtree_\eventId}$ is the operational outcome of fire $\eventId$.
\end{repthm}

In order to carry on the induction on the structure of $\e$, we shall prove the following strengthened version instead.

\begin{lem}
	Assume that the denotational environment is coherent with the operational evolution of the network and $\e$ is a well-typed expression with free variables $\overline\xname$ (that is, $\expTypJud{\LTStypEnv}{\overline\xname : \overline\type}{\e}{\type}$).
	Let $\VarS = \envmap{\overline\xname}{\overline\dvalue}$ and $\overline\anyvaluealt^\eventId$ be such that $\dvalue_i(\eventId) = \denotexp{\anyvaluealt_i^\eventId}{}{}(\eventId)$ for all $i$ and events $\eventId$.

	Then $\denotexp{\e}{\VarS}{}(\eventId) = \denotexp{\anyvalue^\eventId}{}{}(\eventId)$ for each $\eventId$ in $\EventS$, where $\anyvalue^\eventId = \vrootOf{\vtree_\eventId}$ is the operational outcome of fire $\eventId$ evaluating $\applySubstitution{\e}{\substitution{\overline\xname}{\overline\anyvaluealt^\eventId}}$.
\end{lem}
\begin{proof}
	First, recall that coherence of denotational and operational environments implies that for each fire $\eventId$, $\Trees_\eventId = \{ \envmap{\deviceId_{\eventId'}}{\vtree_{\eventId'}}: ~ \neighbour{\eventId}{\eventId'} \}$.
	We prove the assertion simultaneously for all possible (pairs of coherent) environments, set of assumptions and event $\eventId$, by induction on the structure of $\e$.
	\begin{itemize}
		\item
		$\e \; = \; \xname_i$: In this case, $\applySubstitution{\e}{\substitution{\overline\xname}{\overline\anyvaluealt^\eventId}} = \anyvaluealt_i^\eventId = \anyvalue^\eventId$ and by hypothesis $\dvalue_i(\eventId)  = \denotexp{\anyvaluealt_i^\eventId}{}{}(\eventId) = \denotexp{\anyvalue^\eventId}{}{}(\eventId)$.
		
		\item
		$\e \; = \; \dcOf{\dc}{\overline{\lvalue}} \; \BNFmid \; \fvalue \; \BNFmid \; \oname \; \BNFmid \; \fname \; \BNFmid \; (\overline{\xname}) \; \toSymK \; \e$: Since $\e$ is already a value, $\anyvalue^\eventId = \e$ and the thesis follows.
		
		\item
		$\e \; = \; \nbrK\{\e_1\}$: By inductive hypothesis $\denotexp{\e_1}{\VarS}{}(\eventId) = \denotexp{\anyvalue_1^\eventId}{}{}(\eventId)$ for each $\eventId$ in $\EventS$, where $\anyvalue_1^\eventId = \vrootOf{\vtree_1^\eventId}$ is the (operational) outcome of $\applySubstitution{\e_1}{\substitution{\overline\xname}{\overline\anyvaluealt^\eventId}}$ in fire $\eventId$. By rule \ruleNameSize{[E-NBR]}, we have that
		\[
		\anyvalue^\eventId = \{ \envmap{\deviceId_{\eventId'}}{\anyvalue_1^{\eventId'}}: ~ (\eventId' = \eventId) \vee (\neighbour{\eventId}{\eventId'} \wedge \deviceId_{\eventId'} \neq \deviceId_\eventId) \}
		\]
		and
		\[
		\denotexp{\anyvalue^\eventId}{}{}(\eventId) = \{ \envmap{\deviceId_{\eventId'}}{\denotexp{\anyvalue_1^{\eventId'}}{}{}(\eventId)}: ~ (\eventId' = \eventId) \vee (\neighbour{\eventId}{\eventId'} \wedge \deviceId_{\eventId'} \neq \deviceId_\eventId) \}.
		\]
		On the other hand, $\denotexp{\e}{\VarS}{}(\eventId)$ is
		\begin{align*}
			\denotexp{\nbrK\{\e_1\}}{\VarS}{}(\eventId)
			&= \denotf{\deviceId \in \predevices{\eventS}(\eventId)}{} \denotexp{\e_1}{\VarS}{}(\nbrdevice{\eventId}{\deviceId})
			\\
			&= \denotf{\deviceId \in \predevices{\eventS}(\eventId)}{} \denotexp{\anyvalue_1^{\nbrdevice{\eventId}{\deviceId}}}{}{}(\nbrdevice{\eventId}{\deviceId})
			\\
			&= \{ \envmap{\deviceId}{\denotexp{\anyvalue_1^{\nbrdevice{\eventId}{\deviceId}}}{}{}(\eventId)}: ~ (\nbrdevice{\eventId}{\deviceId} = \eventId) \vee (\neighbour{\eventId}{\nbrdevice{\eventId}{\deviceId}} \wedge \deviceId \neq \deviceId_\eventId) \}
			= \denotexp{\anyvalue^\eventId}{}{}(\eventId).
		\end{align*}
		
		\item
		$\e \; = \; \repK(\e_1)\{(\yname) \toSymK \e_2\}$: Recall that by rule \ruleNameSize{[E-REP]}, $\anyvalue^\eventId$ is the evaluation of $\applySubstitution{\e_2}{\substitution{\overline\xname}{\overline\anyvaluealt^\eventId}, \substitution{\yname}{\anyvalue^{\eventId^-}}}$ if $\eventId^-$ exists, $\applySubstitution{\e_2}{\substitution{\overline\xname}{\overline\anyvaluealt^\eventId}, \substitution{\yname}{\anyvalue_1^\eventId}}$ otherwise. We prove by induction on $n$ that if $\eventId$ has at most $n$ predecessors (i.e. $\repdevice{\eventId}$ can be applied at most $n$ times) then $\denotexp{\anyvalue^\eventId}{}{}(\eventId) = \builtindenot{R}{\e}_{n+1}(\eventId)$. The thesis will then follow for $n$ greater than the number of events.
		
		If $n = 0$, by induction hypothesis on both $\e_1$ and $\applySubstitution{\e_2}{\substitution{\overline\xname}{\overline\anyvaluealt^\eventId}, \substitution{\yname}{\anyvalue_1^\eventId}}$:
		\begin{align*}
			\denotexp{\anyvalue^\eventId}{}{}(\eventId)
			= \denotexp{\e_2}{\VarS \cup \envmap{\yname}{\denotexp{\e_1}{\VarS}{}}}{}(\eventId)
			= \builtindenot{R}{\e}_1(\eventId).
		\end{align*}
		
		If $n > 0$, $\anyvalue^{\eventId^-}$ has at most $n-1$ predecessors thus we can apply the inductive hypothesis on $n-1$ and $\applySubstitution{\e_2}{\substitution{\overline\xname}{\overline\anyvaluealt^\eventId}, \substitution{\yname}{\anyvalue^{\eventId^-}}}$ to obtain:
		\begin{align*}
			\denotexp{\anyvalue^\eventId}{}{}(\eventId)
			= \denotexp{\e_2}{\VarS \cup \envmap{\yname}{\shift{\builtindenot{R}{\e}_n}}}{}(\eventId)
			= \builtindenot{R}{\e}_{n+1}(\eventId).
		\end{align*}
		
		\item
		$\e \; = \; \e_{n+1}(\overline{\e})$: Suppose that $\overline{\e}$ evaluates to $\overline{\anyvalue}^\eventId$ and $\e_{n+1}$ to $\funvalue^\eventId$ in event $\eventId$. Notice that by induction hypothesis, $\denotexp{\overline\anyvalue^\eventId}{}{}(\eventId) = \denotexp{\overline\e}{\VarS}{}(\eventId)$ and $\denotexp{\funvalue^\eventId}{}{}(\eventId) = \denotexp{\e_{n+1}}{\VarS}{}(\eventId)$.
		
		If $\funvalue^\eventId$ is a built-in operator $\oname$, the thesis follows from coherence of $\builtindenot{B}{\oname}$ with $\builtinop{\oname}{\deviceId}{\Trees}$, together with the induction hypothesis on $\e_1, \ldots, \e_{n+1}$.
		
		If $\funvalue^\eventId$ is an anonymous function, by rule \ruleNameSize{[E-D-APP]} $\funvalue^\eventId(\overline{\anyvalue}^\eventId)$ (hence $\applySubstitution{\e}{\substitution{\overline\xname}{\overline\anyvaluealt^\eventId}}$) evaluates in $\eventId$ to the result $\anyvalue^\eventId$ of $\applySubstitution{\body{\funvalue^\eventId}}{\substitution{\args{\funvalue^\eventId}}{\overline{\anyvalue}^\eventId}}$ as calculated in $\Trees'_\eventId = \piBof{\funvalue_\eventId}{\Trees_\eventId}$. The value-tree environments $\{\Trees'_{\eventId'} : ~ \funvalue^{\eventId'} = \funvalue^\eventId\}$ together define an operational environment $\Envi_{\funvalue^\eventId}$ consisting only of those devices and events which agree on the evaluation of $\e_{n+1}$. In fact, such restricted environment is coherent with the restricted denotational environment $\EventS(\e_{n+1}, \eventId)$:
		\begin{align*}
			\EventS(\e_{n+1}, \eventId)
			& = \{\eventId': ~ \denotexp{\e_{n+1}}{\VarS}{}(\eventId') = \denotexp{\e_{n+1}}{\VarS}{}(\eventId) \} \\
			& = \{\eventId': ~ \denotexp{\funvalue^{\eventId'}}{}{}(\eventId') = \denotexp{\funvalue^\eventId}{}{}(\eventId) \} \\
			& = \{\eventId': ~ \denotexp{\funvalue^{\eventId'}}{}{} = \denotexp{\funvalue^\eventId}{}{} \}
			= \{\eventId': ~ \funvalue^{\eventId'} = \funvalue^\eventId \}
		\end{align*}
		where we used the inductive hypothesis and the facts that denotations of values are constant field evolutions and that function denotations coincide if and only if the functions are syntactically equal (in order).
		
		Thus we can apply the induction hypothesis on $\overline\e$ in $\EventS$ and on $\applySubstitution{\body{\funvalue^\eventId}}{\substitution{\args{\funvalue^\eventId}}{\overline{\anyvalue}^\eventId}}$ in $\EventS(\e_{n+1}, \eventId)$ to obtain:
		\begin{align*}
			\denotexp{\anyvalue^\eventId}{}{}(\eventId)
			& = \denotexp{\body{\funvalue^\eventId}}{\envmap{\args{\funvalue^\eventId}}{\proj{\denotexp{\overline\e}{\VarS}{\EventS}}{\EventS(\e_{n+1}, \eventId)}}}{\EventS(\e_{n+1}, \eventId)}(\eventId)
		\end{align*}
		On the other hand,
		\begin{align*}
			\denotexp{\e}{\VarS}{}(\eventId)
			& = \sndK\left(\denotexp{\e_{n+1}}{\VarS}{}(\eventId)\right) \left( \proj{\denotexp{\overline\e}{\VarS}{}}{\EventS(\e_{n+1}, \eventId)} \right) (\eventId)
			\\
			& = \sndK\left(\denotexp{\funvalue^\eventId}{}{}(\eventId)\right) \left( \proj{\denotexp{\overline\e}{\VarS}{}}{\EventS(\e_{n+1}, \eventId)} \right) (\eventId)
			\\
			& = \left( \denotf{\overline\dvalue}{} \denotexp{\body{\funvalue^\eventId}}{\envmap{\args{\funvalue^\eventId}}{\overline\dvalue} }{\domof{\overline\dvalue}} \right) \left( \proj{\denotexp{\overline\e}{\VarS}{}}{\EventS(\e_{n+1}, \eventId)} \right) (\eventId)
			\\
			& = \denotexp{\body{\funvalue^\eventId}}{\envmap{\args{\funvalue^\eventId}}{\proj{\denotexp{\overline\e}{\VarS}{\EventS}}{\EventS(\e_{n+1}, \eventId)}} }{\EventS(\e_{n+1}, \eventId)} (\eventId) = \denotexp{\anyvalue^\eventId}{}{}(\eventId)
		\end{align*}
		completing the proof in this case.
		
		If $\funvalue^\eventId$ is a user-defined function, we prove by further induction on $n$ that $\denotexp{\anyvalue^\eventId}{}{}(\eventId)$ is equal to $V_n = \denotexp{\body{\funvalue^\eventId}}{\envmap{\args{\funvalue^\eventId}}{\proj{\denotexp{\overline\e}{\VarS}{\EventS}}{\EventS(\e_{n+1}, \eventId)}}, \envmap{\funvalue^\eventId}{\builtindenot{D}{\funvalue^\eventId}_n}}{\EventS(\e_{n+1}, \eventId)} (\eventId)$ whenever the recursive depth of the function call is bounded by $n$. The thesis will then follow by passing to the limit (and expanding the denotation of $\e$ as for anonymous functions).
		
		If $n = 0$, the evaluation of $\applySubstitution{\body{\funvalue^\eventId}}{\substitution{\args{\funvalue^\eventId}}{\overline{\anyvalue}^\eventId}}$ does not involve further evaluation of the defined function $\funvalue^\eventId$. This also holds on the denotational side, giving that
		\[
		V_0 = \denotexp{\body{\funvalue^\eventId}}{\envmap{\args{\funvalue^\eventId}}{\proj{\denotexp{\overline\e}{\VarS}{\EventS}}{\EventS(\e_{n+1}, \eventId)}}}{\EventS(\e_{n+1}, \eventId)} (\eventId)
		= \denotexp{\anyvalue^\eventId}{}{}(\eventId)
		\]
		If instead $n > 0$, the evaluation of $\applySubstitution{\body{\funvalue^\eventId}}{\substitution{\args{\funvalue^\eventId}}{\overline{\anyvalue}^\eventId}}$ involves evaluation of $\funvalue^\eventId$ with recursive depth bounded by $n-1$. Thus by inductive hypothesis we can denote each such call with $\builtindenot{D}{\funvalue^\eventId}_{n-1}$ and the thesis follows.
	\end{itemize}
\end{proof}

\begin{repthm}[Full Abstraction]{thm:FullAbstraction}
	Suppose that constructors for built-in local types are faithful. Then the full abstraction property holds.
\end{repthm}

\begin{proof}
	First notice that given any two values $\anyvalue_1$, $\anyvalue_2$ of the same type $\denotexp{\anyvalue_1}{}{\eventS} = \denotexp{\anyvalue_2}{}{\eventS}$ if and only if $\anyvalue_1 = \anyvalue_2$. If $\anyvalue_1$, $\anyvalue_2$ are functions, it holds since the denotation includes the function tag (i.e. the syntactic expression itself). If $\anyvalue_1$, $\anyvalue_2$ are constructor expressions, it holds by hypothesis. If $\anyvalue_1$, $\anyvalue_2$ are neighbouring field values, it holds since $\denotexp{\envmap{\overline\deviceId}{\overline\lvalue}}{}{\eventS} = \denotf{\eventId}{} \envmap{\overline\deviceId}{\denotexp{\overline\lvalue}{}{\eventS}(\eventId)}$ and we already proved the equivalence for local values. From this equivalence and adequacy we can prove the left-to-right implication of the full abstraction property as follows.
	
	Suppose that $\e_1$ and $\e_2$ have the same denotation in every denotational environment and fix an operational environment $\Envi$, in which $\e_1$, $\e_2$ evaluate to $\anyvalue^1_\eventId$, $\anyvalue^2_\eventId$ in fire $\eventId$. Then given a denotational environment $\eventS$ that is coherent with $\Envi$, we have that $\denotexp{\e_1}{}{\eventS} = \denotexp{\e_2}{}{\eventS}$. By adequacy, it follows that $\denotexp{\anyvalue^1_\eventId}{}{\eventS}(\eventId) = \denotexp{\anyvalue^2_\eventId}{}{\eventS}(\eventId)$ hence $\anyvalue^1_\eventId = \anyvalue^2_\eventId$ for all $\eventId$. Since $\e_1$ and $\e_2$ evaluate to the same value in every fire $\eventId$ they cannot be distinguished in any context in $\Envi$. Since $\Envi$ was arbitrary, this concludes the left-to-right part of the proof.
	
	For the converse implication, suppose that there exists a denotational environment $\eventS$ such that $\e_1$ and $\e_2$ have denotations which differs in $\eventId$. Let $\Envi$ be an operational environment coherent with $\eventS$, in which $\e_1$, $\e_2$ evaluate to $\anyvalue_1$, $\anyvalue_2$ in fire $\eventId$. By adequacy,
	\[
	\denotexp{\anyvalue_1}{}{\eventS}(\eventId) = \denotexp{\e_1}{}{\eventS}(\eventId) \neq \denotexp{\e_2}{}{\eventS}(\eventId) = \denotexp{\anyvalue_2}{}{\eventS}(\eventId)
	\]
	hence $\anyvalue_1 \neq \anyvalue_2$. Thus $\e_1$, $\e_2$  are distinguished by context $\e' \BNFcce (\e_1 = \xname)$ of type $\btype$.

\end{proof}




\bibliographystyle{plain}

\end{document}